\documentclass[12pt]{article}

\usepackage{graphicx}
\usepackage{amsfonts}
\usepackage{amssymb}
\usepackage{amsmath}
\usepackage{latexsym}
\usepackage{url}
\usepackage{cite}
\usepackage{psfrag}

\newtheorem{theorem}{Theorem}[section]

\newtheorem{lemma}[theorem]{Lemma}

\newtheorem{remark}[theorem]{Remark}
\newenvironment{proof}[1][Proof]{\textsc{#1.\ }} {\ \rule{0.5em}{0.5em}}
\numberwithin{equation}{section}
\numberwithin{figure}{section}

\input{tcilatex}

\begin{document}

\title{The Cauchy problem on a characteristic cone for the Einstein equations
in arbitrary dimensions
}
\author{
Yvonne Choquet-Bruhat \\
Acad\'emie des Sciences, Paris
\and
Piotr T. Chru\'{s}ciel \\
Universit\"at Wien
\and
Jos\'e M. Mart\'in-Garc\'ia \\
Institut d'Astrophysique de Paris, and \\
Laboratoire Univers et Th\'eories, Meudon
}

\maketitle

\begin{abstract}
We derive
explicit formulae for a set of constraints for the Einstein
equations on a null hypersurface, in arbitrary dimensions. We
solve these constraints and show that they provide necessary
and sufficient conditions so that a spacetime solution of the
Cauchy problem on a characteristic cone for the hyperbolic
system of the reduced Einstein equations in wave-map gauge also
satisfies the full Einstein equations. We prove a geometric
uniqueness theorem for this Cauchy problem in the vacuum case.
\end{abstract}

\tableofcontents
\section{Introduction}
 \label{S25IV10.0}

An Einsteinian spacetime is a pair $(V,g)$ with $V$ a manifold
and $g$ a Lorentzian metric solution on $V$ of Einstein's
equations. These are a geometric system of quasilinear second
order partial differential equations which express, in vacuum,
the vanishing of the Ricci tensor of $g$. Two isometric
spacetimes are considered as the same. A fundamental problem is
the determination of generic solutions from initial data, i.e.
the Cauchy problem. It is a badly posed problem from the
analyst point of view in arbitrary local coordinates on $V$
since the characteristic determinant of the PDE system is
identically zero. It is well known that if initial data are
given on a spacelike manifold $M$ the problem splits into
constraints to be satisfied by the geometric initial data, i.e.
the induced metric and second fundamental form of $M$ as
submanifold of $V,$ and an evolution system. The latter must be
well posed and must exhibit propagation of the gravitational
field with the speed of light, that is, must be a hyperbolic
system with characteristic cones the null cones of the
looked-for spacetime metric. The simplest way to obtain such a
system is to suppose that the coordinates satisfy so-called
harmonicity conditions, or, more generally, to introduce a
preassigned metric $\hat{g},$ called target metric, which
permits to write the Ricci tensor as the sum of two tensorial
operators, one of which is a hyperbolic operator acting on $g$,
called the reduced Ricci tensor, and the other a homogeneous
first order differential operator acting on a vector $H,$
called wave-map gauge vector, which vanishes when the identity
map is a wave map from $(V,g)$ onto $(V,\hat{g})$. When $M$ is
spacelike, classical theorems of analysis show existence and
uniqueness of solutions of so-reduced Einstein equations and,
due to the Bianchi identities, of the original Einstein
equations if the initial data satisfy the constraints. The case
where the initial manifold is null has analogies with the
spacelike case but also important differences: First, the
induced metric is degenerate,   and unconstrained in the
regions where $\tau$ (as defined below) has no zeroes. Next,
the second fundamental forms defined on a spacelike and on a
null manifold, for which the normal is also tangent, have very
different properties. Finally, null initial data on a light
cone, or on two-intersecting null hypersurfaces, determine the
solution in one time direction only, past or future.

A complete understanding of this problem is still lacking, even
in space-time dimension four. The most exhaustive studies are
for the case of two intersecting null surfaces
\cite{CagnacEinsteinCRAS1,CagnacEinsteinCRAS2,SachsCIVP,Dautcourt,
PenroseCIVP,ChristodoulouMzH,F1,F2,DamourSchmidt,
MzHSeifertCIVP}; compare, in different settings,
\cite{CaciottaNicoloI,CaciottaNicoloII,HaywardNullSurfaceEquations}.
The most complete construction of equations satisfied by
initial data has been given by Damour and
Schmidt~\cite{DamourSchmidt}, and the most satisfactory
treatment of the local evolution by Rendall~\cite{RendallCIVP}.
The problem with data on a characteristic cone  presents new
mathematical difficulties due to its singularity at the vertex,
and only partial results have been obtained before
in~\cite{YvonneCIVP,DossaAHP,F1,FriedrichCMP86,RendallCIVP2}.

The object of this work is to present a treatment of the
Einstein equations with data on a characteristic cone in all
dimensions $n+1\geq 3$. We proceed as follows:

Though the equations are geometric and the final results
coordinate independent, it is useful to introduce adapted
coordinates to carry-out the analysis. We take a $C^{\infty}$
manifold $V$ diffeomorphic to $\mathbf{R}^{n+1}$, and we consider
a cone $C_{O}$ in $V$ with vertex $O\in V$ and equation, in
coordinates $y^{\alpha}$ compatible with the $C^{\infty}$
structure of $V$
\begin{equation*}
y^{0}=r,\qquad r:=\{\sum_{i=1,\ldots,n}(y^{i})^{2}\}^{\frac{1}{2}}\;.
\end{equation*}
We consider the Cauchy problem with data on $C_{O}$ for the
Einstein equations with unknown a Lorentzian metric $g$,
assuming that $C_{O}$ will
be a characteristic cone of the metric $g$ and the lines $y^{0}=r,$
$\frac{y^{i}}{r}=c^{i},$ where the $c^{i}$ are constants, its null
rays. It is well known\footnote{%
It is known, using normal coordinates, that a null cone in a
$C^{1,1}$, Lorentzian spacetime always admits such a
representation in a neighbourhood of its vertex; the null
geodesics which generate it are represented by the lines where
$r$ only varies. \label{normalcone}\label{f19V10.1} This is
guaranteed to hold only in a neighbourhood of the vertex, as
there can be caustics.}
that the
characteristic cone of a $C^{1,1}$ Lorentzian metric admits
always such a representation in a neighbourhood of its vertex.
We review in section 3 an existence theorem which applies to
the reduced Einstein equations in wave-map gauge with Minkowski
target reading in these coordinates
\begin{equation}
\hat{g}=-(dy^{0})^{2}+\sum_{i=1,\ldots,n}(dy^{i})^{2}.
 \label{19V10.2}
\end{equation}
We introduce in section~\ref{nullhypersurfaces} what we call
adapted null coordinates, singular on the line $r=0,$ in
particular at the vertex $O$ of $C_{O},$ but $C^{\infty}$
elsewhere, by setting
\begin{equation*}
x^{0}:=r-y^{0},\qquad x^{1}:=r,
\end{equation*}
and defining $x^{A},$ $A=2,...n,$ to be local coordinates on the sphere
$S^{n-1}$. In coordinates $x^{\alpha}$ the trace $\overline{g}$ on
$C_{O}$ of the metric $g$ we are looking for has the form
\begin{equation*}
\overline{g}=\overline{g}_{00}(dx^{0})^{2}+2\nu_{0}dx^{0}dx^{1}+2\nu_{A}dx^{0}dx^{A}+
\overline{g}_{AB}dx^{A}dx^{B}.
\end{equation*}

Remark that the question, whether or not
$x^{1}$ is  an affine parameter on the null rays
$x^{A}=c^{A}$, depends on derivatives tranversal to $C_{O}$ of
the spacetime metric $g$, which are usually not considered as
part of the initial data for characteristic Cauchy problems.

The adapted null, but singular at the vertex, coordinates
$x^{\alpha}$ are used to solve ``wave-map gauge constraints''
satisfied by $\overline{g}$.

In section~\ref{ConstraintsAndGauge} we review the standard
argument, that the Bianchi identities imply that if $g$
satisfies the reduced Einstein equations with source a
divergence-free stress energy tensor, then the vector $H$
satisfies a homogeneous hyperbolic system; it vanishes in the
future of $C_{O}$ if its trace $\overline{H}$ vanishes on
$C_{O}$ .

We show in sections~\ref{constraintC1} to~\ref{solutionC0} that
$\overline{H}=0$ if the initial data $\overline{g}$ satisfies a
set of $n+1$ equations which we call the wave-map gauge
constraints. These constraints read as a hierarchical system of
ordinary differential equations along the light rays, singular
at the vertex $O,$ if one uses the adapted null coordinates
$x^{\alpha}$. We write this complete system for a general
$\hat{g}$ and generalized wave gauge, in arbitrary dimensions
$n+1\geq 3$. We integrate them successively  under natural
limit conditions on the unknowns at $O$. We study briefly in
section~\ref{ss5V.1} the case when the degenerate metric
$\tilde{g}$ induced on $C_{O}$ (i.e. the $x^{1}$ dependent
quadratic $\overline{g}_{AB}$) is prescribed.

In section~\ref{solutionC1} in order to have an evolutionary
equation for $\tau$ we set, as many authors before us,
$\overline{g}_{AB}=$ $\Omega^{2}\gamma_{AB},$ with $\gamma$ an
arbitrarily given $x^{1}$ dependent metric on $S^{n-1}$. The
first wave-map gauge constraint can be written in a form which
involves the two unknowns $\nu_{0}$ and $\Omega$. Its general
solution is obtained by the introduction of an arbitrary
function $\kappa$. We study in particular the case $\kappa=0$
which leads to the Raychaudhuri equation for $\tau$ for which
we prove global existence for a small $|\sigma|$ which depends
only on the given $\gamma$. A simple integration determines
then $\Omega$, hence $\overline{g}_{AB}$ and we are back to the
equations for $\nu_{0},\nu_{A}, \overline{g}_{00}$ with given
$\tilde{g}$. We remark that the equation for $\nu_{0}$ (for
$\kappa=0)$ implies that the vector $\ell$ is parallelly
transported  along the null ray by the connection of a
spacetime metric in wave-map gauge satisfying the Einstein
equations. In sections~\ref{constraintCA},~\ref{solutionCA},
\ref{constraintC0} and~\ref{solutionC0} we establish, and
integrate, the other constraints determining $\nu_{A}$ and
$\overline{g}_{00}$. A theorem in section~\ref{summary}
summarizes our analysis of the wave-gauge constraint equations.
A uniqueness theorem is proved in Section~\ref{slgu}.

A major question left open by our work is the description of
the largest class of unconstrained initial data which lead to
solutions of the wave-map gauge constraints such that the
components in $y^{\alpha}$ coordinates of the trace $\overline{g}$
satisfy the (non trivial) initial conditions given in section
\ref{ssCagDos} for the existence theorem for quasilinear wave
equations. The problem is that the wave-gauge constraint
equations determine the components of $\overline{g}$ in the
$x^{\alpha}$ coordinates, and these components are linked with
the components in the $y^{\alpha} $ coordinates by linear
relations which are singular at the vertex. We simply note here
that initial data which are Minkowskian in a neighbourhood of
the vertex are easily seen to be in the class where the
existence theorem holds; see also~\cite{CCM3} for more general
family of data. We plan to return to this problem in a near
future.

\section{Definitions}
 \label{S25IV10.1}

\subsection{Ricci tensor and harmonicity functions}
 \label{sS25IV10.1}

The Ricci tensor of any pseudo Riemannian metric is given in local
coordinates by,
\begin{equation}
R_{\alpha\beta} := \partial_{\lambda}\Gamma_{\alpha\beta}^{\lambda}-%
\partial_{\alpha}\Gamma_{\beta\lambda}^{\lambda}+\Gamma_{\alpha\beta
}^{\lambda}\Gamma_{\lambda\mu}^{\mu}-\Gamma_{\alpha\mu}^{\lambda}\Gamma
_{\beta\lambda}^{\mu},\quad  \partial_{\lambda}:=\frac{\partial}{%
\partial x^{\lambda}},
\label{Ricci}
\end{equation}
with $\Gamma_{\alpha\beta}^{\lambda}$ the Christoffel symbols
\begin{equation}
\Gamma_{\alpha\beta}^{\lambda}:=g^{\lambda\mu}[\mu,\alpha\beta],\quad  [%
\mu,\alpha\beta]:=\frac{1}{2}(\partial_{\alpha}g_{\beta\mu}+\partial_{%
\beta}g_{\alpha\mu}-\partial_{\mu}g_{\alpha\beta}).
\label{Christoffel}
\end{equation}
The Ricci tensor satisfies the identity
\begin{equation}
R_{\alpha\beta}\equiv R_{\alpha\beta}^{(h)}+{\frac{1}{2}}(g_{\alpha\lambda
}\partial_{\beta}\Gamma^{\lambda}+g_{\beta\lambda}\partial_{\alpha}\Gamma^{%
\lambda}) ,
\label{RiccihIdentity}
\end{equation}
where $Ricc^{(h)}(g)$, the reduced Ricci tensor, is a
quasi-linear, quasi-diagonal operator on the components of $g$,
\begin{equation}
 R_{\alpha\beta}^{(h)}\equiv-{\frac{1}{2}}g^{\lambda\mu}
 \partial_{\lambda}\partial_{\mu}g_{\alpha\beta}+f[g,\partial g]_{\alpha\beta}
 \;,
\label{Riccih}
\end{equation}
and $f[g,\partial g]_{\alpha\beta}$ is a quadratic form in the first
derivatives $\partial g$ of $g$
with coefficients polynomial in $g$ and its contravariant associate.

The $\Gamma^{\lambda}$'s, called \emph{harmonicity functions},
are defined as
\begin{equation}
 \Gamma^{\alpha}:=g^{\lambda\mu}\Gamma_{\lambda\mu}^{\alpha}
 \;.
\label{Harmonicity}
\end{equation}
The condition $\Gamma^{\alpha}=0$ expresses that the coordinate
function $x^{\alpha}$ satisfies the wave equation in the metric
$g$.

\subsection{Wave-map gauge}
 \label{ssWmg}

The harmonicity functions are coordinate dependent and only
defined locally in general, whether in space,  or time, or
both. The wave-map gauge, which we are about to define,
provides conditions which are tensorial.
A metric $g$ on a manifold $V$ will be said to be \emph{in
$\hat{g}$-wave-map gauge} if the identity map $V\rightarrow V$
is a harmonic diffeomorphism from the spacetime $(V,g)$ onto
the pseudo-Riemannian manifold ($V,\hat{g})$, with $\hat{g}$ a
given metric on $V$. Recall that a mapping $f:(V,g)\rightarrow
(V,\hat{g})$ is a harmonic map if it satisfies the equation, in
abstract index notation,
\begin{equation}
\hat{\square}f^{\alpha} :=g^{\lambda \mu} (\partial_{\lambda \mu
}^{2}f^{\alpha} -\Gamma_{\lambda \mu}^{\sigma} \partial_{\sigma
}f^{\alpha} +\partial_{\lambda} f^{\sigma} \partial_{\mu} f^{\rho}
\hat{\Gamma}_{\sigma \rho}^{\alpha} )=0\;.
\label{WaveMap}
\end{equation}
In a subset in which $f$ is the identity map defined by
$f^{\alpha}(x)=x^{\alpha}$, the above equation reduces to
$H=0$, where the  \emph{wave-gauge vector} $H$ is given in
arbitrary coordinates by the formula
\begin{equation}
H^{\lambda} :=
g^{\alpha\beta} \Gamma_{\alpha\beta}^{\lambda}-W^{\lambda}
\;, \ \text{with}\
W^{\lambda}:=
g^{\alpha\beta} \hat{\Gamma}_{\alpha \beta}^{\lambda}
\;,
\label{WaveGauge0}
\end{equation}
where $\hat{\Gamma}^{\lambda}_{\alpha\beta}$ are the
Christoffel symbols of the \emph{target} metric $\hat{g}$.
See~\cite{YvonneBook} for a more complete discussion of the
concepts and results in this section.

The following identity has been proved to hold, with
$\hat{D}$ the Riemannian covariant derivative in the metric
$\hat{g}$~\cite[page 163]{YvonneBook},
\begin{equation}
R_{\alpha\beta}\equiv R_{\alpha\beta}^{(H)}+{\frac{1}{2}}(g_{\alpha\lambda
}\hat{D}_{\beta}H^{\lambda}+g_{\beta\lambda}\hat{D}_{\alpha}H^{\lambda
}) ,
\label{RicciHIdentity}
\end{equation}
where $R_{\alpha\beta}^{(H)}(g)$, called the reduced Ricci
tensor of the metric $g$ in $\hat{g}$-wave-map gauge, is a
quasi-linear, quasi-diagonal operator on $g$, tensor-valued,
depending on $\hat{g}$:
\begin{equation}
R_{\alpha\beta}^{(H)}\equiv-{\frac{1}{2}}g^{\lambda\mu}
{\hat{D}}_{\lambda} {\hat{D}}_{\mu}g_{\alpha\beta}+
\hat{f}[g,{\hat{D}}g]_{\alpha\beta}\;,
\label{RicciH}
\end{equation}
where $\hat{f}[g,{\hat{D}}g]_{\alpha\beta}$, independent of the
second derivatives of $g$,  is a tensor quadratic in $\hat Dg$
with coefficients  depending upon  $g$ and  $\hat{g}$, of the
form (see formula (7.7) in chapter VI of~\cite{YvonneBook})
\begin{equation}
P_{\alpha\beta}^{\rho\sigma\gamma\delta\lambda\mu}(g)\hat{D}_{\rho}%
g_{\gamma\delta}\hat{D}_{\sigma}g_{\lambda\mu}+{\frac{1}{2}}g^{\lambda\mu
}\{g_{\alpha\rho}\hat{R}_{\lambda}{}^{\rho}{}_{\beta\mu}+g_{\beta\rho}\hat
{R}_{\lambda}{}^{\rho}{}_{\alpha\mu}\}\;,
\label{fPolynomial}
\end{equation}
with $\hat{R}$ the Riemann curvature tensor of the covariant
derivative $\hat{D}$. We will frequently restrict ourselves to
the case in which the target metric is the Minkowski metric
$\eta$, and then denote by $D$  the covariant derivative. In
this case, and if using coordinates such that  the Minkowski
metric takes the canonical form \eqref{19V10.2}, the reduced
Ricci tensor in wave-map gauge coincides with the one in
harmonic coordinates.

We emphasise that, unless explicitly stated, our computations
are valid for a general $\hat {g}$.

Our main results below assume that $W$ takes the form
(\ref{WaveGauge0}). However, several results apply to a large
class of $W$'s of the
form\footnote{H.Friedrich~\cite{FriedrichCMP} introduced
generalized harmonic coordinates by adding
arbitrary functions to the harmonicity conditions.}
\begin{equation}
W^{\lambda} := g^{\alpha\beta}\hat{\Gamma}^{\lambda}_{\alpha\beta}
+\hat{W}^{\lambda}
\;,
\end{equation}
where $\hat{W}$ is a vector which may depend
upon $g$, $\hat{g}$ and possibly some other fields, but not
upon the derivatives of $g$; the relevant restrictions are
pointed out in (\ref{6VI.21new})-(\ref{6VI.22new}),
(\ref{8VI.21})-(\ref{8VI.22}) and
(\ref{8VI.23})-(\ref{8VI.24}). The reduced Ricci tensor becomes
then
\begin{equation}
R^{(H,\hat{W})}_{\alpha\beta} \equiv
R^{(H)}_{\alpha\beta} +
\frac{1}{2}\left(
g_{\alpha\lambda} \hat{D}_{\beta} \hat{W}^{\lambda} +
g_{\beta\lambda} \hat{D}_{\alpha} \hat{W}^{\lambda} \right)
\; .
\end{equation}
However, unless explicitly indicated otherwise we assume
that $\hat W$ is identically zero.

Another interesting generalization (see, e.g.,~\cite{Pretorius}
and references therein) has been inspired by numerical
simulations: if one uses the decomposition
(\ref{RicciHIdentity}), the identities $H^{\lambda}\equiv 0$ are
only obeyed to some finite precision and $H^{\lambda}$ shows a
generic tendency to deviate from zero. Attempts to cure that
have been made by introducing \emph{constraint damping}
terms~\cite{ConstraintDamping}, changing the choice of the
reduced Ricci tensor $R^{(H)}_{\alpha\beta}$ to
\begin{equation}
R^{(H)}_{\alpha\beta} + \frac{1}{2}\epsilon
(n_{\alpha}g_{\beta\lambda}
+n_{\beta}g_{\alpha\lambda}
-\frac{2\rho}{n-1}n_{\lambda}g_{\alpha\beta}) H^{\lambda} \; ,
\label{DampingRicci}
\end{equation}
or, equivalently, the reduced Einstein tensor
$S^{(H)}_{\alpha\beta}$ to
\begin{equation}
S^{(H)}_{\alpha\beta} + \frac{1}{2}\epsilon
(n_{\alpha}g_{\beta\lambda}
+n_{\beta}g_{\alpha\lambda}
+(\rho-1)n_{\lambda}g_{\alpha\beta}
) H^{\lambda} \; ,
\label{DampingEinstein}
\end{equation}
where $n^{\mu}$ is a vector field and $\epsilon$ is a small
positive constant which controls the rate of damping of the
gauge conditions. (As shown in~\cite{ConstraintDamping} the
constant $\rho$ must also be positive to have damping.) We will
show that the damping terms are consistent with our analysis.
For definiteness we will assume that $n^{\mu}$ has been
prescribed, though certain more general situations can easily
be incorporated into our scheme.

\section{Characteristic Cauchy problem}
 \label{SCagnacDossa}

The Einstein equations in wave-map gauge with source a given
stress-energy tensor $T$,
\begin{equation}
R_{\alpha\beta}^{(H)}\equiv-{\frac{1}{2}}g^{\lambda\mu}
\hat{D}_{\lambda}\hat{D}_{\mu}g_{\alpha\beta}
+\hat{f}[g,\hat{D}g]_{\alpha\beta}=\rho_{\alpha\beta}\;,
\qquad
\rho_{\alpha\beta}:=T_{\alpha\beta}
-\frac{\text{tr}_{g}T}{n-1}g_{\alpha\beta}\;,
\end{equation}
form a quasi-diagonal,
hyperquasi-linear\footnote{{\footnotesize That is, the
principal second order terms are diagonal and their
coefficients depend on the unknowns but not on their
derivatives}} system of wave equations for the Lorentzian
metric $g$. The Cauchy problem for such systems with data on a
spacelike $n$-manifold $M_{0}$ is well understood, the Cauchy
data are the values of the unknown on $M_{0}$ and their first
transversal derivatives. When $M_{0}$ is not spacelike  in the
spacetime $(V,g)$ which we are going to construct, the problem
is more delicate. It is known since Leray's work
(see~\cite{Leray}),
\footnote{See~\cite{NicolasCIVPCRAS,HormanderCIVP} for a
treatment of generalized solutions of a linear wave equation
with data on a achronal Lipschitz section of a spacetime with
compact spacelike sections.}
that the Cauchy problem for a
linear hyperbolic system on a given globally hyperbolic
spacetime is well posed if $M_{0}$ is ``compact towards the
past''; that is, is intersected along a compact set by the past
of any compact subset of $V$. However the data depend on the
nature of $M_{0}$ and the formulation of a theorem requires
more care. In the case where $M_{0}$ is a null hypersurface,
except at some singular subsets (intersection in the case of
two null hypersurfaces, vertex in the case of a null cone) the
data is only the function, not its transversal derivative, with
some hypotheses which need to be made as one approaches the
singular {set}.

In this article we concentrate on the case of the light cone,
though most of {the calculations of our equations }apply to any
null hypersurface.

\subsection{The Cagnac-Dossa theorem}
 \label{ssCagDos}

To prove the local existence of solutions of Einstein equations
with data on a characteristic cone we use a wave-map gauge, and
an existence theorem for solutions of quasi linear wave
equations with such data.

The proof of an existence theorem for such a characteristic
quasilinear Cauchy problem is inspired by the Leray's idea of
the linear case, applied to the characteristic cone and linear
wave equations in Cagnac~\cite{Cagnac1980} (cf. also
Cagnac~\cite{Cagnac73join} and Friedlander~\cite{Friedlander});
extended to the quasilinear case by
Cagnac~\cite{Cagnac1981}.
The most complete results appear in
Dossa's thesis, the second part of which is published in
abbreviated form in~\cite{Dossa97}. One considers
quasi-diagonal, quasi-linear second order system for a set $v$
of scalar functions $v^{I}$, $I=1,\ldots,N$, on
$\mathbf{R}^{n+1}$ of the form
\begin{eqnarray}
 \label{22XI.1}
 &
A^{\lambda\mu}(y,v)\partial_{\lambda\mu}^{2}v+f(y,v,\partial v)=0\;,
\quad y=(y^{\lambda})\in\mathbf{R}^{n+1},
\quad \lambda,\mu=0,1,\ldots,n\geq2
 \;,
 \phantom{xxxx}
&
\\
&
 v =(v^{I}),
 \quad
 \partial v =(\frac{\partial v^{I}}{\partial y^{\lambda}})
 \quad
 \partial_{\lambda\mu}^{2}v
 =(\frac{\partial^{2}v^{I}}{\partial y^{\lambda}\partial y^{\mu}}),
 \quad
 f=(f^{I}),
 \quad
 I=1,\ldots,N\;.
 &
\end{eqnarray}
The initial data
\begin{equation}
 \label{22XI.2}
\overline{v}:=v|_{C_{O}}=\phi
\end{equation}
are given on a subset, including its vertex $O$, of a
characteristic cone $C_{O}$. Throughout this work a bar over an
object denotes the restriction of that object to $C_{O}$.

Cagnac and Dossa assume that there is a domain $U\subset
{\mathbf{R}}^{n+1}$ where $C_{O}$ is represented by the
following cone$^{\mbox{\scriptsize~\ref{f19V10.1}}}$ in
$\mathbf{R}^{n+1}$ (compare Figure~\ref{Fcones} below)
\begin{equation*}
C_{O}:=\{x^{0}\equiv r-y^{0}=0\}\;,\quad
r^{2}:=\sum_{i=1,\ldots,n}(y^{i})^{2}.
\end{equation*}
The initial data $\phi$ is assumed to be defined on the domain
\begin{equation}
 \label{22XI.3}
 C_{O}^{T}:=C_{O}\cap\{0\leq t:=y^{0}\leq T\}
 \;.
\end{equation}
They denote
\begin{equation}
Y_{O}:=\{t:=y^{0}>r\}\;,
\quad
\text{the interior of\ }C_{O}\;,\text{ \ }%
Y_{O}^{T}:=Y_{O}\cap\{0\leq y^{0}\leq T\}\;.
\end{equation}
They also set
\begin{align}
\Sigma_{\tau} & :=C_{O}\cap\{y^{0}=\tau\}\;,
\quad\text{diffeomorphic to\ } S^{n-1}\;,\\
S_{\tau} &  :=Y_{O}\cap\{y^{0}=\tau\}\;,
\quad\text{diffeomorphic to the ball\ }B^{n-1}\;.
\end{align}

We will use the following theorem given in the first part of
Dossa's thesis: it assumes some more differentiability of the
data than the theorem in~\cite{Dossa97}, but it is simpler to
apply to the Einstein equations whose initial data must satisfy
wave-map-gauge constraints, and is sufficient for us here.

Remark that these results assume more regularity from the data
on the cone than the regularity obtained for the solution, a
constant fact in charateristic Cauchy problem
already seen in other contexts.

\begin{theorem}
 \label{ThDossanew1}
Consider the problem (\ref{22XI.1})--(\ref{22XI.2}).
Suppose that:

1. There is an open set $U\times W\subset\mathbf{R}^{n+1}\times\mathbf{R}^{N}%
$, $Y_{O}^{T}\subset U$ where the functions $A^{\lambda\mu}$
are $C^{2m+2}$ in $y$ and $v$. The function $f$ is $C^{2m}$ in
$y\in U$ and $v\in W$ and in $\partial
v\in\mathbf{R}^{(n+1)N}$.

2. For $(y,v)\in U\times W$ the quadratic form $A^{\lambda\mu}$
has Lorentzian signature; it takes the Minkowskian values
for $y=0$ and $v=0$.

3. a. The function $\phi$ takes its values in $W$. The cone
$C_{O}^{T}$ is null for the metric $A^{\lambda\mu}(y,\phi)$ and
$\phi(O)=0$.

b. $\phi$ is the trace on $C_{O}^{T}$ of a $C^{2m+2}$ function
in $U$.

Then there is a number $0<T_{0}\le T < +\infty$, $T_{0}=T$ if
$\phi$ is small enough in $C^{2m+2}$ norm, such that the
problem (\ref{22XI.1})--(\ref{22XI.2}) has one and only one
solution $v$ in $Y_{O}^{T_{0}}$, such that

1. If $m>\frac{n}{2}+1$, $v\in K^{m+1}(Y_{O}^{T_{0}})\cap
F^{m+1}(Y_{O}^{T_{0}})$, in particular $|\partial v|$ is
bounded.

2. If $m=\infty$, $v$ can be extended by continuity to a
$C^{\infty}$ function defined on a neighbourhood
of the origin in ${\mathbf{R}}^{N+1}$.

\end{theorem}

The spaces $K^{m}(Y_{O}^{T})$ and $F^{m}(Y_{O}^{T})$\ are
Banach spaces of sets of functions on $Y_{O}^{T}$\ which
together with their time and space
derivatives of order less or equal to
$m$ admit a square integrable restriction to each $S_{t}$\ and
for which, respectively, the following norms are finite
\begin{equation*}
||v||_{K^{m}(Y_{O}^{T})}:=\sum_{I=1,\ldots,N}
\Big\{\int_{0}^{T}t^{-n}\sum_{0\leq|k|\leq m}
||\partial^{k}v^I||_{L^{2}(S_{t})}^{2}\;dt\Big\}^{\frac{1}{2}}
\; ,
\end{equation*}
\begin{equation*}
||v||_{F^{m}(Y_{O}^{T})}:=\sum_{I=1,\ldots,N}\;\;
\sup_{0\leq t\leq T}t^{-\frac{n}{2}}
\sum_{0\leq |k|\leq m}||\partial^{k}v^I||_{L^{2}(S_{t})}
\; .
\end{equation*}
The Euclidean metric, $e:=\sum_{i}(dy^{i})^{2}$, is used to
define the measure
on $S_{t}$\ and as usual $k$ denotes a multi-index,\ $k:=(k_{0},k_{1}%
,\ldots,k_{n})$, $\partial^{k}$ the derivation of order $|k|:=k_{0}%
+k_{1}+\ldots+k_{n}$:
\begin{equation}
{\partial}^{k}:=(\partial_{ {0}})^{k_{0}}(\partial_{ {1}})^{k_{1}}%
\ldots(\partial_{ {n}})^{k_{n}}\;,\quad\text{with}\ \partial_{\alpha}%
:=\frac{\partial}{\partial y^{\alpha}}.\label{30VI.2t}%
\end{equation}

\subsection{Einstein equations in the wave-map gauge}

We know that the wave-map gauge reduced Einstein equations on a
manifold $V$ are tensorial equations under coordinate changes,
so that any coordinates can be used.  Note that the principal
part of the wave-map reduced Einstein equations is independent
of the target manifold, and so the Einstein equations on
$\mathbf{R}^{n+1}$ in wave-map gauge are of the form
(\ref{22XI.1}) for an unknown $h$, when we set $g\equiv\eta+h$
and work in the $y$ coordinates where the Minkowski metric
takes the standard form
\begin{equation*}
\eta\equiv-(dy^{0})^{2}+\sum_{i}(dy^{i})^{2}.
\end{equation*}
As an application of Theorem~\ref{ThDossanew1} we obtain
(see also \cite{DossaAHP} in space-dimension $n=3$):

\begin{theorem}
\label{Existence1} (Existence for the wave-gauge reduced
Einstein equations.) Let
$\overline{g}=\overline{\eta}+\overline{h}$ be a quadratic form
on $C_{O}^{T}$ such that the components
${\overline{h}_{\mu\nu}}$ in the coordinates $y^{\mu}$ satisfy
the hypotheses of the Existence Theorem~\ref{ThDossanew1}.
Then, if the source $\rho$ is of class $C^{2m}$ in $Y_{O}^{T}$,
there exists $T_{0}>0$ such that the wave-gauge reduced
Einstein equations\footnote{{\footnotesize We use abstract
index notation when it helps formulate properties of geometric
objects.}} $R_{\alpha\beta}^{(H)}=\rho_{\alpha\beta}$ admit one
and only one solution on $Y_{O}^{T_{0}}$, a Lorentzian metric
$g^{(H)}={\eta}+h$, with $h$ satisfying the conclusions of that
theorem.
\end{theorem}

The reader should keep in mind that a solution of the Einstein
equations in wave-map gauge as above will be a solution of the
original Einstein equations if and only if the wave-map gauge
vector $H$ defined by the constructed spacetime metric is
such that $Ricc=Ricc^{(H)}$,
condition satisfied in particular if $H=0$.

The following theorem is a straightforward adaptation of a
theorem proved long ago {by one of us}~\cite{ChBActa} for
spacelike Cauchy data.

\begin{theorem}
\label{Existence2} Let $g^{(H)}$ be a $C^{3}$ Lorentzian
metric, solution on $Y_{O}^{T}$ of the Einstein equations in
wave-map gauge $S_{\alpha\beta}^{(H)}=T_{\alpha\beta}$. Then
$g^{(H)}$ is a solution on $Y_{O}^{T}$ of the full Einstein
equations $S_{\alpha\beta}=T_{\alpha\beta}$ if the wave-gauge
vector vanishes on $C_{O}^{T}$ and the source $T$ satisfies the
conservation law $\nabla_{\alpha}T^{\alpha\beta}=0$.
\end{theorem}

\noindent\begin{proof}
The identity (\ref{RicciHIdentity}) implies (indices raised with $g)$
\begin{equation}
S^{\alpha\beta}\equiv S^{\alpha\beta(H)}+
{\frac{1}{2}}(
    \hat{D}^{\beta}H^{\alpha}
   +\hat{D}^{\alpha}H^{\beta}
   -g^{\alpha\beta}\hat{D}_{\lambda}H^{\lambda}
).\label{EinsteinHIdentity}
\end{equation}
Hence the equations in wave-map gauge
$S_{\alpha\beta}^{(H)}=T_{\alpha\beta}$ and the Bianchi
identities imply that $H$ satisfy the quasidiagonal linear
homogeneous system of second order equations
\begin{equation}
 \nabla_{\alpha}\hat{D}^{\alpha}H^{\beta}
+\nabla_{\alpha}\hat{D}^{\beta}H^{\alpha}
-\nabla^{\beta}\hat{D}_{\alpha}H^{\alpha}=0 \;,
\label{WaveH}
\end{equation}
whose principal terms are wave equations in the metric $g$
since $\nabla_{\alpha} \hat{D}^{\beta} H^{\alpha} - \nabla^{\beta}
\hat{D}_{\alpha} H^{\alpha}$ is at most first order in $H$. If $g$
is $C^{3}$, $H$ is $C^{2}$, and an energy inequality applied to
this linear equation implies easily that $H=0$ in $Y_{O}^{T}$
if $\overline{H}:=H|_{C_{O}^{T}}=0$.
\end{proof}

\medskip

When the support of the initial data is a spacelike manifold
$M_{0}$ the vanishing of $H$ is guaranteed when the constraint
equations
$(S_{\alpha\beta}-T_{\alpha\beta})n^{\beta}|_{M_{0}}=0$ are
satisfied by the initial data, where $n^{\beta}$ is the field
of unit normals to $M_{0}$ in the space-time one seeks to
construct. One of the main goals of this work is to present a
method to construct initial data on the light-cone which
ensures the vanishing of $\overline{H}$.

\section{Null hypersurfaces, adapted coordinates}
 \label{nullhypersurfaces}

The obtention, and solution, of equations to be satisfied by
initial data to ensure the vanishing of $\overline{H}$ is simpler in
coordinates adapted to the geometry of the null initial
manifold.

\subsection{Adapted coordinates}
\label{NullCoordinates}

Let $M_{0}$ be a hypersurface in $\mathbf{R}^{n+1}$ which will
be a null submanifold of the spacetime $(V,g)$ with $V$ some
domain of $\mathbf{R}^{n+1}$. $M_{0}$ is generated by geodesic
null curves, called rays.  In a manner classical for null
surfaces we choose coordinates $x^{\alpha}$ so that $M_{0}$ is
given by the equation $x^{0}=0$, and on $M_{0}$ the coordinate
$x^{1}$ is a parameter along the rays, denoting by $\ell$ the
tangent vector $\frac{\partial}{\partial x^{1}}$. We assume that
the subspaces $\Sigma_{x^{1}}:\{x^{1}$=constant, $x^{0}$=0\} are
spacelike and diffeomorphic to the same $n-1$ manifold $\Sigma$,
except possibly for $\Sigma_{0}$ which reduces to a point in
the case of a characteristic cone. We denote by $x^{A}$ local
coordinates on $\Sigma$. We have $\ell^{0}=0$, $\ell^{1}=1$,
$\ell^{A}=0$.

The assumption that $M_{0}$, $x^{0}=0$, is a null surface for $g$
is equivalent to saying that $\overline{g}_{11}=0$. The covariant
vector $n:={\rm grad} \,x^{0}$, with $x^{0}=0$ the  equation of
$M_{0}$, is a null vector normal and tangent to $M_{0}$ with
components $n_{0}=1$, $n_{1}=n_{A}=0$.  By uniqueness of null
directions tangent to a light cone we have also $\ell_{A}=0$ and
hence, using that $\ell^{\alpha}=\delta^{\alpha}_{1}$,
$\overline{g}_{1A}=0$. Then the trace on $M_{0}$ of the spacetime
metric reduces in the $x^{\alpha}$ coordinates to
(we put an overbar to denote restriction to $M_{0}$ of spacetime
quantities)
\begin{equation}
\overline{g}:=g|_{x^{0}=0}\equiv
 \overline{g}_{00}(dx^{0})^{2}
+2\nu_{0}dx^{0}dx^{1}
+2\nu_{A}dx^{0}dx^{A}
+\overline{g}_{AB}dx^{A}dx^{B}
\;,
\label{null2}
\end{equation}
where
\begin{equation}
\nu_{0}:=\overline{g}_{01},\quad \nu_{A}:=\overline{g}_{0A},
\end{equation}
We remark that the $\overline{g}_{AB}$ are the non zero components of the
quadratic form $\tilde{g}$ induced by $g$ on $M_{0}$ by the identity map.
They define an $x^{1}$-dependent Riemannian metric on $\Sigma$
\begin{equation}
\tilde{g}_{\Sigma}:=\overline{g}_{AB}dx^{A}dx^{B},\quad A,B=2,\ldots,n.
\label{null4}
\end{equation}

The following identities hold on $M_{0}$, because
$\overline{g}_{\alpha\beta}$ and $\overline{g}^{\alpha\beta}$
are inverse matrices.

\begin{equation}
\overline{g}^{00}\equiv \overline{g}^{0A}\equiv 0, \quad
\nu^{0}:=\overline{g}^{01}=\frac{1}{\nu_{0}} ,
\label{null5}
\end{equation}
\begin{equation}
\overline{g}^{AB}\equiv \tilde{g}^{AB},
\text{ \ with\ }
\tilde{g}^{AB}\text{ the inverse matrix of } \overline{g}_{AB}.
\label{null6}
\end{equation}
We denote
\begin{equation}
\nu^{B}:=\overline{g}^{AB}\nu_{A}
\label{null7}
\end{equation}
then
\begin{equation}
\overline{g}^{A1}\equiv -\nu^{0}\nu^{A},\quad
\overline{g}^{11}\equiv
-(\nu^{0})^{2}\overline{g}_{00}+(\nu^{0})^{2}\nu^{A}\nu_{A}.
\label{null8}
\end{equation}

\begin{remark} {\rm
We consider a null adapted coordinates condition for the trace $\overline{g}$
but, in contradistinction with Bondi-Sachs and other authors, we do not
assume these coordinates for the spacetime metric. Such an assumption is a
gauge condition which can be considered as the analogous for the
characteristic Cauchy problem to the zero shift and constant lapse of the
$n+1$ classical decomposition, gauge which is not well adapted to the
well posedness of the evolution problem.
}\end{remark}

In Appendix~\ref{A7VI.1} we collect formulae useful for
explicit calculations, such as the trace on
$M_{0}$ of the Christoffel symbols of $g$, etc.

See~\cite{galloway-nullsplitting,GourgoulhonJaramillo,JKCPRD,
VinceJimcompactCauchyCMP}
for various useful results concerning null surfaces.

\subsection{Characteristic cones}

\subsubsection{General properties}
\label{generalproperties}

It is no geometric restriction\footnote{See
footnote~\ref{normalcone} and details in section~\ref{ssBcv}}
to assume that in a neighbourhood of its vertex the
characteristic cone\footnote{A cone is a topological manifold
but it is not differentiable at its vertex} of the spacetime we
are looking for is represented in some admissible coordinates
$y:=(y^{\alpha})\equiv (y^{0}$, $y^{i}$, $i=1,\ldots,n)$ of
$\mathbf{R}^{n+1}$ by the equation of a Minkowskian cone with
vertex $O$,
\begin{equation}
r-y^{0}=0\;,\quad  r:=\big\{\sum_{i}(y^{i})^{2}\big\}^{\frac{1}{2}}.
\end{equation}
Given the coordinates $y^{\alpha}$ we can define coordinates
$x^{\alpha}$ on $\mathbf{R}^{n+1}$ adapted to the null cone
$C_{O}$ as we did for a general null surface by setting
\begin{equation}
 \label{1IX.1}
x^{0}=r-y^{0}, \quad
x^{1}=r, \quad
x^{A}=\mu^{A}(\frac{y^{i}}{r})
 \;,
\end{equation}
with $x^{A}$ local coordinates\footnote{ They can be angular
coordinates, see e.g.~\cite[Chapter V, Section 4]{ChBdWMII}, or
stereographic coordinates, as in Christodoulou~\cite{ChrBHF}.}
on $S^{n-1}$.
The  null geodesics issued from $O$ have equation $x^{0}=0$,
{}$x^{A}=$constant, so that $\frac{\partial}{\partial x^{1}}$
is tangent to those geodesics. On $C_{O}$ (but not outside of
it in general) the spacetime metric $g$ {that we are going to
construct} takes the form (\ref{null2}), that is, such that
$\overline{g}_{11}=0$ and $\overline{g}_{1A}=0$.

We emphasize that our assumption $\overline{g}$ is given by
(\ref{null2}) is no geometric restriction on a metric $g$ to
have such a trace, however $\overline{g}_{00}$, $\nu_{0}$,
$\nu_{A}$ are not invariant under an isometry of spacetime with
leaves $C_{O}$ invariant, they are gauge-dependent quantities
(see sections~\ref{ssBcv} and~\ref{slgu}).

We compute the relation between the components of a tensor
$T$ in the coordinates $y$ and $x$ using (\ref{1IX.1}) and
its inverse:
\begin{equation}
y^{0}=x^{1}-x^{0}, \quad
y^{i}=x^{1}\Theta^{i}(x^{A}), \quad
{\rm with} \quad
\sum_{i}(\Theta^{i})^{2}=1
 \;.
\end{equation}

If the components of a spacetime symmetric tensor $T$ in the
coordinates $x^{\alpha}$ are denoted $T_{\alpha\beta}$ and if
in the coordinates $y^{\alpha}$ they are denoted
$\underline{T_{\alpha\beta}}$, then the identities for the
change of coordinates of tensors,
$T_{\lambda\mu}=\underline{T_{\alpha\beta}}\frac{\partial y^{\alpha}%
}{\partial x^{\lambda}}\frac{\partial y^{\beta}}{\partial x^{\mu}}$,
give the identities
\begin{equation}
 \label{22XI.21}
T_{00}\equiv\underline{T_{00}},\quad
T_{01}\equiv-\underline{T_{00}} -\underline{T_{0i}}\Theta^{i},\quad
T_{0A}\equiv-\underline{T_{0i}} r\frac{\partial\Theta^{i}}{\partial x^{A}}
 \;,
\end{equation}
\begin{equation}
 \label{22XI.22}
T_{11}\equiv\underline{T_{00}}+2\underline{T_{0i}}\Theta^{i}
           +\underline{T_{ij}}\Theta^{i}\Theta^{j},\quad
T_{1A}\equiv\underline{T_{0i}} r\frac{\partial\Theta^{i}}{\partial x^{A}}
      +\underline{T_{ji}}r\Theta^{j}
          \frac{\partial\Theta^{i}}{\partial x^{A}}\;,
\end{equation}
\begin{equation}
 \label{22XI.23}
T_{AB}\equiv\underline{T_{ij}}r^{2}\frac{\partial\Theta^{i}}{\partial x^{A}%
}\frac{\partial\Theta^{j}}{\partial x^{B}} \;.
\end{equation}
Conversely, $\underline{T_{\lambda\mu}}=\frac{\partial x^{\alpha}%
}{\partial y^{\lambda}}\frac{\partial x^{\beta}}{\partial y^{\mu}}%
T_{\alpha\beta}$ gives
\begin{equation}
 \label{20XII.1}
\underline{T_{00}}\equiv T_{00},\quad
\underline{T_{0i}}\equiv
-(T_{00}+T_{01})\Theta^{i}-T_{0A}\frac{\partial x^{A}}{\partial y^{i}}
\;,  %
\end{equation}
\begin{equation}
 \label{20XII.2}
\underline{T_{ij}}=(T_{00}+2T_{01}+T_{11})\Theta^{i}\Theta^{j}
+(T_{0A}+T_{1A})
(\Theta^{i}\frac{\partial x^{A}}{\partial y^{j}}+\Theta^{j}\frac{\partial
x^{A}}{\partial y^{i}})
+T_{AB}\frac{\partial x^{A}}{\partial y^{i}}\frac{\partial x^{B}}{\partial y^{j}}
 \;,
\end{equation}
with
\begin{equation*}
\frac{\partial x^{A}}{\partial y^{i}}=r^{-1}\mu_{i}^{A}
\;,
\end{equation*}
where the $\mu_{i}^{A}$'s are $C^{\infty}$ functions of the
$x^{B}$ on any subset of $S^{n-1}$ where the $x^{A}$'s are
admissible local coordinates.

One checks, using the identities
\begin{equation*}
\sum_{i}(\Theta^{i})^{2}=1\ \text{ and }\  \sum_{i}\frac{\partial\Theta^{i}%
}{\partial x^{A}}\frac{\partial\Theta^{i}}{\partial x^{B}}\equiv s_{AB},
\end{equation*}
with
\begin{equation*}
s_{AB}dx^{A}dx^{B}\equiv s_{n-1}\;,\quad \text{the metric of}\  S^{n-1}
 \;,%
\end{equation*}
that, when $T=\eta$,
\begin{equation*}
-(dy^{0})^{2}+\sum_{i}(dy^{i})^{2}=-(dx^{0})^{2}+2dx^{0}dx^{1}+(x^{1}%
)^{2}s_{n-1}
 \;.%
\end{equation*}

\subsubsection{Case $\hat{g}=\protect\eta$, the Minkowski metric}

It is natural to take as given metric $\hat{g}$ the metric of a
model spacetime such as Minkowski, or de Sitter, or anti-de
Sitter. While most our formulae will be completely general, the
analysis will mainly be concerned with the case where the
metric $\hat{g}$ is a Minkowski metric $\eta$ given by the
formula written above in the introduced coordinates
$y^{\alpha}$ and in the adapted null coordinates $x^{\alpha}$.
The Riemannian curvature of the Minkowski metric $\eta$ is
zero. The non zero Christoffel symbols of $\eta$ are in our
coordinates $x^{\alpha}$, with $S_{BC}^{A}$ the Christoffel
symbols of the metric $s$,
\begin{equation}
\hat{\Gamma}_{1A}^{B}\overset{\eta}{\equiv}\frac{1}{x^{1}}\delta_{A}^{B}\;,\qquad
\hat{\Gamma}_{AC}^{B}\overset{\eta}{\equiv}S_{AC}^{B}\;,
\label{Minkowski3}
\end{equation}
\begin{equation}
\hat{\Gamma}_{AB}^{0}\overset{\eta}{\equiv}-x^{1}s_{AB}\;,\qquad
\hat{\Gamma}_{AB}^{1}\overset{\eta}{\equiv}-x^{1}s_{AB}\;.
\label{Minkowski4}
\end{equation}
Equalities and identities assuming given metric $\hat{g}=\eta$
and $W^{\lambda}\equiv
g^{\alpha\beta}\hat{\Gamma}_{\alpha\beta}^{\lambda}$ will be
denoted with symbols $\overset{\eta}{=}$ and
$\overset{\eta}{\equiv}$, respectively when ambiguous.

We have
\begin{equation}
\overline{W}^{0}\stackrel{\eta}{\equiv} -x^{1}\overline{g}^{AB}s_{AB}
\stackrel{\eta}{\equiv} \overline{W}^{1}
\;,
\label{Minkowski6}
\end{equation}
\begin{equation}
\overline{W}^{A}\stackrel{\eta}{\equiv}
2\overline{g}^{1C}\hat{\Gamma}_{1C}^{A}
+\overline{g}^{BC}\hat{\Gamma}_{BC}^{A}
\stackrel{\eta}{\equiv}
-\frac{2}{x^{1}}\nu ^{0}\nu ^{A}+\overline{g}^{BC}S_{BC}^{A}
\;.
\label{Minkowski7}
\end{equation}

\subsubsection{Limits at the vertex}
 \label{ssslv}

We set $g=\eta+h$. The condition
$\underline{\overline{h}_{\alpha\beta}}(O)=0$ of
Theorem~\ref{Existence1} can always be satisfied by choice of
an orthonormal frame for the natural frame of the coordinates
$y^{\alpha}$ at the vertex. Since the coordinates $x$ are
singular for $x^{1}=0$ the behaviour near $x^{1}=0$ of the
components $\overline{h}_{\alpha\beta}$ in $x$ coordinates is
obtained only by considering limits.  As explained above, we
can, and will, choose coordinates on $C_{O}$ such that
$\overline{h}_{11}\equiv0$, i.e. $C_{O}:x^{0}=0$ is a null cone
for $g$, and $\overline{h}_{1A}=0$ i.e. the vector
$\ell:=\frac{\partial}{\partial x^{1}}$ is on $C_{O}$ a null
vector. Then the components of $\underline{\overline{h}}$ are
\begin{equation*}
\underline{\overline{h}_{00}}\equiv\overline{h}_{00},\quad  \underline{\overline{h}_{0i}%
}  \equiv-(\overline{h}_{00}+\overline{h}_{01})\Theta^{i}-\overline{h}_{0A}%
\frac{\partial x^{A}}{\partial y^{i}}\text{\ \ with \ \ } \overline{h}_{01}:=\nu
_{0}-1,\quad  \overline{h}_{0A}:=\nu_{A}
\;,%
\end{equation*}
\begin{equation*}
\quad  \underline{\overline{h}_{ij}}=(\overline{h}_{00}+2\overline{h}_{01})\Theta^{i}%
\Theta^{j}+\overline{h}_{0A}(\Theta^{i}\frac{\partial x^{A}}{\partial y^{j}}%
+\Theta^{j}\frac{\partial x^{A}}{\partial y^{i}})+\overline{h}_{AB}\frac{\partial
x^{A}}{\partial y^{i}}\frac{\partial x^{B}}{\partial y^{j}}.
\end{equation*}

We see that the condition
$\underline{\overline{h}_{\alpha\beta}}(O)=0$ is equivalent to
the following conditions in the coordinates $x^{\alpha}$:
\begin{equation}
\lim_{r\rightarrow 0}(1+\overline{g}_{00})=
\lim_{r\rightarrow 0}(\nu_{0}-1)=
\lim_{r\rightarrow 0}(r^{-1}\nu_{A})=
\lim_{r\rightarrow 0} \,(r^{-2}\overline{h}_{AB})=0
\;.
\label{gvertexlimits}
\end{equation}

\subsubsection{A lemma}

For further use we note the following observation:

\begin{lemma}
\label{lemma4.2}
If a $C^{1}$ spacetime function $f$ is such that
on $C_{O}$ in the coordinates $x^{\alpha}$ it holds that
\begin{equation*}
\lim_{r\equiv x^{1}\rightarrow 0}\partial_{1}\overline{f}=0
\end{equation*}
then it also holds
\begin{equation*}
\lim_{r\equiv x^{1}\rightarrow 0}\overline{\partial_{0}f}=0\;.
\end{equation*}
\end{lemma}

\begin{proof}
One has the trivial identity
\begin{equation*}
\partial_{1}f\equiv \underline{\partial_{\alpha }f}\frac{\partial
y^{\alpha }}{\partial x^{1}}\equiv
\underline{\partial_{0}f}+\underline{\partial_{i}f}\Theta^{i} .
\end{equation*}
If $f$ is $C^{1}$ in a neighbourhood of $O$,
$\underline{\partial_{i}f}$ tends to a limit, a number $a_{i}$
at $O$, hence the above equation implies
\begin{equation*}
\lim_{r\rightarrow 0}\overline{\underline{\partial_{0}f}}
+a_{i}\Theta^{i}=0\;,
\end{equation*}
condition which can be satisfied for all $x^{A}$ if and only if $a_{i}=0$
and $\lim_{r\rightarrow 0}\overline{\underline{\partial_{0}f}}=0$.
Therefore
\begin{equation*}
 \lim_{r\rightarrow 0}\overline{\partial_{0}f}\equiv
-\lim_{r\rightarrow 0}\underline{\overline{\partial_{0}f}}=0\;.
\end{equation*}
\end{proof}

\subsection{The affine-parameterization condition}

The vector field $\ell:=\frac{\partial}{\partial x^{1}}$,
tangent to the null rays in $M_{0}$, obeys the geodesic
property
\begin{equation}
\overline{\ell^{\alpha}\nabla_{\alpha}\ell^{\beta}}
=\overline{\Gamma}_{11}^{\beta}
\; ,
\label{chi5}
\end{equation}
with
\begin{equation}
\overline{\Gamma}^{0}_{11} \equiv \overline{\Gamma}^{A}_{11} \equiv 0
\quad \text{and} \quad
\overline{\Gamma}_{11}^{1}\equiv
\nu^{0}(\partial_{1}\nu_{0}-\frac{1}{2}\overline{\partial_{0}g_{11}})
\; .
\label{geodesic}
\end{equation}
If we impose the condition $\overline{\Gamma}_{11}^{1}=0$, then
the vector $\ell$ is parallelly transported, and $x^{1}$ is
said to be an affine parameter on the rays; this condition
gives an equation involving $\nu_{0}$, a metric coefficient which
will appear in our first wave-map gauge constraint. However, we
stress that the equation $\overline{\Gamma}^{1}_{11}=0$
involves also
a derivative transversal to $M_{0}$, and thus cannot
be made to hold just by a gauge choice of the coordinate
$x^{1}$ on the cone. We will see later how we can circumvent
this problem in the wave-map gauge.

\subsection{Null extrinsic curvature}
\label{NullExtrinsicCurvature}

\subsubsection{General properties}

Let $M_{0}$ be a null hypersurface with a field of null tangents
$\ell$. The null extrinsic curvature at $x\in M_{0}$ is defined
(see, e.g.,~\cite{galloway-nullsplitting}) as the bilinear form
with components $\overline{\nabla_{\alpha} \ell_{\beta}}$
acting on the quotient of the tangent space to $M_{0}$ at $x$
by the direction defined by $\ell$, i.e. equivalence classes of
tangent vectors of the form $\overline{X}\equiv
$\underline{$\overline{X}$}$+c\ell$ with
\underline{$\overline{X}$}$\in T_{x}M_{0}$, $c$ an arbitrary
number. Indeed, the action of the bilinear form on a pair of
such tangent vectors,
$\overline{\nabla_{\alpha}\ell_{\beta}X^{\alpha}Y^{\beta}}$,
depends only on the equivalence class, that is in our
coordinates on the components $X^{A}$ and $Y^{A}$, hence it is
defined by the components
$\chi_{AB}:= \overline{\nabla_{A}\ell_{B}}$ of the bilinear form.
Using $\ell^{\alpha}=\delta_{1}^{\alpha}$ and
$\ell_{\alpha}:=\overline{g}_{\alpha\beta}\delta_{1}^{\beta}
\equiv \nu_{0}\delta_{\alpha}^{0}$ we have
\begin{equation}
\chi_{AB}\equiv
 -\overline{\Gamma}_{AB}^{0}\nu_{0}\equiv
 \frac{1}{2}\partial_{1}\overline{g}_{AB}.
\label{chi2}
\end{equation}
We denote by
\begin{equation}
\chi_{A}{}^{B}:=\overline{g}^{BC}\chi_{AC} \equiv \overline{\Gamma}_{1A}^{B}
\label{chi3}
\end{equation}
the mixed, $x^{1}$-dependent, 2-tensor on $S^{n-1}$ deduced from the null
second fundamental form. We define its trace
\begin{equation}
\tau :=\chi_{A}{}^{A}\equiv
\overline{g}^{AB}\chi_{AB}\equiv
\partial_{1}(\log \sqrt{\det \tilde{g}_{\Sigma}}),
\label{tau}
\end{equation}
and its traceless part
\begin{equation}
\sigma_{A}{}^{B}:=\chi_{A}{}^{B}-\frac{1}{n-1}\delta_{A}^{B}\tau ,
\text{ \ and we set }
|\sigma|^{2}:=\sigma_{A}{}^{B}\sigma_{B}{}^{A}\; .
\label{sigma}
\end{equation}

See~\cite{GourgoulhonJaramillo,galloway-nullsplitting} for an
analysis of the null second fundamental form through the
Weingarten map.

\subsubsection{Harmonicity functions}
\label{HarmonicityFunctions}

In adapted coordinates,  and using the notation above, the
harmonicity functions $\Gamma^{\alpha} \equiv g^{\lambda\mu}
\Gamma_{\lambda\mu}^{\alpha}$ reduce on $M_{0}$ to
\begin{equation}
\overline{\Gamma}^{0}\equiv \overline{g}^{\lambda \mu} \overline{\Gamma}_{\lambda \mu
}^{0}\equiv 2\nu ^{0}\overline{\Gamma}_{01}^{0}+\overline{g}^{AB}\overline{\Gamma}
_{AB}^{0}\equiv \nu ^{0}(\nu ^{0}\overline{\partial_{0}g_{11}}-\tau )
 \;,
\label{Gamma2}
\end{equation}
\begin{align}
\overline{\Gamma}^{A}\equiv &
\overline{g}^{\lambda \mu} \overline{\Gamma}_{\lambda \mu} ^{A}
\equiv
2\nu^{0}\overline{\Gamma}_{01}^{A}+2\overline{g}^{B1}\overline{\Gamma}_{B1}^{A}+\overline{g}%
^{BC}\overline{\Gamma}_{BC}^{A}
\label{Gamma3} \\
\equiv &\ \nu ^{0}\nu ^{A}(\tau -\nu ^{0}\overline{\partial_{0}g_{11}}%
)+\nu ^{0}\overline{g}^{AB}(\overline{\partial_{0}g_{1B}}+\partial_{1}\nu
_{B}-\partial_{B}\nu _{0}) \notag \\
&\  - 2\nu^{0}\nu^{B}\chi_{B}{}^{A} +\tilde{\Gamma}^{A},
\label{Gamma4}
\end{align}
\begin{align}
\overline{\Gamma}^{1} \equiv&\ \overline{g}^{\lambda \mu} \overline{\Gamma}_{\lambda \mu
}^{1}\equiv
\overline{g}^{11}\overline{\Gamma}^{1}_{11} +
 2\nu ^{0}\overline{\Gamma}_{01}^{1}-2\nu ^{0}\nu ^{A}\overline{\Gamma}%
_{A1}^{1}+\overline{g}^{AB}\overline{\Gamma}_{AB}^{1}
\label{Gamma5} \\
\equiv &\ (\nu^{0})^{2}\partial_{1}\overline{g}_{00}+\overline{g}^{11}\nu^{0}
(\frac{1}{2} \overline{\partial_{0}g_{11}}+\partial_{1}\nu_{0}-\tau\nu_{0})
+2(\nu^{0})^{2}\nu^{A}(-\partial_{1}\nu_{A}+\nu^{B}\chi_{AB}) \notag \\
& +\nu^{0}\overline{g}^{AB}\tilde{\nabla}_{B}\nu_{A}-
\frac{1}{2}\nu^{0}\overline{g}^{AB}\overline{\partial_{0}g_{AB}}
\label{Gamma6} \\
\equiv &\
-\partial_{1} \overline{g}^{11}
+ \overline{g}^{11} \nu^{0} (
    \frac{1}{2} \overline{\partial_{0} g_{11}}
  - \partial_{1} \nu_{0}
  - \tau \nu_{0} )
+ \nu^{0} \overline{g}^{AB} \tilde{\nabla}_{B} \nu_{A}
- \frac{1}{2} \nu^{0} \overline{g}^{AB} \overline{\partial_{0}g_{AB}}
 .
\label{27V.1}
\end{align}
We have defined
\begin{equation}
\label{22XII.1}
 \tilde{\Gamma}^{A}:=\overline{g}^{BC}\tilde{\Gamma}^{A}_{BC}
  \;,
\end{equation}
with $\tilde{\Gamma}^{A}_{BC}$ being the Christoffel symbols of
the metric $\overline{g}_{AB}$. We shall also use
\begin{eqnarray}
\overline{\Gamma}_{1} &:=& \overline{g}_{1\mu} \overline{\Gamma}^{\mu} =
\nu_{0} \overline{\Gamma}^{0} \;, \\
\overline{\Gamma}_{A} &:=& \overline{g}_{AB} \overline{\Gamma}^{B}
\not = \overline{g}_{A\mu} \overline{\Gamma}^{\mu} \;,
\label{GammaA}
\end{eqnarray}
and similarly for components of $\overline{W}$ and
$\overline{H}$ with subindices.

\subsubsection{Vertex limits}
 \label{sssvl2}

We set
\begin{equation}
\overline{g}_{AB}\equiv r^{2}(s_{AB}+\overline{f}_{AB})
\;.
\end{equation}
We have seen in Section~\ref{ssslv} that  it is no geometric
restriction for smooth metrics to assume
\begin{equation*}
\lim_{r\to 0} \,(r^{-2}\overline{g}_{AB}-s_{AB})=0\;,
\quad\text{i.e.}\quad
\lim_{r\rightarrow 0}\overline{f}_{AB}=0
 \;,
\end{equation*}
Then
\begin{equation}
\lim_{r\rightarrow 0} r^{2}\overline{g}^{AB} = s^{AB}
\; .
\end{equation}

Recalling that $\chi_{A}{}^{C}\equiv\overline{\Gamma}_{A1}^{C}$ and using
the relation between connections in different frames gives
\begin{equation*}
\chi_{A}{}^{C}\equiv
\overline{\Gamma}_{A1}^{C}\equiv
\frac{\partial x^{C}}{\partial y^{\alpha}}
\frac{\partial y^{\beta}}{\partial x^{A}}
\frac{\partial y^{\gamma}}{\partial x^{1}}
\underline{\overline{\Gamma}_{\beta\gamma}^{\alpha}}+
\frac{\partial x^{C}}{\partial y^{\alpha}}
\frac{\partial}{\partial x^{A}}
\frac{\partial y^{\alpha}}{\partial x^{1}}
\; .
\end{equation*}
Using
\begin{equation*}
\frac{\partial y^{0}}{\partial x^{1}}=1, \quad
\frac{\partial y^{i}}{\partial x^{1}}=\frac{y^{i}}{r},
\quad\text{hence}\quad
\frac{\partial x^{C}}{\partial y^{\alpha}}
\frac{\partial}{\partial x^{A}}
\frac{\partial y^{\alpha}}{\partial x^{1}}=
\frac{1}{r}
\frac{\partial x^{C}}{\partial y^{\alpha}}
\frac{\partial y^{\alpha}}{\partial x^{A}}=
\frac{1}{r}\delta_{A}^{C}
\; ,
\end{equation*}
we find
\begin{equation*}
\chi_{A}{}^{C}\equiv
\frac{1}{r}\delta_{A}^{C}+
\frac{\partial x^{C}}{\partial y^{i}}
\frac{\partial y^{j}}{\partial x^{A}}
(\underline{\overline{\Gamma}_{j0}^{i}}
+ \frac{y^{h}}{r}\underline{\overline{\Gamma}_{jh}^{i})}
\; .
\end{equation*}
Therefore if the coefficients
$\underline{\overline{\Gamma}_{j0}^{i}}$ $\ $and
$\underline{\overline{\Gamma}_{jk}^{i}}$ are bounded for $0\leq
r\leq a$, the same property holds for for
$\chi_{A}{}^{C}-\frac{1}{r}\delta_{A}^{C}$, for
$\psi:=\frac{n-1}{r}-\tau$ and for
$\sigma_{A}{}^{C}:=\chi_{A}{}^{C}-\frac{1}{n-1}\delta_{A}^{C}\tau$.
These quantities are then also continuous on each null ray.
However the limits when $r$ tends to zero are in general angle
dependent.

We have already said no geometric restriction for a given
spacetime metric to choose for it normal geodesic coordinates
in a neighbourhood of a point, here $O$. Then the Christoffel
symbols vanish at $O$ and if continuous are such that
\begin{equation*}
\lim_{r\rightarrow0}\overline{\Gamma}_{A1}^{C}
=
\lim_{r\rightarrow0}
   \frac{\partial x^{C}}{\partial y^{\alpha}}
   \frac{\partial}{\partial x^{A}}
   \frac{\partial y^{\alpha}}{\partial x^{1}}
\equiv
   \frac{1}{r}
   \frac{\partial x^{C}}{\partial y^{\alpha}}
   \frac{\partial y^{\alpha}}{\partial x^{A}}
\equiv
   \frac{1}{r} \delta_{A}^{C}.
\end{equation*}
Hence
\begin{equation}
\lim_{r\rightarrow0}\psi=\lim_{r\rightarrow0}\sigma_{A}{}^{C}=0 .
\end{equation}
See further results in next section.

\subsection{Boundary conditions in coordinates normal at the vertex}
 \label{ssBcv}

In the following sections we will give explicit expressions for
the  wave-map gauge  constraints. To study their solutions we
will need to know the behaviour of the unknowns at the tip of
the light cone, aiming at finding solutions of the constraints
which satisfy the Cagnac-Dossa hypotheses. The purpose of this
section is to describe this behaviour in coordinate systems
useful for the problem at hand. The analysis here is also
useful for proving geometric uniqueness of solutions.

Consider a smooth space-time $(V, g)$. Let $O\in V$ and let
$C_{O}$ be the future light-cone emanating from $O$. Let $T$ be
any unit timelike vector at $O$, and normalize  null vectors
$\ell$ at $O$ by requiring that $g(\ell,T)=-1$.  {The parallel
transport of $\ell$} defines an affine parameter, denoted by
$s$, on the future null geodesics $s\mapsto\gamma_{\ell}(s)$ with
$\gamma_{\ell}(0)=O$ and with initial tangent $\ell$. Let
$(z^{\mu})$, $\mu=0,\ldots ,n$ be a normal coordinate system
centred at $O$ with $T=\partial_{z^{0}}$, see
e.g.~\cite{Thomas,Lee:Rm}, or~\cite[Chapter 12, Section
7]{YvonneBook}. In those coordinates the future light-cone
emanating from $O$ is given by the equation
\begin{equation*}
 C_{O}=\{z^{0} = |\vec z|\}
 \;,
 \
 \text{where} \ \vec z :=(z^{1},\ldots ,z^n)\;, \
 |\vec z|^{2}:=  {\sum_{i=1}^n (z^{i})^{2}}
 \;.
\end{equation*}
As is well known, in normal coordinates  at $O$, $z=0$, the
Christoffel symbols vanish at $O$. Hence for a $C^{1,1}$ metric
we have $\partial_{\sigma} g_{\mu\nu}(O)=0$, and so, for small
$|z|:=|z^{0}|+|\vec z|$,
\begin{equation}
\label{7VI.2}
 |g_{\mu\nu} - \eta_{\mu\nu}| + |z\|\partial_{\sigma} g_{\mu\nu}|
 \le C |z|^{2}
  \;,
\end{equation}
for some constant $C$.

In the coordinate system $(x^{\mu})=(x^{0}\equiv u,x^{1}\equiv
r,x^{A})$, $A=2,\ldots ,n$, where
\begin{equation}
 \label{7VI.9}
 u =|\vec z|- z^{0} \;, \ r= |\vec z|
 \;,
\end{equation}
and where the $x^{A}$'s are any local coordinates on $S^{n-1}$
parameterizing the unit vector $\vec z/|\vec z|$, the trace of
the metric $g$ on $C_{O}$ takes the desired form (\ref{null2})
as long as the metric, and the light-cone, are smooth;
assuming smoothness of $g$, this will always be the case in
some neighbourhood of the tip $O$.

Equation (\ref{7VI.9}) shows that
\begin{equation*}
 dz^{0} = dr-du\;, \quad
 dz^{i} = \frac{z^{i}}{r} dr + \partial_{A} z^{i} dx^{A}
 \;,
\end{equation*}
which allows us to translate the estimates (\ref{7VI.2}) to the
asymptotic behaviour of the objects of interest near $r=0$: From
\begin{eqnarray*}
 g_{\mu\nu}dz^{\mu} dz^{\nu}
 & = &
 (\eta_{\mu\nu} + O(|z|^{2})) dz^{\mu} dz^{\nu}
\\
 & = &
 (-1+ O(|z|^{2})) (du-dr)^{2} + O(|z|^{2}) (du-dr)
 \left(\frac {z^{i}}r dr + \partial_{A} z^{i} dx^{A}\right)
\\
 &&
 + (\delta_{i}^{j} +  O(|z|^{2}) )
 \left(\frac {z^{i}}r dr + \partial_{A} z^{i} dx^{A}\right)
 \left(\frac {z^{j}}r dr + \partial_{B} z^{j} dx^{B}\right)
\\
 &=&
    (-1+ O(|z|^{2}) ) (du)^{2} + (2+ O(|z|^{2})) du dr +  O(r|z|^{2}) du dx^{A}
\\
 &&
    +  O(|z|^{2}) (dr)^{2}+  O(r |z|^{2}) dx^{A} dr  +
  r^{2} \left(s_{AB}+ O(|z|^{2}) \right) dx^{A} dx^{B}
\end{eqnarray*}
we obtain, at $u=0$, for small $r$,
\begin{eqnarray}
 \label{7VI.33}
 &
 \overline{g}_{00} = -1 + O(r^{2})\;,\
 \partial_{r} \overline{g}_{00} = O(r)\;,\
 \overline{\partial_{u} g_{00}} = O(r)\;,\
 \partial_{A} \overline{g}_{00} = O(r^{2})\;,
  \phantom{xxx}
 &
\\
 \label{7VI.34}
 &
 \nu_{0} = 1 + O(r^{2})\;,\
 \partial_{r} \nu_{0} = O(r)\;,\
 \overline{\partial_{u} g_{01}} = O(r)\;,\
 \partial_{A} \nu_{0} = O(r^{2})\;,
 &
\\
 \label{7VI.35}
 &
 \nu_{A} = O(r^3)\;,\
 \partial_{r} \nu_{A} = O(r^{2})\;,\
 \overline{\partial_{u} g_{0A}} = O(r^{2})\;,\
 \partial_{A} \nu_{0} = O(r^3)\;,
 &
\\
 \label{7VI.36}
 &
 \overline{g}_{AB} = r^{2} \left(s_{AB}+ O(r^{2}) \right)\;,\
 \partial_{r} (\overline{g}_{AB} - r^{2} s_{AB}) = O(r^{3})\;,
 &
\\
 \label{7VI.37}
 &
 \overline{\partial_{u} g_{AB}} = O(r^{3})\;,\
 \partial_{A} (\overline{g}_{AB} - r^{2} s_{AB}) = O(r^{4})\;,
 &
\\
 \label{7VI.38}
 &
 \overline{g}_{11} =0\;,\
 \partial_{r} \overline{g}_{11} = 0 \;,\
 \overline{\partial_{u} g_{11}} = O(r)\;,\
 \partial_{A} \overline{g}_{11} = 0\;,
  \phantom{xxx}
 &
\\
 \label{7VI.39}
 &
 \overline{g}_{1A} =0\;,\
 \partial_{r} \overline{g}_{1A} = 0\;,\
 \overline{\partial_{u} g_{1A}} = O(r^{2})\;,\
 \partial_{A} \overline{g}_{1A} = 0\;.
  \phantom{xxx}
 &
\end{eqnarray}
One also has associated second-derivatives estimates,
\begin{equation}
 \label{10VI.1}
 \overline{\partial_{u} \partial_{r} g_{AB}} = O(r^{2})\;,\
 \partial^{2}_{r} (\overline{g}_{AB} - r^{2} s_{AB}) = O(r^{2})\;,\
 \partial_{A} \partial_{r} (\overline{g}_{AB} - r^{2} s_{AB}) = O(r^{3})\;,
\end{equation}
etc. From (\ref{7VI.36}) and (\ref{10VI.1}) we obtain
\begin{eqnarray}
 \label{7VI.40}
 &
 \chi_{A}{}^{B} = \frac{1}{r} \delta_{A}^{B} + O(r)\;,\
 \text{hence}\
 \tau = \frac{n-1}{r} + O(r)\;,\
 \sigma_{A}{}^{B} = O(r)\;,
 &
\\
 \label{7VI.41}
 &
 \text{as well as} \
 \partial_{r}(\tau - \frac{n-1}{r}) = O(1)\;,\
 \partial_{A} \tau = O(r) \;,
 &
\\
 \label{7VI.42}
 &
 \partial_{r} \sigma_{A}{}^{B} =  O(1)\;,\
 \partial_{C} \sigma_{A}{}^{B} =  O(r)\;.
 &
\end{eqnarray}

Note that (\ref{7VI.40})-(\ref{7VI.42}) will hold in any
coordinate system which \emph{coincides with the normal
coordinates $z^{\mu}$ on the light-cone}.
This is due to the fact
that the vectors $\partial_{r}$ and $\partial_{A}$ are tangent to
the light-cone, which implies that the quadratic form
$\tilde{g}_{\Sigma} = \overline{g}_{AB}dx^{A} dx^{B}$ is intrinsically
defined on the light-cone, independently of how the
coordinates are extended away from the light-cone, and from the
fact that the matrix $\overline{g}^{AB}$ in (\ref{29VIII.1}) is the
inverse of $\overline{g}_{AB}$.

\subsection{The light-cone theorem}
 \label{ss7VI.1}

A result closely related to our analysis here is the
\emph{light-cone theorem}, proved in~\cite{CCG}, which reads as
follows: Let $s$ be an affine parameter as defined at the
beginning of Section~\ref{ssBcv}. Let $\Sigma_s$ denote the
$(n-1)$--dimensional surface reached by these geodesics after
affine time $s$:
\begin{equation}\label{20II.5}
 \Sigma_s=\{\gamma_{\ell}(s)\}\subset C_{O}
 \;,
\end{equation}
where the vectors $\ell$ run over all null future vectors at
$O$ normalized as above; see Figure~\ref{Fcones}.
Note that in this subsection~\ref{ss7VI.1}, and only here, we
shall slice the null cone $C_{O}$ and its interior using an
affine parameter $s$, and not coordinate time $y^{0}$ as in the
rest of the article.

\begin{figure}[t]
\begin{center}
{
 \psfrag{gamltwo}{\Huge $\gamma_{\ell_{2}} $}
 \psfrag{gamlone}{\Huge $\gamma_{\ell_{1}} $}
 \psfrag{Tv}{\Huge $T$}
 \psfrag{point}{\Huge $O$}
 \psfrag{Cdes}{\Huge $C_{O}^{s}$}
 \psfrag{Ades}{\Huge $\Sigma_s$}
 \resizebox{3in}{!}{\includegraphics{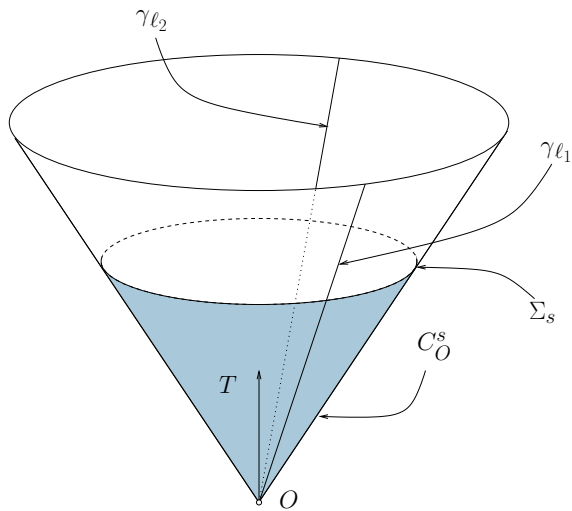}}
}
\caption{\label{Fcones}
The cross-section $\Sigma_s$ of the light-cone $C_{O}$; $C_{O}^{s}$ is the
shaded blue region. Two generators $\gamma_{\ell_{1}}$
and $\gamma_{\ell_{2}}$ are also shown.}
\end{center}
\end{figure}

We denote by $C_{O}^{t}$ the subset
of the light-cone covered by all the geodesics up to affine time $t$:
\begin{equation}\label{20II.6}
 C_{O}^{t} = \cup_{0\le s\le t} \Sigma_s
 \;.
\end{equation}
Note that $\gamma_{\ell}(s)$ might not be defined for all $s$.
Further, some of the $\Sigma_{s}$'s might not be  smooth.
However, there exists a maximal $s_{0}>0$ such that $\Sigma_{s}$ is
defined and smooth for all $0<s<s_{0}$. Our considerations  only
apply to that last region.

It is proved in~\cite{CCG} that, assuming the Einstein
equations with a cosmological constant and with sources
satisfying the dominant energy conditions, the areas of the
$\Sigma_{s}$'s are less than or equal to the corresponding areas
in Minkowski, de Sitter, or anti-de Sitter space-time.
Furthermore, if equality holds at some $s_{2}$, then on
$C_{O}^{s_{2}}$ we have
\begin{equation*}
 \sigma_{AB}= 0 = \overline{T}_{11}\;,\
 \tau = \frac{n-1}{r}
 \;.
\end{equation*}
(This situation will be referred to as that of the Null-Cone
Theorem (NCT).) It is further shown in~\cite{CCG} that, under
suitably stronger energy conditions, equality implies that the
metric is that of the model space on the domain of dependence
of $C_{O}^{s_{2}}$. The proofs of those facts provide a
non-trivial illustration of the formalism developed here, as
specialized to the simpler problem treated in~\cite{CCG}.

\section{Constraints and gauge preservation}
 \label{ConstraintsAndGauge}

The obvious analogue on a null submanifold $M_{0}$ of the
spacelike constraints operator is $\overline{S}_{\alpha\beta}
\ell^{\beta}$, where $\ell$ denotes the field of null normals to
$M_{0}$ normalized in some arbitrary way. Derivatives of the
metric in $\overline{S}_{\alpha\beta}\ell^{\beta}$ transversal
to the light-cone appear only at first order. Some of them
\footnote{Compare~\cite{RendallCIVP,DamourSchmidt} in
space-dimension three.} cancel magically between the various
terms contributing to $\overline{S}_{\alpha\beta}\ell^{\beta}$,
and those that remain can be expressed in terms of
$\overline{H}$. So, in the explicit form of
$\overline{S}_{\alpha\beta} \ell^{\beta} $, one can replace
every occurrence of $\overline{\partial_{0} g_{01}}$,
$\overline{\partial_{0} g_{0A}}$ and
$\overline{g}^{AB}\overline{\partial_{0} g_{AB}}$ by
$\overline{H}_{\alpha}$, $\overline{W}_{\alpha}$, and terms
containing only derivatives along $M_{0}$. We define $n+1$
operators $\mathcal{L}_{\alpha}(\overline{H})$, $\alpha=0,\ldots
,n$, by adding all the terms involving $\overline{H}$ in
$\overline{S}_{\alpha\beta}\ell^{\beta}$.
One can then define $n+1$ operators $\mathcal{C}_{\alpha}$
by whatever remains; thus the $\mathcal{C}_{\alpha}$'s coincide
with $\overline{S}_{\alpha\beta}\ell^{\beta}  $ when
$\overline{H}_{\alpha}$ vanishes. Explicit formulae for
$\mathcal{C}_{\alpha}$ are given in (\ref{C1final}),
(\ref{CAfinal}) and in (\ref{C0_2}) below, while
$\mathcal{L}_{\alpha}$ can be found in (\ref{L1final}),
(\ref{LAfinal}) and (\ref{L0final}).

We will prove the following theorem, which is the key element
of our analysis of the Cauchy problem for the Einstein
equations on a light-cone:

\begin{theorem}
\label{Existence3}
\begin{enumerate}
\item The operator $\overline{S}_{\alpha\beta}\ell^{\beta}$
    on a null submanifold $M_{0}$ can be written as a sum,
\begin{equation*}
 \overline{S}_{\alpha\beta} \ell^{\beta} \equiv
 \mathcal{L}_{\alpha} +\mathcal{C}_{\alpha}
 \;,
\end{equation*}
where $\mathcal{L}_{\alpha} $ vanishes when
$\overline{H}=0$, while the operator $\mathcal{C}_{\alpha}$
depends only on the values $\overline{g}$ on $M_{0}$ of the
spacetime metric, on the choice of the null vector $\ell$,
and on  $\overline{W}$, which depends on the
chosen target space of the wave-map gauge.
The operators $\mathcal{C}_{\alpha}$ will be called
\emph{Einstein-wave-map gauge constraint operators}.

\item In adapted null coordinates:
\begin{enumerate}
\item[a)] The operators $\mathcal{C}_{\alpha}$ lead to
    a hierarchy of ordinary differential operators for
    the coefficients of $\overline{g}$ along the
    generators, all linear when the first constraint
    $\overline{S}_{\alpha\beta}\ell^{\alpha}\ell^{\beta}=
    \overline{T}_{\alpha\beta}\ell^{\alpha}\ell^{\beta}$
    has been solved.

\item[b)] The operators $\mathcal{L}_{\alpha}$ together
    with the wave-gauge reduced Einstein equations lead
    to a hierarchy of homogeneous first order {ordinary
    linear differential operators along the generators}
    for the components $\overline{H}_{\alpha}$ if the
    spacetime metric $g$ satisfies on $M_{0}$ the reduced
    Einstein equations.
\end{enumerate}
\end{enumerate}
\end{theorem}

Theorem~\ref{Existence3} will be proved in
Sections~\ref{constraintC1}-\ref{constraintCA}
and~\ref{constraintC0}.

A consequence of Theorems~\ref{Existence2} and~\ref{Existence3}
is the following:

\begin{theorem}
 \label{Existence4}
A $C^{3}$ Lorentzian metric $g^{(H)}$, solution of the Einstein
equations in wave-map gauge
$S_{\alpha\beta}^{(H)}=T_{\alpha\beta}$ in $Y_{O}^{T}$ with
$\nabla_{\alpha} T^{\alpha\beta}=0$,
and taking an initial value $\overline{g}$
on $C_{O}^{T}$, is a solution of the full Einstein equations
$S_{\alpha\beta} =T_{\alpha\beta}$ if and only if
$\overline{g}$ satisfies the constraints
$\mathcal{C}_{\alpha}=\overline{T}_{\alpha\beta}\ell^{\beta}$.
\end{theorem}

\noindent\begin{proof} Theorem~\ref{Existence3} gives the
following identities, with $\mathcal{L}_{\alpha} $ a linear and
homogeneous first order differential operator along the null
vector $\ell$ for the vector $\overline{H}$,
\begin{equation}
\overline{S}_{\alpha\beta} \ell^{\beta}
\equiv
\overline{S}_{\alpha\beta}^{(H)}\ell^{\beta}
+\frac{1}{2}(\overline{g}_{\alpha\lambda} \overline{\hat{D}_{\beta} H^{\lambda} }
+\overline{g}_{\beta\lambda} \overline{\hat{D}_{\alpha}H^{\lambda} }
-\overline{g}_{\alpha\beta} \overline{\hat{D}_{\lambda} H^{\lambda} }) \ell^{\beta}
\equiv
\mathcal{C}_{\alpha} +\mathcal{L}_{\alpha} .
\label{ConstraintsDecomposition}
\end{equation}

The ``only if" part of the theorem results inmediately from the
identity (\ref{ConstraintsDecomposition}) when the metric $g$
is a solution of the full Einstein equations and is in wave
gauge, since then only $C_{\alpha}$ remains in that identity.

The ``if" part will be proved later by showing that
$\overline{H}^{\alpha}=0$ is the only solution,  for metrics
which are uniformly $C^{1}$ near the tip of the cone, of the
equations
\begin{equation}
\frac{1}{2}(\overline{g}_{\alpha\lambda}\overline{\hat{D}_{\beta}H^{\lambda}}+
\overline{g}_{\beta\lambda}\overline{\hat{D}_{\alpha}H^{\lambda}}-
\overline{g}_{\alpha\beta}\overline{\hat{D}_{\lambda}H^{\lambda}})\ell^{\beta}
=\mathcal{L}_{\alpha}
 \label{Lalpha}%
\end{equation}
which result from the identity (\ref{ConstraintsDecomposition})
when
$\mathcal{C}_{\alpha}=\overline{T}_{\alpha\beta}\ell^{\beta}$ and
$\overline{S}_{\alpha\beta}^{(H)}=\overline{T}_{\alpha\beta}$.
\end{proof}
\bigskip

The question of local geometric uniqueness of solutions is
addressed in Section~\ref{slgu}.

\section{The first constraint}
 \label{constraintC1}

\subsection{Computation of
$\overline{R}_{11}\equiv \overline{S}_{11}\equiv
\overline{S}_{\protect\alpha \protect\beta}
\ell^{\protect\alpha} \ell^{\protect\beta}$}

The component $R_{11}$  can be separated as
\begin{equation*}
R_{11}\equiv R_{11}^{(1)}+R_{11}^{(2)}
\; ,
\label{R11_1}
\end{equation*}
where $R_{11}^{(1)}$ is linear in first derivatives of the Christoffel
symbols and $R_{11}^{(2)}$ is quadratic in them. They are given by,
after a trivial simplification,
\begin{align}
R_{11}^{(1)}\equiv &\
\partial_{0}\Gamma_{11}^{0}+\partial_{A}\Gamma_{11}^{A}-%
\partial_{1}\Gamma_{10}^{0}-\partial_{1}\Gamma_{1A}^{A}
 \;,
\label{R11_2}
\\
R_{11}^{(2)}\equiv &\
\Gamma_{11}^{0}(\Gamma_{00}^{0}+\Gamma_{01}^{1}+\Gamma_{0A}^{A})+
\Gamma_{11}^{1}(\Gamma_{10}^{0}+\Gamma_{11}^{1}+\Gamma_{1A}^{A})+
\Gamma_{11}^{B}(\Gamma_{B0}^{0}+\Gamma_{B1}^{1}+\Gamma_{BA}^{A})
\notag \\
&\
-\Gamma_{10}^{0}\Gamma_{10}^{0}-2\Gamma_{10}^{1}\Gamma_{11}^{0}-
\Gamma_{11}^{1}\Gamma_{11}^{1}-2\Gamma_{1A}^{1}\Gamma_{11}^{A}-
2\Gamma_{1A}^{0}\Gamma_{10}^{A}-\Gamma_{1B}^{A}\Gamma_{1A}^{B}.
\label{R11_3}
\end{align}

We must take care when taking derivatives transversal to the
cone, i.e. $\partial_{0}$, that our coordinates are valid only
on the cone. We will then use the trivial identity
\begin{equation}
\overline{\partial_{\lambda}\Gamma_{\beta\gamma}^{\alpha}}\equiv
\overline{g}^{\alpha\mu}(\overline{\partial_{\lambda}[\mu,\beta\gamma]})
+\overline{(\partial_{\lambda}g^{\alpha\mu})}\overline{[\mu,\beta\gamma]}.
\label{R11_4}
\end{equation}
In $\overline{R}_{11}^{(1)}$ only $\overline{\Gamma}_{11}^{0}$
is differentiated transversally to $C_{a}$. We have, since
$\overline{g}_{11}=\overline{g}_{1A}=0$,
\begin{equation}
\overline{\partial_{0}\Gamma_{11}^{0}}=\frac{1}{2}\nu^{0}\partial_{1}%
\overline{\partial_{0}g_{11}}+(\partial_{1}\nu_{0}-\frac{1}{2}\overline{%
\partial_{0}g_{11}})\overline{\partial_{0}g^{00}},\text{ with \ } \overline{%
\partial_{0}g^{00}}=-(\nu^{0})^{2}\overline{\partial_{0}g_{11}},
\label{R11_5}
\end{equation}
\begin{equation}
-\partial_{1}\overline{\Gamma}_{10}^{0}\equiv -\partial_{1}(\frac{1}{2}\nu^{0}%
\overline{\partial_{0}g_{11}}) .
\label{R11_6}
\end{equation}
By using also
\begin{equation}
\overline{\Gamma}_{11}^{A}=0\;,\quad
\overline{\Gamma}_{1A}^{A}=\frac{1}{2}\overline{g}^{AB}\partial_{1}\overline{g}_{AB}=:\tau
\label{R11_7}
\end{equation}
and the harmonicity function (\ref{Gamma2}), we get
\begin{align}
\overline{R}_{11}^{(1)}\equiv &\
  (\nu^{0})^{2}\frac{1}{2}\overline{\partial_{0}g_{11}}
  \,\overline{\partial_{0}g_{11}}
+ \frac{1}{2}\partial_{1}\nu^{0}\overline{\partial_{0}g_{11}}
-\partial_{1}\tau
\label{R11_8}
\\
\equiv &\ \frac{1}{2}(\overline{\Gamma}_{1}+\tau )^{2}-\frac{1}{2}%
\nu^{0}\partial_{1}\nu_{0}(\overline{\Gamma}_{1}+\tau )-\partial_{1}\tau .
\label{R11_9}
\end{align}

The part $\overline{R}_{11}^{(2)}$ depends only on the values
of the Christoffel symbols on $C_{O}$. Using
$\overline{\Gamma}^{0}_{11}=
\overline{\Gamma}^{A}_{11}=\overline{\Gamma}^{0}_{1A}=0$ and
trivial simplifications we find that
\begin{equation*}
\overline{R}_{11}^{(2)}\equiv \overline{\Gamma}_{11}^{1}(\overline{\Gamma}_{10}^{0}+\overline{%
\Gamma}_{1A}^{A})-\overline{\Gamma}_{10}^{0}\overline{\Gamma}_{10}^{0}-\overline{\Gamma}%
_{1B}^{A}\overline{\Gamma}_{1A}^{B} \;.
\label{R11_10}
\end{equation*}
In the chosen coordinates, $\overline{R}_{11}^{(2)}$ reduces to
\begin{equation}
\overline{R}_{11}^{(2)}\equiv -\frac{1}{2}(\overline{\Gamma}_{1}+\tau )^{2}
+ \frac{1}{2} \nu^{0}\partial_{1}\nu _{0}(\overline{\Gamma}_{1}+\tau )+\nu
^{0}\partial_{1}\nu _{0}\tau -\frac{1}{2}\tau (\overline{\Gamma}_{1}+\tau )
-\chi _{A}{}^{B}\chi _{B}{}^{A} \;.
\label{R11_11}
\end{equation}
Adding (\ref{R11_9}) and (\ref{R11_11}) we obtain
\begin{eqnarray}
\overline{R}_{11}&\equiv& -\partial_{1}\tau +\nu^{0}\partial_{1}\nu_{0}\tau -
\frac{1}{2}\tau (\overline{\Gamma}_{1}+\tau )-\chi_{A}{}^{B}\chi_{B}{}^{A}
\label{R11_13}
\\
 & \equiv &  -\partial_{1}\tau
+ \overline{\Gamma}^{1}_{11}\tau -\chi_{A}{}^{B}\chi_{B}{}^{A}
 \;,
 \label{R11_13x}
\end{eqnarray}

\subsection{The $\protect\mathcal{C}_{1}$ constraint operator}
\label{ss5V.1}

By definition of the wave-gauge vector $H$ we have
$\overline{\Gamma}_{1} \equiv \overline{W}_{1} + \overline{H}_{1}$,
and hence (\ref{R11_13}) decomposes as
\begin{equation}
\overline{R}_{11} \equiv \mathcal{C}_{1} + \mathcal{L}_{1} \; ,
\end{equation}
with
\begin{equation}
\mathcal{C}_{1} := -\partial_{1}\tau
+\left(\nu^{0}\partial_{1}\nu_{0}-\frac{1}{2}(\overline{W}_{1}+\tau)\right)\tau
- |\sigma|^{2} - \frac{\tau^{2}}{n-1} \; ,
\label{C1final}
\end{equation}
where we have separated $\chi_{A}{}^{B}$ in trace-free and pure
trace parts (see (\ref{sigma})), and
\begin{equation}
\mathcal{L}_{1} := -\frac{1}{2}\overline{H}_{1}\tau \; .
\label{L1final}
\end{equation}

As announced (\ref{C1final}) involves only the values of the
metric coefficients on the light-cone; equivalently, no
derivatives of the metric transverse to the light-cone occur
there:
\begin{equation}
\overline{W}_{1}
=
2\hat{\Gamma}^{0}_{01}+\nu_{0}\overline{g}^{AB}\hat{\Gamma}^{0}_{AB}
\stackrel{\eta}=
- \nu_{0} x^{1} \overline{g}^{AB}s_{AB}
\;,
\label{W1Minkowski}
\end{equation}
where we have assumed that the target metric takes the adapted
form \emph{in the same coordinate system}, so that
$\hat{\Gamma}^{0}_{11}=0$ and $\hat{\Gamma}^{0}_{1A}=0$. The
Einstein equation $R_{11}=T_{11}$ in wave-gauge provides,
in this sense, a constraint equation $\mathcal{C}_{1}=\overline{T}_{11}$
for the metric components $\overline{g}_{\mu\nu}$.

The constraint equation $\mathcal{C}_{1}=\overline{T}_{11}$
contains as unknowns only the components $\overline{g}_{AB}$
and $\nu_{0}$ if it is so of $\overline{T}_{11}$. A simple
strategy is then to prescribe $\overline{g}_{AB}$ (compare
Bondi et al.~\cite{BondiEtal62}) and use the definition
(\ref{chi2}) to compute $\chi_{AB}$, hence also $\sigma_{A}{}^{B}$
and $\tau$. The first constraint reads then as a differential
first order equation for $\nu_{0}$, linear if
$\overline{T}_{11}$ is independent of $\nu_{0}$ since
$\overline{W}_{1}$ is linear in $\nu_{0}$. (Recall that we are
assuming $\hat W^{\mu}=0$ unless explicitly indicated
otherwise.) The solution will lead to a Lorentzian metric as
long as $\nu_{0}$ is positive.

However, the equation is singular wherever $\tau$ vanishes, as
the resulting ODE for $\nu_{0}$ involves inverse powers of
$\tau$. For this reason it is of interest  to look for
alternatives, where $\tau$ is computed from the constraint,
rather than prescribed in advance. Following \cite{RendallCIVP}
we will prescribe only the conformal class of $\tilde{g}$. The
wave-map gauge constraint deduced from (\ref{C1final}) is then
an equation for $\nu_{0}$ and the conformal factor
$\Omega^{2}$. We can prescribe arbitrarily $\nu_{0}$ and then
determine $\Omega$. We can also,
generalizing an idea of Damour and Schmidt, impose to $\nu_{0}$
to satisfy a well chosen differential equation containing an
arbitrarily given function $\kappa$. We treat in detail the
case $\kappa=0$, which implies that for the obtained solution
$\nu_{0}$ the vector $\ell$ will be parallelly transported, in
other words $r$ will be an affine parameter, in the resulting
space-time.

\section{Solution of the $\protect \mathcal{C}_{1}$ constraint for
given $\protect \sigma$}
\label{solutionC1}

The operator $\mathcal{C}_{1}$ relates the three functions
$\tau$ (which, via equation (\ref{tau}), essentially describes
the evolution of the volume element of the sections $\Sigma$),
$\nu_{0}$ and $|\sigma|^{2}:=\sigma_{A}{}^{B}\sigma_{B}{}^{A}$. We
recall the following well known fact:

\begin{lemma}
The tensor $\sigma$ is determined by the conformal class of the
induced quadratic form $\tilde{g}$.
\end{lemma}

\begin{proof}
To see this, let us write
\begin{equation*}
\tilde{g}=\Omega^{2}\gamma
\;,
\end{equation*}
with $\gamma$ a degenerate quadratic form on $C_{O}$ such that
$\gamma_{11}\equiv\gamma_{1A}\equiv0$. Then
\begin{equation*}
\chi_{AB}\equiv
\frac{1}{2}\Omega^{2}\partial_{1}\gamma_{AB}
+\gamma_{AB} \Omega\,\partial_{1}\Omega
\;,
\end{equation*}
and thus
\begin{equation}
\chi_{A}{}^{C}\equiv
\frac{1}{2}\gamma^{BC}\partial_{1}\gamma_{AB}
+\delta_{A}^{C}\,\partial_{1}\log\Omega
\;,
\end{equation}
hence the trace-free part of $\chi_{A}{}^{B}$ is
\begin{equation}
\sigma_{A}{}^{C}\equiv
\frac{1}{2}\gamma^{BC}\partial_{1}\gamma_{AB}
-\delta_{A}^{C}\frac{\partial_{1}(\log\sqrt{\det\gamma_{\Sigma}})}{n-1}
\; ,
\label{5VI.5bnew}
\end{equation}
where $\gamma_{\Sigma}$ denotes the positive definite $x^{1}$
dependent quadratic form on $\Sigma$ with components
$\gamma_{AB}$. We see that the traceless tensor $\sigma$ is
independent of the conformal factor, hence depends only on the
conformal class of $\tilde{g}$. In particular $\sigma$ vanishes
if $\tilde{g}$ is conformal to a quadratic form independent of
$r:=x^{1}$.
\end{proof}

If $\gamma$ and its first derivatives satisfy the vertex limits
spelled out for $\tilde{g}$ in section~\ref{ssBcv}, then
$\underset{r\rightarrow 0}{\lim}\,r|\sigma|=0;$ we say that a
degenerate quadratic form $\gamma$ on $C_{O}$, with
$\gamma_{11}\equiv\gamma_{1A}\equiv0$, is admissible if it is
$C^{1}$ on $C_{O}-O$, i.e. for $r>0$, and such that
$|\sigma|^{2}$ is $C^{0}$ for $r\geq0$, hence bounded for
finite $r\geq0$. Given $\sigma$ the constraint
$C_{11}=\overline{T}_{11}$ appears as a relation between the
functions $\tau$ and $\nu_{0}$. Since it involves radial
derivatives of both $\tau$ and $\nu_{0}$ (which can actually be
grouped as $\partial_{1}(\nu^{0}\tau)$) we could prescribe one
of them and integrate for the remaining field, or else provide
an additional differential equation for either of $\tau$ or
$\nu_{0}$ and integrate simultaneously the coupled system of
the constraint and this new equation.  In the remainder of this
section we show how to solve the constraint by prescribing
$\nu_{0}$, either explicitly (section~\ref{prescribenu0}) or
through a differential condition (rest of section
\ref{solutionC1}).

\subsection{Prescribed $\protect \nu_{0}$}
\label{prescribenu0}

Suppose the function $\nu_{0}$ is  arbitrarily  prescribed,
then the constraint equation becomes a differential equation
for $\tau$. It is convenient to introduce the scalar  function
(recall that $\tilde{g}_{\Sigma}$ denotes the restriction of
$\tilde{g}$ to $\Sigma$)
\begin{equation}
 \label{10II.1}
\varphi :=
\left(\frac{\det \tilde{g}_{\Sigma}}{\det s_{n-1}}\right)^{1/(2n-2)} =
\Omega \, \left(\frac{\det\gamma_{\Sigma}}{\det s_{n-1}}\right)^{1/(2n-2)}
\; ,
\end{equation}
so that
\begin{equation}
\tau = (n-1) \, \partial_{1}\log\varphi ,
\quad \text{ or } \quad
\partial_{1} \varphi = \frac{\tau}{n-1}\varphi \; .
\label{tauvarphi}
\end{equation}
The normalization of $\varphi$ has been chosen to have $\varphi=r$
for Minkowski. Using this variable the constraint reads
\begin{equation}
- \partial_{11}^{2}\varphi
+ \left(\nu^{0}\partial_{1}\nu_{0}-\frac{1}{2}\overline{W}_{1}
       -\frac{n-1}{2} \frac{\partial_{1}\varphi}{\varphi}\right)
 \partial_{1}\varphi
=
\frac{|\sigma|^{2}+\overline{T}_{11}}{n-1} \, \varphi
 \;,
\label{R11_13dnew}
\end{equation}
to be integrated outwards with initial data $\varphi(O)=0$ and
$\partial_1\varphi(O)=1$. As already mentioned,
$\overline{W}_{1}$ contains $\varphi$ nonlinearly, and in
principle $\overline{T}_{11}$ could also depend on $\varphi$.
In general this scheme could be considered for a larger class
of gauge
\begin{equation}
 \label{6VI.21new}
 \overline{W}_{1} = \overline{W}_{1}(\gamma_{AB},\varphi,\nu_{0}, r,x^{A})
 \;,
\end{equation}
and sources of the form
\begin{equation}
 \label{6VI.22new}
 \overline{T}_{11} = \overline{T}_{11}(
 \text{source data},
  \gamma_{AB},\partial_{i}\gamma_{AB}, \varphi,\partial_{1}\varphi,\nu_{0}, \partial_{i}\nu_{0},r,x^{A})
 \;, \phantom{xxx}
\end{equation}
where $\partial_{i}$ denotes derivatives tangential to the
light-cone,  and by ``source data" we mean non gravitational
data, for example fields determined from characteristic initial
data for scalar, or Maxwell, fields. The wave-map gauge
condition (\ref{Minkowski6}) is clearly of the form
(\ref{6VI.21new}). In Section~\ref{sss5V.3} we show that both
scalar and Maxwell fields lead to a coefficient
$\overline{T}_{11}$ of the energy-momentum tensor compatible
with (\ref{6VI.22new}).

\subsection{Differential equation for $\protect\nu_{0}$}

The choice made by Rendall is to assume that $x^{1}$ is an
affine parameter along the null rays, in other words that the
vector $\ell = \frac{\partial}{\partial x^{1}}$ is parallelly
transported along the null rays by the connection of the
spacetime he constructs, i.e.
$\overline{\Gamma}^{1}_{11}=0$; equivalently
$\partial_{1}\nu_{0} = (\overline{\Gamma}_{1}+\tau)\nu^{0}/2$.
Now, this last equation contains a derivative transversal to
the light-cone, which is not part of the characteristic initial
data.
Extending
to the cone an idea of Damour and Schmidt\cite{DamourSchmidt}
concerning two intersecting surfaces, in anticipation of the
fact that our solution will satisfy $\overline{H}_{1}=0$, we
could impose to $\nu_{0}$ to satisfy the equation
\begin{equation}
\partial_{1}\nu_{0}=\frac{1}{2}(\overline{W}_{1}+\tau)\nu_{0}
\; ,
\label{nu0equation_kappa0}
\end{equation}
which implies, modulo $\overline{H}_{1}=0$, that
$\overline{\Gamma}_{11}^{1}=0$. When $\nu_{0}$ satisfies
(\ref{nu0equation_kappa0}) the constraint
$\mathcal{C}_{1}=\overline{T}_{11}$
reduces to a Raychaudhuri type equation for the only unknown
$\tau$
\begin{equation}
\partial_{1}\tau +\frac{\tau^{2}}{n-1}+|\sigma|^{2}+\overline{T}_{11}=0
\;.
 \label{9IV10.1}
\end{equation}
More generally all solutions of (\ref{9IV10.1}) can be
obtained by introducing  an arbitrary function $\kappa$ and
solving the pair of equations
\begin{equation}
\partial_{1}\tau -\kappa \tau
+ \frac{\tau^{2}}{n-1}
+ |\sigma|^{2} + \overline{T}_{11}
=0
\;,
 \label{19V10.3}
\end{equation}
whose only unknown is $\tau$ when $|\sigma|^{2}$
and $\overline{T}_{11}$ are known, and
\begin{equation}
\partial_{1}\nu_{0}=\frac{1}{2}(\overline{W}_{1}+\tau)\nu_{0}
+\kappa \nu_{0}
\; .
\label{nu0equation}
\end{equation}
When $\nu_{0}$ satisfies this equation and
$\overline{H}_{1}=0$, then $\overline{\Gamma}_{11}^{1}=\kappa$.

Once $\tau$ is determined we can use (\ref{tauvarphi}) to
obtain $\varphi$ and finally (\ref{nu0equation}) to compute
$\nu_{0}$.

\begin{remark} {\rm
The equations derived here would be dramatically
simplified if one simultaneously imposes $\nu_{0}=1$ and
$\kappa=0$; see, e.g.,~\cite{CJK}. However, these two
conditions together with the wave-gauge condition, which is the
cornerstone of our analysis, would impose undesirable
geometrical restrictions on the initial data.
} \end{remark}

Equation (\ref{19V10.3}) is, if $\overline{T}_{11}$ does not
depend on $\tau$ or $\varphi$, a Riccati differential equation
on each null ray and hence can be rewritten, precisely using
the variable $\varphi$, as a linear second-order equation
\begin{equation}
- \partial_{11}^{2} \varphi + \kappa \, \partial_{1} \varphi =
\frac{|\sigma|^{2}+\overline{T}_{11}}{n-1} \, \varphi \;,
 \label{31III10.2}
\end{equation}
to be integrated outwards with initial values $\varphi(O)$ and
$\partial_{1}\varphi(O)$. We have assumed chosen an admissible
$\gamma$, hence $|\sigma|^{2}$ continuous for $r\ge 0$. We
assume the same holds for $\overline{T}_{11}$. Remark that for
a continuous tensor $T$, i.e. with continuous components
$\underline{T_{\mu\nu}}$ in the coordinates $y^{\alpha}$, we
will have $\overline{T}_{11}$ continuous, but
$\lim_{r\rightarrow 0} \overline{T}_{11}$ a function of angles
in general since it holds that
\begin{equation*}
\lim_{r\to 0} \overline{T}_{11}
= \underline{T_{00}}(0) + 2 \underline{T_{i0}}(0) \frac{y^{i}}{r}
 + \underline{T_{ij}}(0) \frac{y^{i} y^{j}}{r^{2}}
\;.
\end{equation*}

When $|\sigma|^{2}+\overline{T}_{11}$ is continuous for $r\ge
0$ standard ODE theory guarantees that a solution with given
initial values exists globally. However, a positive definite
metric $\overline{g}_{AB}$ is only obtained from the positive
part of the solution. The relevant initial conditions are
$\varphi(O)=0$ and $\partial_{1}\varphi(O)=1$, so $\varphi$ is
initially tangent to $\varphi=r$.

We consider the case $\kappa=0$, that is, $x^{1}$ is an affine
parameter. Assuming $\overline{T}_{11}\ge 0$, the equation
satisfied by $\varphi$ shows that it is a concave function of
$x^{1}$ on each null ray wherever positive, and hence there are
two possibilities: either $\varphi$ is a monotone increasing
function for all real $r$, with $0\le\varphi\le r$ and $0< \tau
< (n-1)/ r$, or else there is a first local maximum, at which
$\partial_{1}\varphi=0$ and hence the expansion $\tau$ also
vanishes there. This is related to the formation of
outer-trapped surfaces on the cone $C_{O}$. Once a maximum has
been reached, $\varphi$ will necessarily vanish for some larger
value of $r$.

We now turn to a direct analysis in terms of $\tau$,
which allows stating results in a more geometric way.

\subsection{Solution of the Raychaudhuri equation}

We continue to use a Minkowski target and we make the choice
$\kappa=0$, so that $x^{1}$ will be an affine parameter along the
null rays.
Equation (\ref{19V10.3}) then reads as a Raychauduri equation
\begin{equation}
 \label{8.1}
\partial_{1}\tau +\frac{1}{n-1}\tau^{2}+|\sigma|^{2}+\overline{T}_{11} = 0
\;.
\end{equation}
This is a first order ODE for $\tau$ when
$|\sigma|^{2}:=\sigma_{A}{}^{B}\sigma_{B}{}^{A}$ and
$\overline{T}_{11}$ are known.

\subsubsection{NCT case}

When $|\sigma|^{2}+\overline{T}_{11}=0$, the equation admits the
solution corresponding to the Minkowskian cone\footnote{It is the only
solution such that ${\tau^{-1}}$ tends to zero at the vertex
of the cone.}:
\begin{equation}
\tau_{\eta}=\frac{n-1}{x^{1}}.\quad
\end{equation}
The value $|\sigma|=0$ further imposes
\begin{equation}
\chi_{A}{}^{B}=\frac{1}{2x^{1}}\delta_{A}^{B},
\text{ \ i.e. \ \ }
\partial_{1}\overline{g}_{AB}=\frac{2}{x^{1}}\overline{g}_{AB}
 \;.
\end{equation}
With our choice of  frame of  coordinates at the vertex, the
solution is the Minkowskian solution
\begin{equation}
\overline{g}_{AB}=(x^{1})^{2}s_{AB}
 \;,
\end{equation}
as used in the null-cone theorem of~\cite{CCG}.

\subsubsection{General case, a global existence theorem}
\label{globalexistenceC1}

We denote $x^{1}$ by $r$ and $\partial_{1}$ by a prime.
In all that follows we write, and solve,
differential equations in $r$, with constant initial values
(mostly zero) for $r=0$. We do not write explicitly the
dependence on the other coordinates $x^{A}$, though it occurs in
the solutions and in the coefficients.

1. In a neighbourhood of $r=0$ we define a new function $y$ by
\begin{equation*}
y:=\frac{n-1}{\tau}
 \;.
\end{equation*}
Equation~(\ref{8.1}) becomes
\begin{equation}
 \label{3V.1}
y^{\prime}=1+\frac{1}{n-1}f^{2}y^{2},
\qquad
f^{2}:=|\sigma|^{2}+\overline{T}_{11}
 \;.
\end{equation}
In agreement with Section~\ref{sssvl2}, we seek
a solution such that $y(0)=0$. The equation implies that $y$ is
increasing and $y\ge r$.

We assume as before that $\frac{1}{n-1}f^{2}$ is continuous and
bounded by a number $A^{2}$. Then $y$ exists, is of class $C^{1}$,
is unique, and is bounded by the solution of the problem
\begin{equation*}
z^{\prime} =1+A^{2}z^{2}, \quad z(0)=0
 \;,
\end{equation*}
as long as that solution exists. The solution is
\begin{equation}
 \label{zsolution}
z=A^{-1}\tan(Ar)
\;.
\end{equation}
Hence $z$ is defined, $C^{\infty} $, and bounded, as well as
all its derivatives, for $0\leq r\leq a$, for any $a<A^{-1}
\frac{\pi} {2}$.

For $0\le r\le a \le A^{-1}$, $z$ is such that
\begin{equation*}
r\le z \le r+A^{2}r^{3} \quad \;.
\end{equation*}
We have defined $\psi$ as
\begin{equation}
 \label{taupsidef}
 \psi:=\tau_{\eta}-\tau , \qquad
 \tau_{\eta}\equiv\frac{n-1}{r},
\end{equation}
and hence we have, since $r<y\leq z$,
\begin{equation}
 \label{zineq}
 0\le\frac{1}{n-1}\psi \equiv \frac{1}{r}-\frac{1}{y}
 \le \frac{1}{r}-\frac{1}{z}
 \le \frac{1}{r}-\frac{1}{r+A^{2}r^{3}}
 = \frac{A^{2} r}{1+A^{2}r^{2}}
 \le A^{2} r \; .
\end{equation}
That is
\begin{equation}
 \label{22XI.41}
0 \leq \psi \leq (n-1)A^{2} r
 \;.
\end{equation}

2. For large $r$ we use the decay of $f^{2}$. Using the
definition (\ref{taupsidef}) we obtain
\begin{equation}
\psi^{\prime} +\frac{2}{r}\psi =\frac{1}{n-1}\psi ^{2}+ f^{2}
 \;.
\end{equation}
This gives
\begin{equation}
u^{\prime} =\frac{1}{n-1}\frac{u^{2}}{r^{2}}+r^{2}f^{2} \ge 0\;,
\qquad u:=r^{2}\psi \;.
\label{1V.1}
\end{equation}
Hence $u$ is an increasing function of $r$, and both $u$ and
$\psi$ are positive as $u(0)=0$.

For $r\geq a$ we replace the problem to solve by the integral equation
\begin{equation}
 \label{30IV.1}
u(r)=\frac{1}{n-1}\int_{a}^{r} \frac{u^{2}(\rho)}{\rho^{2}} d\rho
+\int_{a}^{r}\rho ^{2}f^{2}(\rho )d\rho +u_{a}
 \;,
\end{equation}
with $u_{a}:=u(a)\equiv a^{2}\psi (a)$. By (\ref{zsolution})
and the inequality $y\le z$ for $r\leq a$ we have
\begin{equation*}
\psi(r)\leq \tau_{\eta}(r)-\frac{n-1}{z(r)}=(n-1)\left( \frac{1}{r}-\frac{A}{\tan (Ar)}%
\right)
 \;,
\end{equation*}
hence
\begin{equation}
 \label{30IV.2}
u_{a}\leq (n-1)a\left(1-\frac{Aa}{\tan (Aa)}\right)
 \;.
\end{equation}

We assume that $r^{2}f^{2}$ is integrable for $r\in \lbrack
a,\infty )$, and we set
\begin{equation*}
C_{a}:=u_{a}+B_{a},\qquad B_{a}:=\int_{a}^{\infty} r^{2}f^{2}dr
 \;.
\end{equation*}
The solution $u$ of the integral equation (\ref{30IV.1}) exists
and is bounded by a solution $v$ of the equation
\begin{equation*}
v(r)=\frac{1}{n-1}\int_{a}^{r} \frac{v^{2}(\rho )}{\rho ^{2}} d\rho +C_{a}
 \;,
\end{equation*}
as long as such a solution $v$ exists; equivalently, as long as
the differential equation
\begin{equation*}
v^{\prime} =\frac{1}{n-1}\frac{v^{2}}{r^{2}}
\end{equation*}
admits a solution $v$ with $v(a)=C_{a}$. The general solution of the
above equation is
\begin{equation*}
\frac{1}{v}=\frac{1}{(n-1)r}+c\;,\text{\ i.e. \ } v=\frac{(n-1)r}{1+(n-1)rc}
 \;.
\end{equation*}
It takes the value $C_{a}$ for $r=a$ if and only if
\begin{equation}
 \label{7V.1}
\frac{a(n-1)}{1+(n-1)ac}=C_{a},\text{ \ \ i.e. \ } c=c_{a}:=\frac{1}{C_{a}}-\frac{1}{(n-1)a}.
\end{equation}
The function $v$ remains positive and bounded if
\begin{equation}
1+(n-1)rc_{a}>0
 \;,
\end{equation}
hence $v$ is defined and bounded for all $r$ if $c_{a}\geq 0$, \
i.e. $C_{a}\leq (n-1)a;$ that is, when
\begin{equation}
 \label{30IV.3}
u_{a}+B_{a}\leq (n-1)a \;.
\end{equation}
It follows from (\ref{30IV.2}) that this last inequality will
hold when
\begin{equation*}
B_{a}\leq (n-1)\frac{Aa^{2}}{\tan (Aa)}.
\end{equation*}

In the case where $f^{2}\equiv 0$ for $0\leq r\leq a$ it holds also that
$u(r)\equiv 0$ in that interval we have $u_{a}=0$ and $C_{a}=B_{a}$.
Condition (\ref{30IV.3}) reduces to
\begin{equation}
B_{a}\leq (n-1)a
 \;.
\end{equation}

3. Assume that $\tau\geq0$ exists in the interval $r\in(0,b]$
and denote $\tau_{b}:=\tau(b)$. If for $r\geq b$, $f^{2}=0$,
the equation for $\tau$ reduces, for $r\geq b$, to
\begin{equation}
\tau^{\prime}+\frac{1}{n-1}\tau^{2}=0\;,
\end{equation}
with initial value
\begin{equation}
\tau(b)=\tau_{b},\ \ 0<\tau_{b}\leq\tau_{0}(b)\equiv\frac{n-1}{b}.
\end{equation}
The solution is
\begin{equation}
\frac{1}{\tau}=\frac{r}{n-1}+\frac{1}{\tau_{b}}-\frac{b}{n-1} .
\end{equation}
Therefore
\begin{equation}
\tau=\frac{(n-1)\tau_{b}}{(n-1)+(r-b)\tau_{b}}=\frac{(n-1)}{r+d_{b}},
\end{equation}
with
\begin{equation}
d_{b}:=(n-1)\tau_{b}^{-1}-b\geq 0 .
\end{equation}
Hence, for large $r$,
\begin{equation}
\tau=\frac{n-1}{r}(1-\frac{d_{b}}{r}+\ldots  )
 \;.
\end{equation}

We have proved the following theorem.

\begin{theorem}
 \label{T11V.1}  The equation for $\tau$ deduced from the first
constraint,
\begin{equation*}
\tau^{\prime}+\frac{1}{n-1}\tau^{2}+f^{2}=0
 \;,
\end{equation*}
with $f^{2}:=|\sigma^{2}|+\overline{T}_{11}$ continuous and
$r^{2}f^{2}$ integrable in $r$ for $r\in \lbrack 0,\infty )$,
has a global solution $\tau (r)>0$, and the function
\begin{equation*}
\psi := \frac{n-1}{r} - \tau
\end{equation*}
is of class $C^{1}$ if:

1. We assume that there exists $a\in (0,\infty)$ such that
it holds
\begin{equation}
 \label{Aacondition}
A<\frac{\pi} {2a},\text{ \ \ with \ \ } A^{2}:=\sup_{0\leq r\leq a}\frac{1}{n-1}f^{2}
 \;.
\end{equation}
In the interval $0\leq r\leq a$ it then follows that
\begin{equation*}
\tau \ge \frac{n-1}{r}z(Ar)
 \;,
\end{equation*}
with
\begin{equation}
z(x):=\frac{x}{\tan x} \le 1.
\end{equation}

2. In the interval $a\leq r<\infty $ we assume that
\begin{equation}
 \label{Bacondition}\int_{a}^{\infty
}r^{2}f^{2} \, dr
\leq (n-1) a\, z(Aa)\; .
\end{equation}
In this interval we then have
\begin{equation*}
\tau \ge \frac{n-1}{r + k_a},
\end{equation*}
with
\begin{equation}
k_a=(n-1)(\tau_a-a^{-2}B_a)^{-1}-a, \quad
B_a=\int_a^{\infty} r^{2}f^{2}\,dr .
\end{equation}

3. Regardless of point 1., if $\sigma=0=\overline{T}_{11}$ for $r\geq b$,
and if $\tau_{b}:=\tau(b)>0$, then the solution exists for all
$r\ge b$ and it holds that
\begin{equation*}
\tau=\frac{n-1}{r+k_{b}^{(0)}}, \qquad
k_{b}^{(0)}:=(n-1)\tau_{b}^{-1}-b\ge 0.
\end{equation*}
\end{theorem}

\begin{remark} {\rm
If $f^{2}:=|\sigma|^{2}+\overline{T}_{11}$ has compact support
$\{a\leq r\leq b\}$ with $a>0$, it follows from (\ref{30IV.3})
that (\ref{Bacondition}) can be replaced by
\begin{equation*}
\int_{a}^{b}r^{2}f^{2}dr\leq (n-1)a
 \;,
\end{equation*}
which will be satisfied if, e.g.,
\begin{equation*}
\sup_{a\leq r\leq b}r^{2}f^{2}\leq \frac{(n-1)a}{b-a}.
\end{equation*}
} \end{remark}

\begin{remark} {\rm
It follows from the equations above
(compare~\cite[Proposition~2.2]{CCG}) that if there exists
$r_{2}>0$ such that
\begin{equation}
\label{Y9XII8.1}
  \int_{0}^{r_{2}} \rho^{2} f^{2}(\rho,x^{A})
  d\rho \ge (n-1)r_{2}
 \;,
\end{equation}
then the expansion $\tau(r,x^{A})$ will become negative at some
value of $r$ strictly smaller than $r_{2}$. If this happens for
all $x^{A}$, then one expects existence of an outer trapped
surface in the associated space-time.
(See~\cite{ChrBHF,ReitererTrubowitz,RodnianskiKlainerman:scarred}
for recent important results concerning formation
of trapped surfaces.) }
\end{remark}

\subsection{Determination of $\overline{g}_{AB}$}
 \label{ss5V.1a}

Recall that we have set
\begin{equation}
\overline{g}_{AB}:= \Omega^{2} \gamma_{AB} \equiv
\varphi^{2} \left(\frac{\det s_{n-1}}{\det\gamma_{\Sigma}}\right)^{1/(n-1)}
\gamma_{AB}
\; ,
\label{conformal}
\end{equation}
and that $\varphi$ satisfies the equation
\begin{equation}
\partial_{1}\log\varphi = \frac{\tau}{n-1}
=\frac{1}{r} - \frac{\psi}{n-1}
\; ,
\end{equation}
with the initial condition $\varphi(0)=0$. Its integration gives
\begin{equation}
\varphi(r) = r \exp\left(- \int_{0}^{r} \frac{\psi(\rho)}{n-1}d\rho\right) \;.
\end{equation}

We assume that the free data $\overline{T}_{11}$ and
$\gamma_{AB}$ are such that $\tau$\ exists and satisfies the
conclusions of Theorem~\ref{T11V.1}, with some $a\in (0,\infty)$.
We have then

1. For small $r$, using the inequality (\ref{22XI.41}),
valid for $r<a\le A^{-1}$,
\begin{equation*}
0 \le \psi \le (n-1) A^{2} r
\end{equation*}
we conclude that in such interval we have
\begin{equation*}
\exp\left(-\int_{0}^{r} \frac{\psi(\rho)}{n-1} d\rho \right) \ge
\exp\left(-\frac{1}{2} A^{2} r^{2}\right),
\end{equation*}
and therefore
\begin{equation}
0 \le r -\varphi(r)
  \le r \left(1 - \exp\left(-\frac{1}{2}A^{2}r^{2}\right)\right)
  \le \frac{1}{2} A^{2} r^{3}
\; .
\end{equation}

2. For $r\ge a$, let $\psi$ be as in (\ref{taupsidef}), we use
\begin{equation}
\psi\equiv\frac{u}{r^{2}}\leq\frac{v}{r^{2}}\equiv\frac{(n-1)}%
{r\{1+(n-1)rc_{a}\}}
\end{equation}
to obtain
\begin{equation}
\partial_{1}\log\varphi=\frac{1}{r}-\frac{\psi}{n-1}\geq
\frac{(n-1)c_{a}}{1+(n-1)rc_{a}}
 \;.
\end{equation}
This shows that $\varphi$ is an increasing $C^{1}$ function bounded below by
\begin{equation}
\varphi(a)\frac{1+(n-1)rc_{a}}{1+(n-1)ac_{a}} .
\end{equation}

3. In the case where one assumes that $f^{2}=0$ for $r>b$
it holds exactly
\begin{equation}
\partial_{1}\log\varphi=\frac{1}{r+d_{b}}>0\;,\text{ \ \ \ } %
d_{b}=(n-1)\tau_{b}^{-1}-b\geq0\text{ \ .}%
\end{equation}
Therefore, using the notation $\varphi_{b}:=\varphi(b)$,
\begin{equation}
\varphi(r)=\varphi_{b}\frac{r+d_{b}}{b+d_{b}}>\varphi_{b},
\text{ \ if \ } r>b\; .
\end{equation}

In conclusion if $\gamma_{AB}$ is admissible, and
$\overline{T}_{11}$ is known and continuous we can solve
(\ref{9IV10.1}) for $\tau$ on some maximal
(possibly angle-dependent) interval of $r$'s. Subsequently,
(\ref{tauvarphi}) can be solved with initial value
$\varphi(0)=0$. This provides $\overline{g}_{AB}$.  The
quantity $\varphi$, hence also $\Omega$, depend only on the
conformal class of $\gamma$; the same is true of $\tilde{g}$,
defined by (\ref{conformal}).

\subsection{Determination of $\protect\nu_{0}$}
 \label{sss5VI.4a}

Once $\overline{g}_{AB}$ is known we can integrate equation
(\ref{nu0equation}) for $\nu_{0}$,
with the initial condition $\nu_{0}|_{r=0}=1$,
\begin{equation*}
\partial_{1}\nu_{0} = \frac{1}{2}(\overline{W}_{1}+\tau)\nu_{0} + \kappa\nu_{0}
\; .
\end{equation*}
(Note that at this stage any ``wave-gauge source"
$\overline{W}_{1}$ of the form
\begin{equation}
 \overline{W}_{1} (\tau, \overline{g}_{AB},\nu_{0},r,x^{A})
 \label{5VI.9}
\end{equation}
with an appropriate behaviour near $r=0$ could be used, though
as said before, in this section we assume $\kappa=0$, and a
Minkowski target). The equation for $\nu_{0}$ reads
\begin{equation}
\frac{\partial_{1}\nu_{0}} {\nu_{0}^{2}}=\frac{1}{2}\{-\overline{g}^{AB}rs_{AB}%
+\frac{\tau}{\nu_{0}}\}
 \;,
\end{equation}
i.e., since $\nu^{0}=\frac{1}{\nu_{0}}$,
\begin{equation}
\partial_{1}\nu^{0}=-\frac{1}{2}\tau\nu^{0}+\frac{1}{2}\overline{g}^{AB}rs_{AB}
 \;.
 \label{5VI.13}
\end{equation}
This is a linear equation for $\nu^{0}$,
with coefficients singular for $r=0$, but continuous
for $r>0$. Its solution taking a given initial value for $r_{0}>0$
exists, is $C^{1}$ and unique for $r\ge r_{0}$ as long as $\tau$
and $\Omega^{-1}$ exist and are continuous.
Note, however, that $\nu^{0}$ could go to zero in finite affine time,
which would lead to a (geometric or coordinate) singularity.

\subsubsection{NCT case}

To study solutions with initial data at $r=0$, we
start with the NCT case. We have then
$\tau\equiv\tau_{\eta}= \frac{n-1}{r}$, and (\ref{5VI.13}) reduces to
\begin{equation}
\partial_{1} \nu^{0}=-\frac{n-1}{2r}(\nu^{0}-1)
 \;.
\end{equation}
The general solution is, for some constant $k$,
\begin{equation}
\nu^{0}-1=kr^{-\frac{n-1}{2}}
 \;.
\end{equation}
The solution tending to one as $r$ tends to zero corresponds to
$k=0$, and is $\nu^{0}=\nu_{0}=1$.

\subsubsection{General case}

To construct a solution tending to 1 as $r$ tends to zero we set
\begin{equation}
 \label{30VI.11}
Y:=1-\nu^{0}
 \;.
\end{equation}
The equation (\ref{5VI.13}) for $\nu^{0}$ becomes the linear
non homogeneous equation
\begin{equation}
Y^{\prime}=-\frac{1}{2}\tau Y+F
 \;,
  \label{6VI.41}
\end{equation}
with $F$ a continuous function (recall
the notations $\tau\equiv\tau_{\eta}-\psi$,
$\tau_{\eta}=\frac{n-1}{r}=\eta^{AB}rs_{AB}$
and the assumed boundary conditions
(\ref{7VI.36}) and (\ref{7VI.40}))
\begin{equation}
 \label{23VII.1x}
 F :=\frac{1}{2}(\tau-\overline{g}^{AB}rs_{AB})\equiv
\frac{1}{2}\{(\eta^{AB} -\overline{g}^{AB})rs_{AB}-\psi\} =
 O(r)
 \;,
\end{equation}
where $\psi\geq0$ and
\begin{equation}
(\eta^{AB}-\overline{g}^{AB})rs_{AB}=
\frac{n-1}{r}- r \Omega^{-2} \gamma^{AB}s_{AB}.
\end{equation}
Incidentally, this implies that $F\leq0$
for initial data such that
\begin{equation}  \label{10II.2}
r^{2} \Omega^{-2}
\gamma^{AB}s_{AB} \geq n-1
 \;.
\end{equation}
Now, in the notation of (\ref{10II.1}), this
can be rewritten in the form
\begin{equation}
r^{2} \Omega^{-2} \gamma^{AB}s_{AB} =
r^{2} \varphi^{-2}
\left[(\det\gamma_{\Sigma})^{1/(n-1)} \gamma^{AB}\right]
\left[(\det s_{n-1})^{-1/(n-1)} s_{AB}\right]
 \;,
\end{equation}
such that the two expressions in square brackets have unit
determinants. Using $\varphi\le r$, hence
$r^{2}\varphi^{-2}\geq 1$, the last equation allows one to
deduce (\ref{10II.2}) from a condition involving only the
conformal metric $\gamma_{AB}$.

We want to find a solution $Y$ which tends to zero with $r$,
but this solution will lead to data for a Lorentzian metric
only if $\nu^{0}$ remains bounded and non zero; that
is, if $Y<1$.

The homogeneous equation associated to (\ref{6VI.41}) is
\begin{equation}
Y^{\prime}=(-\frac{n-1}{2r}+\frac{1}{2}\psi)Y
 \;.
\end{equation}
Setting $Y=\exp Z$, this equation reads
\begin{equation}
Z^{\prime}=-\frac{n-1}{2r}+\frac{1}{2}\psi.
\end{equation}
The general solution of (\ref{6VI.41}) is of the form
\begin{equation}
Y=w\exp Z\;,\text{ \ \ with \ \ } w^{\prime}=\exp(-Z)F\;.
\end{equation}

1. Case $0\leq r\leq a$.

Without loss of generality we can choose
\begin{equation}
Z=-\frac{n-1}{2}\log r+\frac{1}{2}\int_{0}^{r}\psi(\rho)d\rho
 \;,
\end{equation}
hence
\begin{equation}
Y(r)=wr^{-\frac{n-1}{2}}
\exp\Big(\frac{1}{2}\int_{0}^{r}\psi(\rho)d\rho\Big),
\text{ \ with \ \ }
w^{\prime}=r^{\frac{n-1}{2}}
\exp\Big(-\frac{1}{2}\int_{0}^{r}\psi(\rho)d\rho\Big)F(r).
\end{equation}
We find the solution $Y:=1-\nu^{0}$ tending to zero with $r$
{(compare~(\ref{7VI.34}))} by integrating $w^{\prime}$ between
$0$ and $r$, it gives
\begin{equation}
 Y(r)
  =
  r^{-\frac{n-1}{2}}
  \exp\Big(\frac{1}{2}\int_{0}^{r}\psi(\rho)d\rho\Big)
  \int_{0}^{r}\rho^{\frac{n-1}{2}}
     \exp\Big(-\frac{1}{2}\int_{0}^{\rho}\psi(\chi)d\chi\Big)
     F(\rho)d\rho
\; .
 \label{30VI.12}
\end{equation}

Keeping in mind  that there exist numbers $C_a$, a=1,2, such
that
\begin{equation*}
0\leq\psi\leq C_{1}\,r\;, \ |F(r)|\leq C_{2}\,r
 \;,
\end{equation*}
we see that for $0\leq r\leq a$ we have
\begin{equation*}
\exp\Big(\frac{1}{2}\int_{0}^{r}\psi(\rho)d\rho\Big)\leq
\exp\Big(\frac{1}{4}C_{1}r^{2}\Big),
\end{equation*}
\begin{equation*}
\int_{0}^{r}\rho^{\frac{n-1}{2}}
\exp\Big(-\frac{1}{2}\int_{0}^{\rho}\psi(\chi)d\chi\Big)|F(\rho)|d\rho
\leq C_{2}\int_{0}^{r}\rho^{\frac{n+1}{2}}d\rho
 =\frac{2}{n+3}C_{2}\,r^{\frac{n+3}{2}}\;,
\end{equation*}
leading to the bound, still for $r\leq a$,
\begin{equation*}
|Y|\leq \frac{2C_{2}r^{2}}{n+3}\exp\Big(\frac{1}{4}C_{1}r^{2}\Big)
 \;.
\end{equation*}
Since $Y:=1-\nu^{0}$ the function $\nu^{0}$ is bounded. From
(\ref{11VI.A7}) the metric will have Lorentzian signature,
$\overline{g}_{AB}$ being Riemannian, if and only if $\nu^{0}$
remains bounded and non zero (hence positive since equal to 1
for $r=0$). This will hold if $Y<1$, which will be true
for any $C_{2}$ if $a$ is small enough. In vacuum $C_{2}$ is
determined by $|\sigma|^{2}$, so for any $a$ it will hold that
$Y<1$ for $r\in[0,a]$ if $\sigma$ is small enough.

Note that if $F\leq0$, then $Y\leq0$, hence $\nu^{0}\ge 1$ without
restriction on the size of $a$ or $|\sigma|$.

2. $a\le r < \infty$.

By the same reasoning as for $r\leq a$, the solution $Y$ taking
the value $Y(a)$ for $r=a$ is
\begin{equation*}
Y(r)=Y(a) +
r^{-\frac{n-1}{2}}
\exp\Big(\frac{1}{2}\int_{a}^{r}\psi(\rho)d\rho\Big)
\int_{a}^{r}\rho^{\frac{n-1}{2}}
  \exp\Big(-\frac{1}{2}\int_{a}^{\rho}
  \psi(\chi)d\chi\Big) F(\rho)d\rho
\; .
\end{equation*}

3. Suppose that for $r\geq b$ we have
$\overline{g}^{AB}=\varphi^{-2} s^{AB}$, hence
\begin{equation*}
F(r)\equiv\frac{1}{2}\{(s^{AB}(r^{-2}-\varphi^{-2})rs_{AB}-\psi\}\equiv
\frac{1}{2}\{(1-r^{2}\varphi^{-2})\frac{n-1}{r}-\psi\}
 \;.
\end{equation*}
We have seen that $\varphi\le r$, that is $r^{2}\varphi^{-2}\ge 1$,
hence $F(r)\leq0$ and $\nu^{0}(r)\geq\nu^{0}(b)$.

\bigskip

\subsection{Vanishing of $\protect \overline{H}_{1}$}
 \label{ss5V.5}

Consider a solution of the wave-map-reduced Einstein equations
$\overline{R}_{11}^{(H)}=\overline{T}_{11}$ with initial data
on $C_{O}$, and with Minkowski target. Suppose that the
data there satisfy the constraint
$\mathcal{C}_{1}=\overline{T}_{11}$. The identity  (see
(\ref{RicciHIdentity}))
\begin{equation*}
\overline{R}_{11}
\stackrel{\eta}{\equiv}
\overline{R}_{11}^{(H)}+\nu_{0}D_{1}\overline{H}^{0}
\end{equation*}
shows that $\overline{H}^{0}$ satisfies a linear homogeneous
differential equation on $C_{O}$, namely,
\begin{equation}
 \label{19V.1}
D_{1}\overline{H}^{0}+\frac{1}{2}\overline{H}^{0}\tau
\stackrel{\eta}{=}
0\;.
\end{equation}
Keeping in mind that $D$ is the covariant derivative of the
Minkowski metric, in our adapted coordinate system we have
\begin{equation*}
D_{1}\overline{H}^{0}\stackrel{\eta}{\equiv} \partial_{1}\overline{H}^{0}.
\end{equation*}

For all solutions which satisfy uniform $C^{1}$
bounds near the vertex in the $(y^{\mu})$ coordinate system,
the $y^{\mu}$-components of the wave-gauge vector are bounded
near the vertex. It follows that $\overline{H}^{0}$ is bounded
near the vertex. But every solution of (\ref{19V.1}) which is
not identically zero behaves, for small $r$, as $r^{-(n-1)/2}$
along some generators. So, in the uniformly $C^{1}$ case, we can
deduce from (\ref{19V.1}) that
\begin{equation*}
\overline{H}^{0}=0\;, \quad \text{hence also} \quad
\overline{H}_{1}\equiv \nu _{0}\overline{H}^{0}=0
 \;.
\end{equation*}

\begin{remark} {\rm
If we add constraint damping terms as in
(\ref{DampingRicci}),  we obtain instead
\begin{equation}
\mathcal{L}_{1} = \left(-\frac{1}{2} \tau + \epsilon n_{1}\right)  \overline{H}_{1}
\; .
\label{damping1}
\end{equation}
No term proportional to $\overline{H}^{A}$ or
$\overline{H}^{1}$ appears, and hence the damping term is
compatible with this first step of the  wave-map-gauge
constraint hierarchy. The new term does not change the terms
which are singular in $r$ in (\ref{19V.1}), and hence
$\overline{H}_{1}=0$ is still the only solution with the
required behavior. } \end{remark}

\subsection{Scalar and Maxwell fields}
 \label{sss5V.3}

We wish to check that scalar fields lead to equations
compatible with the required hierarchical structure of the
equations. For this, consider a scalar field $\phi$ coupled
with the gravitational field through an energy-momentum tensor
of the form
$$
 T_{\mu\nu} = \partial_{\mu} \phi\partial_{\nu} \phi
  -  \left(\frac{1}{2} g^{\alpha\beta} \partial_{\alpha}\phi
     \partial_{\beta}\phi + V(\phi)\right) g_{\alpha\beta}
   \;.
$$
In the adapted coordinate system the components relevant for
our argument are
\begin{eqnarray}
\overline{T}_{11}  & =   & (\partial_{1} \overline{\phi} )^{2}
\label{Tscalar11}
\;, \\
\overline{T}_{A1} &=   & \partial_{A}\overline{\phi}\, \partial_{1}\overline{\phi}
\label{TscalarA1}
 \;, \\
\overline{T}_{01}  & =   &
- \nu_{0} \frac{1}{2} \overline{g}^{AB}\partial_{A}\overline{\phi}\, \partial_{B}\overline{\phi}
+ \nu^{A} \partial_{A}\overline{\phi}\, \partial_{1}\overline{\phi}
- \nu_{0} \frac{1}{2} \overline{g}^{11} (\partial_{1}\overline{\phi})^{2}
 \nonumber
\\
 && - V(\phi)\nu_{0}
  \;.
\label{Tscalar01}
\end{eqnarray}
Keeping in mind that the initial data for the scalar field on
$C_{O}$ are provided by $\overline{\phi}:= \phi|_{C_{O}}$, we see that
prescribing  $\overline{\phi}$ provides a $\overline{T}_{11}$ which can be
used in (\ref{R11_13dnew}) or in (\ref{19V10.3}) (compare
(\ref{6VI.22new})).

Next, the relevant components of the stress-energy tensor for
the Maxwell field $F_{\mu\nu}$ are:
\begin{align}
\overline{T}_{11} \ = \ & \ \overline{g}^{AB} \overline{F}_{A1} \overline{F}_{B1}
\label{TMaxwell11}
\;,
\\
 \overline{T}_{A1} \ = \ & -\nu^{0} \overline{F}_{A1} \overline{F}_{01} -
 \overline{g}^{BC}\overline{F}_{AB}\overline{F}_{C1} + \nu^{0}\nu^{B}
 \overline{F}_{A1}\overline{F}_{B1}
 \label{TMaxwellA1}\;,
\\
\overline{T}_{01} \ = \ &
- \frac{1}{4}\nu_{0} \overline{g}^{AC}\overline{g}^{BD}\overline{F}_{AB} \overline{F}_{CD}
- \overline{g}^{BC}\nu^{A}\overline{F}_{AB}\overline{F}_{C1}
+ \frac{1}{2}\nu^{0} \nu^{A} \nu^{B} \overline{F}_{A1} \overline{F}_{B1}
\notag \\ &
- \frac{1}{2}\nu_{0} \overline{g}^{AB} \overline{g}^{11} \overline{F}_{A1} \overline{F}_{B1}
- \frac{1}{2}\nu^{0} (\overline{F}_{01})^{2} \;.
\label{TMaxwell01}
\end{align}
We defer a complete discussion of the Cauchy problem for the
Einstein-Maxwell equations to separate work. Here we note that
if $F_{1A}$ is given on the null cone, then (\ref{TMaxwell11})
is \emph{not} of the right form for viewing (\ref{19V10.3}) as
a first order equation for $\tau$: Instead (\ref{19V10.3})
should be considered as a second order equation for $\varphi$,
using (\ref{tauvarphi}). On the other hand, (\ref{TMaxwell11}) is
of the form (\ref{6VI.22new}), needed for the analysis of the
problem when $\nu_{0}$ has been given. The remainder of our
analysis of the $\mathcal{C}_{1}$ constraint goes through as
before.

For further reference, we note that the combination of
stress-energy components appearing in the final  constraint
$\mathcal{C}_{0}$ is
\begin{equation}
\overline{g}^{11}\overline{T}_{11}+2\overline{g}^{A1}\overline{T}_{A1}+2\overline{g}^{01}\overline{T}_{01} =
-\frac{1}{2}\overline{g}^{AC}\overline{g}^{BD}\overline{F}_{AB}\overline{F}_{CD}
-(\overline{g}^{01}\overline{F}_{01}+\overline{g}^{A1}\overline{F}_{A1})^{2}
\end{equation}
for the Maxwell field, and
\begin{equation}
\overline{g}^{11}\overline{T}_{11}+2\overline{g}^{A1}\overline{T}_{A1}+2\overline{g}^{01}\overline{T}_{01} =
-\overline{g}^{AB}\partial_{A}\overline{\phi}\,\partial_{B}\overline{\phi}
\end{equation}
for the scalar field.

\section{The $\mathcal{C}_{A}$  constraint}
 \label{constraintCA}

The $\mathcal{C}_{A}$  wave-map-gauge constraint  operator will
be obtained from an analysis of
\begin{equation}
\overline{S}_{1A}\equiv\overline{R}_{1A}\equiv
\overline{R}_{1A}^{(1)}+\overline{R}_{1A}^{(2)}
 \;,
 \label{CAdef}
\end{equation}
where we have again separated terms including derivatives of Christoffels
in $\overline{R}_{1A}^{(1)}$ from the rest in $\overline{R}_{1A}^{(2)}$.
Trivial simplification gives
\begin{equation}
\overline{R}_{1A}^{(1)}\equiv
   \overline{\partial_{0}\Gamma_{1A}^{0}}
  +\partial_{B}\overline{\Gamma}_{1A}^{B}
  -\partial_{1}\overline{\Gamma}_{A0}^{0}
  -\partial_{1}\overline{\Gamma}_{AB}^{B}
 \;.
\label{R1A_1}
\end{equation}
We have by the choice of coordinates
$\overline{\Gamma}_{1A}^{0}\equiv
\overline{\Gamma}_{11}^{B}\equiv
\overline{\Gamma}_{11}^{0}\equiv 0$, and therefore
\begin{align}
\overline{R}_{1A}^{(2)}\equiv &\
\overline{\Gamma}_{1A}^{1}(\overline{\Gamma}_{10}^{0}+\overline{
\Gamma}_{1B}^{B})+\overline{\Gamma}_{1A}^{B}(\overline{\Gamma}_{B0}^{0}+\overline{\Gamma}
_{B1}^{1}+\overline{\Gamma}_{BC}^{C})
\notag \\
& -\overline{\Gamma}_{10}^{0}\overline{\Gamma}_{A0}^{0}-\overline{\Gamma}_{10}^{B}\overline{\Gamma}
_{AB}^{0}-\overline{\Gamma}_{1B}^{1}\overline{\Gamma}_{A1}^{B}-\overline{\Gamma}_{1C}^{B}
\overline{\Gamma}_{AB}^{C} .
\label{R1A_2}
\end{align}

We find, for the terms in $\overline{R}_{1A}^{(1)}$,
\begin{align}
\overline{\partial_{0}\Gamma_{1A}^{0}} = &\
\overline{\partial_{0}g^{00}[0,1A]}+\overline{\partial_{0}g^{0B}[B,1A]}
+\frac{1}{2}\nu ^{0}\overline{\partial_{0}\partial_{A}g_{11}}
\label{R1A_3}
\\
= &\ \frac{1}{2} \partial_{A}(\nu^{0}\overline{\partial_{0}g_{11}})
 -\nu^{0}\chi_{A}{}^{B}\overline{\partial_{0}g_{1B}}
\notag \\
&
+\frac{1}{2}(\nu^{0})^{2}
(\overline{\partial_{0}g_{1A}}-\partial_{1}\nu_{A}+2\nu_{B}\chi_{A}{}^{B})
\overline{\partial_{0}g_{11}}
\; ,
\label{R1A_4} \\
-\partial_{1}\overline{\Gamma}_{A0}^{0}\equiv&\ -\frac{1}{2}\partial_{1}\{\nu
^{0}(\partial_{A}\nu_{0}+\overline{\partial_{0}g_{1A}}-\partial_{1}\nu _{A})\} ,
\label{R1A_5} \\
\partial_{B}\overline{\Gamma}_{1A}^{B}-\partial_{1}\overline{\Gamma}_{AB}^{B}\equiv
&\
\partial_{B}\chi _{A}{}^{B}-\partial_{1}(\nu ^{0}\nu _{B}\chi _{A}{}^{B})-\partial_{A}\tau .
\label{R1A_6}
\end{align}
And for the terms in $\overline{R}_{1A}^{(2)}$ we find
\begin{equation}
\overline{\Gamma}_{1A}^{1}\overline{\Gamma}_{1B}^{B}\equiv \{\frac{1}{2}\nu
^{0}(\partial_{1}\nu _{A}+\partial_{A}\nu _{0}-\overline{\partial
_{0}g_{1A}})-\nu ^{0}\nu _{B}\chi _{A}{}^{B}\}\tau
\; ,
\label{R1A_7}
\end{equation}
\begin{equation}
\overline{\Gamma}_{1A}^{B}
(\overline{\Gamma}_{B0}^{0}+\overline{\Gamma}_{B1}^{1}+\overline{\Gamma}_{BC}^{C})
\equiv
\chi_{A}{}^{B}(\nu^{0}\partial_{B}\nu _{0}+\tilde{\Gamma}_{BC}^{C})
\; ,
\label{R1A_8}
\end{equation}
\begin{equation}
-\overline{\Gamma}_{10}^{0}\overline{\Gamma}_{A0}^{0}\equiv
\frac{1}{4} (\nu^{0})^{2}
(\partial_{1}\nu_{A}-\partial_{A}\nu_{0}-\overline{\partial_{0}g_{1A}})
\overline{\partial_{0}g_{11}}
\; ,
\label{R1A_9}
\end{equation}
\begin{equation}
\overline{\Gamma}_{A1}^{1}\overline{\Gamma}_{10}^{0}
\equiv
\frac{1}{4}(\nu^{0})^{2}
(\partial_{1}\nu_{A}+\partial_{A}\nu_{0}-\overline{\partial_{0}g_{1A}}
-2\nu_{B}\chi_{A}{}^{B}) \overline{\partial_{0}g_{11}}
\; ,
\label{R1A_10}
\end{equation}
\begin{equation}
-\overline{\Gamma}_{10}^{B}\overline{\Gamma}_{AB}^{0}
\equiv
\frac{1}{2}\nu^{0}\chi_{A}{}^{B}
(\overline{\partial_{0}g_{1B}}+\partial_{1}\nu_{B}-\partial_{B}\nu_{0}
 - \nu^{0}\nu_{B}\overline{\partial_{0}g_{11}})
\; ,
\label{R1A_11}
\end{equation}
\begin{equation}
 -\overline{\Gamma}_{1B}^{1}\overline{\Gamma}_{A1}^{B}
 -\overline{\Gamma}_{1C}^{B}\overline{\Gamma}_{AB}^{C}
\equiv
 \frac{1}{2}\nu^{0}\chi_{A}{}^{B}
  (\overline{\partial_{0}g_{1B}}-\partial_{1}\nu_{B}-\partial_{B}\nu_{0})
 -\chi_{B}{}^{C}\tilde{\Gamma}_{AC}^{B}
\; .
\label{R1A_12}
\end{equation}
All terms in these formulae can be computed on $C_{O}$, except
for those that contain $\overline{\partial_{0}g_{1B}}$ or
$\overline{\partial_{0}g_{11}}$, and whose sum simplifies to
\begin{equation}
\overline{R}_{1A,\partial_{0}}=
-\frac{1}{2}\partial_{1}(\nu^{0}\overline{\partial_{0}g_{1A}})
-\frac{1}{2}\tau\,\nu^{0}\overline{\partial_{0}g_{1A}}
+\frac{1}{2}\partial_{A}(\nu^{0}\overline{\partial_{0}g_{11}})
\; .
\label{R1A_14}
\end{equation}
(We see that all terms quadratic in $\partial_{0}$
derivatives cancel out.) The rest is given by
\begin{align}
\overline{R}_{1A}-\overline{R}_{1A,\partial_{0}}
\equiv &\
\frac{1}{2}\nu^{0}\partial_{1}(
\partial_{1}\nu_{A}-\partial_{A}\nu_{0}-2\nu_{B}\chi_{A}{}^{B})
+\tilde{\nabla}_{B}\chi _{A}{}^{B}
-\nu_{0}\partial_{A}(\tau\nu^{0})
\notag \\
& +\frac{1}{2}(\partial_{1}\nu^{0}+\tau\nu^{0})
(\partial_{1}\nu_{A}-\partial_{A}\nu_{0}-2\nu_{B}\chi_{A}{}^{B}) .
\label{R1A_15}
\end{align}

\subsection{Use of harmonicity functions}
 \label{ss6VI.1}

{}From the identities (\ref{Gamma2}) and (\ref{Gamma4}) we get
\begin{align}
\overline{\Gamma}_{A}:= \overline{g}_{AB}\overline{\Gamma}^{B}\equiv &\
-\nu_{A}\overline{\Gamma}^{0}+\nu^{0}
(\overline{\partial_{0}g_{1A}}+\partial_{1}\nu_{A}-\partial_{A}\nu_{0})
-2\nu^{0}\nu_{B}\chi_{A}{}^{B}+\tilde{\Gamma}_{A}.
\label{GammaA_3}
\end{align}
Hence
\begin{equation}
\overline{\partial_{0}g_{1A}}\equiv -\partial_{1}\nu_{A}+\partial_{A}\nu
_{0}+2\nu_{B}\chi_{A}{}^{B}+\nu_{A}\overline{\Gamma}_{1}+\nu_{0}(\overline{\Gamma}_{A}-
\tilde{\Gamma}_{A}) .
\label{GammaA_5}
\end{equation}
On the other hand
\begin{equation}
\overline{\Gamma}_{A}\equiv \overline{H}_{A}+\overline{W}_{A} , \text{ \ with \ \ \ }
\overline{H}_{A}:= \overline{g}_{AB}\overline{H}^{B} ,
\label{GammabarA}
\end{equation}
similarly for $\overline{\Gamma}_{A} $ and $\overline{W}_{A}$. Therefore we
have
\begin{equation}
\overline{\partial_{0}g_{1A}}\equiv
-\partial_{1}\nu_{A}+\partial_{A}\nu_{0}+2\nu_{B}\chi_{A}{}^{B}
+\nu_{0}f_{A}+\nu_{0}\overline{H}_{A} +\nu_{A} \overline{H}_{1}
\; ,
\label{Harmonicity2}
\end{equation}
with
\begin{equation}
f_{A}:= \nu^{0}\nu_{A}\overline{W}_{1}+\overline{W}_{A}-\tilde{\Gamma}_{A} ,
\label{fVector}
\end{equation}
For a Minkowski target, using (\ref{Minkowski6}) and
(\ref{Minkowski7}), this is
\begin{equation}
f_{A}\stackrel{\eta}{=}
-\left(x^{1}\overline{g}^{CD}s_{CD}+\frac{2\nu^{0}}{x^{1}}\right)\nu_{A}
+ \overline{g}_{AB}\overline{g}^{CD}(S_{CD}^{B}-\tilde{\Gamma}_{CD}^{B})
\; .
\end{equation}

\subsection{Computation of $\protect\mathcal{L}_{A}$ and
$\protect\mathcal{C}_{A}$}

We see from the identities obtained that $\overline{R}_{1A}$ is
the sum of a linear homogeneous operator $\mathcal{L}_{A}$ on
$\overline{H}_{A}:=\overline{g}_{AB}\overline{H}^{B}$ and a
second order linear operator $\mathcal{C}_{A}$ on $ \nu_{A}$,
both with coefficients depending only on the $x^{1}$-dependent
metric $\tilde{g}$ and scalar $\nu_{0}$ previously determined.
(Strictly speaking, $\nu_{A}$ also appears in $\mathcal{L}_{A}$,
but multiplied by $\overline{H}^{0}$ which,  with appropriate
boundary conditions and a Minkowski target, can be  shown to be
zero at this stage of the argument, as explained above.)
\begin{equation}
\mathcal{L}_{A}\equiv
-\frac{1}{2}\partial_{1}(\overline{H}_{A}+\nu_{A}\overline{H}^{0})
-\frac{1}{2}\tau(\overline{H}_{A}+\nu_{A}\overline{H}^{0})
+\frac{1}{2}\partial_{A}(\nu_{0}\overline{H}^{0})
\; .
\label{LAfinal}
\end{equation}

{}From (\ref{R1A_14}),
\begin{equation*}
\overline{R}_{1A,\partial_{0}}=
-\frac{1}{2}\nu^{0}\partial_{1}\overline{\partial_{0}g_{1A}}
-\frac{1}{2}(\partial_{1}\nu^{0}+\tau\nu^{0})\overline{\partial_{0}g_{1A}}
+\frac{1}{2}\partial_{A}(\nu^{0}\overline{\partial_{0}g_{11}})
\; ,
\end{equation*}
and using the formula (\ref{Harmonicity2}) we find that
\begin{align}
\overline{R}_{1A,\partial_{0}}-\mathcal{L}_{A}\equiv
-\frac{1}{2}\nu^{0}\partial_{1}
\{-\partial_{1}\nu_{A}+\partial_{A}\nu_{0}+2\nu_{B}\chi_{A}{}^{B}
+\nu_{0}f_{A}\}
\notag \\
-\frac{1}{2}(\partial_{1}\nu^{0}+\tau\nu^{0})
\{-\partial_{1}\nu_{A}+\partial_{A}\nu_{0}+2\nu_{B}\chi_{A}{}^{B}
+\nu_{0}f_{A}\}
\notag \\
+\frac{1}{2}\partial_{A}(\overline{W}_{1}+\tau)
\;  .
\label{CA_14}
\end{align}
Finally, assembling results (\ref{R1A_15}) and (\ref{CA_14}) gives
\begin{align}
\mathcal{C}_{A}\equiv &\ \overline{R}_{1A}-\mathcal{L}_{A}\equiv
\ \frac{1}{2}\nu^{0} \partial_{1} \{
   2\partial_{1}\nu_{A} -4\nu_{B}\chi_{A}{}^{B}
   -\nu_{0}f_{A} \}
\notag\\ &
+\frac{1}{2}(\partial_{1}\nu^{0}+\tau\nu^{0})\{
   2\partial_{1}\nu_{A} -4\nu_{B}\chi_{A}{}^{B}
   -\nu_{0}f_{A} \}
\notag \\ &
+ \tilde{\nabla}_{B}\chi_{A}{}^{B}
-\frac{1}{2}\partial_{A}(\tau-\overline{W}_{1}+2\nu^{0}\partial_{1}\nu_{0}) .
\label{CAfinal}
\end{align}
It turns out that there is a simple way of rewriting
(\ref{CAfinal}) in terms of
\begin{equation}
\xi_{A} :=
- 2 \nu^{0}\partial_{1}\nu_{A} + 4 \nu^{0}\nu_{B}\chi_{A}{}^{B}
+ \underbrace{\nu^{0}\nu_{A}\overline{W}_{1}+\overline{W}_{A}-\tilde{\Gamma}_{A}}_{\equiv f_{A}}
\; ,
\label{xiA}
\end{equation}
where $\tilde{\Gamma}_{A}:= \overline{g}_{AD} \overline{g}^{BC}
\tilde\Gamma^D_{BC}
$ (compare (\ref{22XII.1})). The vector
$\xi_{A}$ equals $-2\overline{\Gamma}^{1}_{1A}$ after using the
harmonicity conditions $\overline{\Gamma} = \overline{W}$. Note that
$\xi_{A}$ vanishes when $\nu_{A}=0$ and
$\overline{g}_{AB}=\hat{g}_{AB}$. The  wave-map-gauge
constraint operator $\mathcal{C}_{A}$ can be expressed in terms
of $\xi_{A}$ as
\begin{equation}
\mathcal{C}_{A} = -\frac{1}{2}(\partial_{1}\xi_{A}+\tau\xi_{A})
+ \tilde{\nabla}_{B}\chi_{A}{}^{B}
-\frac{1}{2}\partial_{A}(\tau-\overline{W}_{1}+2\nu^{0}\partial_{1}\nu_{0})
\; .
\label{CAfinal2}
\end{equation}

Separating different orders of $\partial_{1}$ derivatives we
get, for a Minkowski target,
\begin{align}
\mathcal{C}_{A}\stackrel{\eta}{=}&\
\nu^{0} \partial^{2}_{11} \nu_{A}
-2\nu^{0} \chi_{A}{}^{B}\partial_{1}\nu_{B}
+\nu^{0}\left(\tau+\frac{1}{x^{1}}- \frac{1}{2}\overline{W}_{1}
-\nu^{0}\partial_{1}\nu_{0}\right) \partial_{1}\nu_{A}
\notag \\
& -\nu^{0}\left(2\partial_{1}\chi_{A}{}^{B}
+2(\nu_{0}\partial_{1}\nu^{0}+\tau)\chi_{A}{}^{B}
\right)\nu_{B}
+\tilde{\nabla}_{B}\chi_{A}{}^{B}
\notag \\
& -\left(\frac{1}{2}\partial_{1}\overline{W}_{1}
+\frac{1}{2}(\nu_{0}\partial_{1}\nu^{0}+\tau)(\overline{W}_{1}-\frac{2}{x^{1}})
+\frac{1}{(x^{1})^{2}}\right)
\nu^{0}\nu_{A}
\notag \\
&
-\partial_{A}\left(\frac{1}{2}\tau+\nu^{0}\partial_{1}\nu_{0}
-\frac{1}{2}\overline{W}_{1}\right)
-\frac{1}{2} (\partial_{1}+\tau)
(\overline{g}_{AB}\overline{g}^{CD}S^{B}_{CD}-\tilde{\Gamma}_{A})
\; .
\label{CAfinal3}
\end{align}

In the general case, in addition to
(\ref{6VI.21new})-(\ref{6VI.22new}) one can assume that
\begin{eqnarray}
 \label{8VI.21}
 \overline{W}_{A} & = & \overline{W}_{A}(\gamma_{AB},  \varphi,  \nu_{0}, \nu_{A}, r,x^{A})
 \;,
\\
 \label{8VI.22}
 \overline{T}_{1A} & = &
 \overline{T}_{1A}( \text{source data},\gamma_{AB},\partial_{i} \gamma_{AB}, \varphi, \partial_{i} \varphi,
 \nu_{0}, \partial_{i}\nu_{0},
 \nu_{A}, \partial_{1} \nu_{A}, \overline{\partial_{0} g_{11}}
  , r,x^{A})
 \;, \phantom{xxxxx}
\end{eqnarray}
where as before $\partial_{i}$ denotes derivatives tangential to
the light-cone. This is clearly compatible with the wave-map
gauge (\ref{Minkowski7}), and with scalar fields or Maxwell
fields as sources (compare Section~\ref{sss5V.3}).

\section{Solution of the $\protect\mathcal{C}_{A}$ constraint}
 \label{solutionCA}

\subsection{NCT case}

In the vacuum case with Minkowski target and when
$\sigma_{A}{}^{B}=0$ we have
$\chi_{A}{}^{B}=r^{-1}\delta_{A}^{B}$,
$\overline{g}_{AB}=r^{2}s_{AB}$ (therefore
$\tilde{\Gamma}_{BC}^{A}=S_{BC}^{A}$), $\nu_{0}=1$ and
$\tau=-\overline{W}_{1}=\frac{n-1}{r}$. It has been shown in
\cite{CCG} that the $\mathcal{C}_{A}$  wave-map-gauge
constraint reduces then to
\begin{equation}
\mathcal{C}_{A}\equiv\overline{R}_{1A}-\mathcal{L}_{A}\equiv
\partial_{11}^{2}\nu_{A}+\frac{3n-5}{2r}\partial_{1}\nu_{A}
+\frac{1}{2}\frac{(n-2)(n-3)}{r^{2}} \nu_{A}=0\;.
\label{CA_NCT}
\end{equation}
This is a Fuchsian type linear equation, with Fuchsian
exponents $p=\frac {3-n}{2}$ and $2-n$. Thus the only solution
satisfying (\ref{gvertexlimits}), i.e. $\lim_{r\to
0}(r^{-1}\nu_{A})=0$, is $\nu_{A}\equiv 0$. (In fact, the only
solution $\nu_{A}=o(r^{\frac{3-n}{2}})$ is zero.)

\subsection{General case}

From the identity (\ref{CAfinal2}) we then find

\begin{lemma}
 \label{lCA}
Assuming (\ref{6VI.21new})-(\ref{6VI.22new}) and
(\ref{8VI.21})-(\ref{8VI.22}), the  wave-map-gauge constraint
operator $\mathcal{C}_{A} \equiv \overline{S}_{1A}-{\cal L}_{A}$
is a first order linear ordinary differential operator for the
field $\xi_{A}$, with $\kappa$ as in (\ref{nu0equation})
\begin{equation}
\mathcal{C}_{A}\equiv-\frac{1}{2}(\partial_{1}\xi_{A}+\tau\xi_{A}
)+\tilde{\nabla}_{B}\sigma_{A}{}^{B}-\frac{n-2}{n-1}\partial_{A}\tau
-\partial_{A}\kappa
 \;,
 \label{5VI.32}
\end{equation}
where $\xi_{A}$ is defined as (\ref{xiA}), which particularizes for a
Minkowski target as
\begin{equation}
\xi_{A} \stackrel{\eta}{:=}
 -2\nu^{0}\partial_{1}\nu_{A}+4\nu^{0}\nu_{B}\chi_{A}{}^{B}
+\left(\overline{W}^{0}-\frac{2}{r}\nu^{0}\right)\nu_{A}+
\gamma_{AB}\gamma^{CD}(S_{CD}^{B}-\tilde{\Gamma}_{CD}^{B})
 \;.
 \label{xiAbis}
\end{equation}
\end{lemma}

Anticipating, we note that $\nu_{A}$ will also appear in the
last  wave-map-gauge constraint $\mathcal{C}_{0}$ through
$\xi_{A}$ only.

If one assumes that $\overline{T}_{1A}$ is known (e.g., in vacuum, or
for scalar fields, compare (\ref{TscalarA1})), the homogeneous
part of the equation $\mathcal{C}_{A}=\overline{T}_{1A}$ reads
$$-\frac{1}{2}(\partial_{1}\xi_{A}+\tau\xi_{A})=0
 \;,
$$
and admits as general solution, keeping in mind that
$\tau\equiv(n-1)\partial_{1}\log\varphi$,
\begin{equation}
\xi_{A}={\check\xi}_{A}\varphi^{-({n-1})}\; ,
 \label{11VI.9}
\end{equation}
for some vector field on the sphere $ {\check\xi}_{A}$. The
final solution $\xi_{A}$ is of the form (\ref{11VI.9}), with
$\check\xi_{A} $ obtained by integrating the following equation
for ${\check\xi}_{A}$
\begin{equation}
 \partial_{1} {\check\xi}_{A}=
2\varphi^{n-1}\{
 \tilde{\nabla}_{B}\sigma_{A}{}^{B}
-\frac{n-2}{n-1}\partial_{A}\tau
-\partial_{A}\kappa
-\overline{T}_{1A}
\}
 \;,
 \label{7VI.1}
\end{equation}
with the boundary condition ${\check\xi}_{A}=0$,  deduced from
the finiteness of
$$ \lim_{r\to 0}\xi_{A} = \lim_{r\to 0}
 r^{-n+1}\check{\xi}_{A}
$$
(compare (\ref{7VI.34})-(\ref{7VI.37}), and(\ref{7VI.40})). The
field $\nu_{A}$ is then obtained by integrating (\ref{xiAbis}),
with the boundary condition $\nu_{A}=0$ at $r=0$.
These equations are a first order linear system
of ODEs with coefficients singular for $r=0$.

In the NCT case we have $\xi_{A}\equiv0$ and
$\overline{W}^{0}\overset{\eta}{\equiv}-x^{1}\overline{g}^{AB}s_{AB}$,
and hence
\begin{equation}
2\partial_{1}\nu_{A}+\frac{n-3}{r}\nu_{A}=0
 \;,
\end{equation}
whose general solution is
\begin{equation*}
\nu_{A}=k_{A}r^{-\frac{n-3}{2}}
\end{equation*}
with $k_{A}$ independent of $x^{1}\equiv r$. The solution
tending to zero with $r$, and compatible with the
boundary condition (\ref{7VI.35}) for $n\ge 2$, is
$\nu_{A}=0$. In the general case, but with Minkowski target,
the equation (\ref{xiAbis}) reads
\begin{equation*}
2\partial_{1}\nu_{A}+\frac{n-3}{r}\nu_{A}-4\nu_{C}\sigma_{A}{}^{C}+
\{r\nu_{0}\overline{g}^{AB}s_{AB}-\frac{n-1}{r}+\frac{4\psi}{n-1}\}\nu_{A}
\overset{\eta}{=}
\end{equation*}
\begin{equation}
\nu_{0}\gamma_{AB}\gamma^{CD}(S_{CD}^{B}
-\tilde{\Gamma}_{CD}^{B})-\nu_{0} \xi_{A}
\; .
\end{equation}
Setting
\begin{equation*}
\nu_{A}=k_{A}r^{-\frac{n-3}{2}}
\end{equation*}
gives for $k_{A}$, with $k_{A}$ which must tend to zero with
$r$, a differential system with coefficients continuous and
right hand side tending to zero like $r^{\frac{n-3}{2}}$,
\begin{equation*}
\partial_{1}k_{A}
\overset{\eta}{=}
2k_{C}\sigma_{A}{}^{C}+\lambda k_{A}+\mu_{A}
\end{equation*}
with
\begin{eqnarray}
\lambda &:=& -\frac{1}{2}\{r\nu_{0}\overline{g}^{AB}s_{AB}
-\frac{n-1}{r}+\frac{4\psi}{n-1}\} ,
 \nonumber \\
\mu_{A} &:=& \frac{1}{2}r^{\frac{n-3}{2}}\{\nu_{0}\gamma_{AB}\gamma^{CD}
(S_{CD}^{B}-\tilde{\Gamma}_{CD}^{B})-\nu_{0}\xi_{A}\} .
 \nonumber
\end{eqnarray}
Such a system can be solved by iterated integration starting from
$k_{A}^{(0)}=0$, %
\begin{equation*}
k_{A}^{(p)}=\int_{0}^{r}\{2k_{C}^{(p-1)}\sigma_{A}{}^{C}+\lambda k_{A}%
^{(p-1)}+\mu_{A}\}(\rho)d\rho
\end{equation*}
The convergence and bound $|k_{A}|\leq Cr^{\frac{n+1}{2}}$
results from the bounds of $\sigma$, $\lambda$ and $\mu$. In
conclusion, in vacuum, and in the wave-map gauge, the solution
of (\ref{7VI.1}) exists as long as $\tau$ does, with
$\nu_{A}\in C^{1}$ and $|\nu_{A}|\leq Cr^{2}$.

\subsection{Vanishing of $\overline{H}_{A}$}

The general identity
(\ref{RicciHIdentity}) gives in our coordinates
\begin{equation}
\mathcal{C}_{A}+\mathcal{L}_{A}\equiv
\overline{R}_{1A}\equiv
\overline{R}_{1A}^{(H)}+
\frac{1}{2}(\nu_{0}\overline{\hat{D}_{A}H^{0}}+
        \nu_{A}\overline{\hat{D}_{1}H^{0}}+
        \overline{g}_{AB}\overline{\hat{D}_{1}H^{B}}) ,
\label{HApreserve1}
\end{equation}
with $\hat{D}$ the covariant derivative of the target metric,
which in this subsection will be chosen to be the Minkowski
metric, and hence $\hat{D}=D$. Therefore, if a metric solves
$R_{1A}^{(H)}=T_{1A}$ and $\mathcal{C}_{A}=\overline{T}_{1A}$,
we will have, {taking $\overline{H}^{0}=0$ (which, for
sufficiently regular solutions, and for a Minkowski target, has
been justified in Section~\ref{ss5V.5}) on the left hand side},
\begin{equation}
\mathcal{L}_{A}
= -\frac{1}{2}\{\partial_{1}\overline{H}_{A}+\tau \overline{H}_{A}\}
\stackrel{\eta}{=}
  \frac{1}{2}(\nu_{0}\overline{D_{A}H^{0}}
             +\nu_{A}\overline{D_{1}H^{0}}
             +\overline{g}_{AB}\overline{D_{1}H^{B}}),
\label{HApreserve2}
\end{equation}
For Minkowski target these derivatives are, still for
$\overline{H}^{0}=0$,
\begin{align}
\overline{D_{A}H^{0}}
\stackrel{\eta}{=} &\ \hat{\Gamma}_{AB}^{0}\overline{H}^{B}
\stackrel{\eta}{=} -rs_{AB}\overline{H}^{B}, \qquad
\overline{D_{1}H^{0}} \stackrel{\eta}{=} 0\;,
\label{HApreserve3}
\\
\overline{D_{1}H^{B}} \stackrel{\eta}{=} &\
\partial_{1}\overline{H}^{B}+\hat{\Gamma}_{1C}^{B}\overline{H}^{C}
\stackrel{\eta}{=}
\overline{g}^{BC}\partial_{1}\overline{H}_{C}+\overline{H}_{C}\overline{g}^{CD}
(\frac{1}{r}\delta_{D}^{B}-2\chi_{D}{}^{B}) .
\label{HApreserve4}
\end{align}

Therefore (\ref{HApreserve2}) reads
\begin{equation}
2\partial_{1}\overline{H}_{A}+\tau \overline{H}_{A}
\stackrel{\eta}{=}
\{\nu_{0}rs_{AB}\overline{g}^{BC}-\frac{1}{r}\delta_{A}^{C}
   +2 \chi_{A}{}^{C}\}\overline{H}_{C}
 \;.
\label{HApreserve5}
\end{equation}

Taking leading orders in $r$ near the vertex, as given in
Section~\ref{ssBcv}, we find
\begin{equation}
\partial_{1}\overline{H}_{A} +
\left(\frac{n-3}{2r} \delta_{A}^{B} + O_{A}^{B}\right)\overline{H}_{B}
\stackrel{\eta}{=} 0
\;,
\label{HApreserve6}
\end{equation}
where the $O_{A}^{B}$ are $O(r)$ functions. Hence
\begin{equation}
\overline{H}_{A}=r^{-\frac{n-3}{2}}k_{A}\;,
\qquad \text{with} \qquad
\partial_{1}k_{A}+O_{A}^{B}k_{B}=0 \; .
\end{equation}
Standard ODE arguments show that $\overline{H}_{A}= 0$ is the
only solution of (\ref{HApreserve5}) such that
$\overline{H}_{A}=O(r)$, which is the case for metrics having
uniform $C^{1}$ estimates at the vertex.

\begin{remark} {\rm
If we add constraint damping terms as in
(\ref{DampingRicci}) we obtain instead, using again $\overline{H}^{0}=0$,
\begin{equation}
\mathcal{L}_{A} = -\frac{1}{2}\partial_{1}\overline{H}_{A}
+\frac{1}{2}(-\tau+\epsilon n_{1}) \overline{H}_{A}
\; .
\label{dampingA}
\end{equation}
No term proportional to $\overline{H}^{1}$ appears, and hence the constraint
damping term is compatible with this second step of the constraint
hierarchy.
The new term does not change the leading orders in $r$ of equation
(\ref{HApreserve6}) and hence $\overline{H}_{A}=0$ is still the only
regular solution.
} \end{remark}

\section{The $\mathcal{C}_{0}$  constraint}
 \label{constraintC0}

We compute $S_{01}\equiv S_{0\alpha} \ell^{\alpha}$ on
$C_{O}$. We have
\begin{equation}
S_{01}:=R_{01}-\frac{1}{2}g_{01}R
 \;,
\label{S01_1}
\end{equation}
hence in our coordinates
\begin{equation}
\overline{S}_{01}\equiv -\frac{1}{2}\nu _{0}\overline{g}^{AB}\overline{R}_{AB}+\overline{R}%
_{1A}\nu ^{A}-\frac{1}{2}\nu _{0}\overline{g}^{11}\overline{R}_{11} \;.
\label{S01_2}
\end{equation}
We write
\begin{equation}
\overline{R}_{AB}:=\overline{R}_{AB}^{(1)}+\overline{R}_{AB}^{(2)} \;,
\label{RAB}
\end{equation}
with
\begin{equation}
\overline{R}_{AB}^{(1)}:=\overline{\partial_{\alpha} \Gamma_{AB}^{\alpha}}
- \partial_{A}\overline{\Gamma}_{B\alpha}^{\alpha} , \qquad
\overline{R}_{AB}^{(2)}:=
\overline{\Gamma}_{AB}^{\alpha} \overline{\Gamma}_{\alpha\lambda}^{\lambda}
-\overline{\Gamma}_{A\beta}^{\alpha} \overline{\Gamma}_{B\alpha}^{\beta} .
\label{RAB_12}
\end{equation}

We will see  that the $\mathcal{C}_{0}$ wave-map-gauge
constraint is obtained, like the other constraints, by
decomposing the term in $\overline{S}_{01}$ which has not
already been computed, $\overline{g}^{AB}\overline{R}_{AB}$,
into terms defined by data of the degenerate metric on the cone
and terms which vanish when harmonicity conditions are
satisfied on the cone.

Equations (\ref{Gamma2}) and (\ref{Gamma6}) allow us to express
the transversal derivative
$\overline{g}^{AB}\overline{\partial_{0}g_{AB}}$ in terms of
harmonicity functions,
\begin{align}
\frac{1}{2}\nu^{0}\overline{g}^{AB}\overline{\partial_{0}g_{AB}}
\equiv &\ - \overline{\Gamma}^{1}
+(\nu^{0})^{2}\partial_{1}\overline{g}_{00}
+\overline{g}^{11}(\nu^{0}\partial_{1}\nu_{0}-\frac{1}{2}\tau
+\frac{1}{2}\overline{\Gamma}_{1} )
\notag \\ &
+2(\nu^{0})^{2}\nu^{A}(-\partial_{1}\nu_{A}+\nu_{B}\chi_{A}{}^{B})
+\nu^{0}\overline{g}^{AB}\tilde{\nabla}_{B}\nu_{A}
\label{Harmonicity3} \\
\equiv & - \overline{\Gamma}^{1}
-\partial_{1}\overline{g}^{11} + \overline{g}^{11}
( -\nu^{0}\partial_{1}\nu_{0} -\frac{1}{2}\tau +
   \frac{1}{2}\overline{\Gamma}_{1})
+\nu^{0}\overline{g}^{AB}\tilde{\nabla}_{B}\nu_{A}.
\label{27V.2}
\end{align}

\subsection{Computation of $\overline{g}^{AB}\overline{R}_{AB}^{(1)}$}
\label{gABRAB1}

We have
\begin{equation}
\overline{g}^{AB}\overline{R}_{AB}^{(1)}:=\overline{g}^{AB}\{\overline{\partial_{0}%
\Gamma_{AB}^{0}}+\partial_{1}\overline{\Gamma}_{AB}^{1}+\partial_{C}\overline{\Gamma} %
_{AB}^{C}-\partial_{A}\overline{\Gamma}_{B\alpha}^{\alpha}\}
\label{gR1_1}
\end{equation}

To compute we proceed in a straightforward way, using the
values of the Christoffel symbols of the first kind and
elementary algebraic relations in our coordinates on the cone.
Equations (\ref{null9}) and (\ref{null10}) of
Appendix~\ref{A7VI.1} are useful for the calculations that
follow.

We set
\begin{equation}
\overline{g}^{AB}\overline{\partial_{0}\Gamma_{AB}^{0}}\equiv I_{1}+II_{1},
\label{gR1_4}
\end{equation}
with
\begin{align}
I_{1}:=&\
\overline{g}^{AB}\nu^{0}\overline{\partial_{0}[1,AB]}\equiv -\frac{1}{2}
\overline{g}^{AB}\nu^{0}\partial_{1}\overline{\partial_{0}g_{AB}}
+\overline{g}^{AB}\nu^{0}\partial_{A}\overline{\partial_{0}g_{1B}},
\label{gR1_5}
\\
II_{1}:=&\
\overline{g}^{AB}\overline{\partial_{0}g^{0\alpha}}
\overline{[\alpha,AB]}\equiv
(\nu^{0})^{2}\overline{\partial_{0}g_{11}}\overline{g}^{AB}(\frac{1}{2}
\overline{\partial_{0}g_{AB}}-\partial_{A}\nu _{B})-\overline{\partial
_{0}g^{01}}\tau +\overline{\partial_{0}g^{0C}}\tilde{\Gamma}_{C},
\label{gR1_6}
\end{align}
and
\begin{align}
\overline{\partial_{0}g^{01}} = &\
-(\nu^{0})^{2}\overline{\partial_{0}g_{01}}
-\nu^{0}\overline{g}^{11} \overline{\partial_{0}g_{11}}
+(\nu^{0})^{2}\nu^{C}\overline{\partial _{0}g_{1C}} , \\
\overline{\partial_{0}g^{0C}} = &\
(\nu^{0})^2\nu^{C}\overline{\partial_{0}g_{11}}
-\nu^{0}\overline{g}^{CA} \overline{\partial_{0}g_{1A}} .
\end{align}
Grouping terms gives
\begin{align}
\overline{g}^{AB}\overline{\partial_{0}\Gamma _{AB}^{0}}\equiv &\
-\frac{1}{2}\nu ^{0}\overline{g}^{AB}\partial_{1}\overline{\partial_{0}g_{AB}}%
+\nu ^{0}\overline{g}^{AB}\tilde{\nabla}_{A}\overline{\partial_{0}g_{1B}}
 \notag
 \\
 &
 +(\nu ^{0})^{2}\overline{\partial_{0}g_{11}}\{\frac{1}{2}\overline{g}^{AB}%
 \overline{\partial_{0}g_{AB}}-\tilde{\nabla}_{A}\nu ^{A}+\nu _{0}\tau \overline{g}^{11}\}
 \notag
\\
 &
 +(\nu ^{0})^{2}\overline{\partial_{0}g_{01}}\tau
 -(\nu ^{0})^{2}\nu^{A}\overline{\partial_{0}g_{1A}}\tau
 \;.
\label{gR1_8}
\end{align}

We now separate
\begin{equation}
\overline{g}^{AB}\overline{\partial_{1}\Gamma_{AB}^{1}}\equiv III_{1}+IV_{1} ,
\label{gR1_10}
\end{equation}
with
\begin{align}
III_{1}\equiv &\
\overline{g}^{AB}\{
    \nu^{0}\partial_{1}\overline{[0,AB]}
   +\overline{g}^{11} \partial_{1}\overline{[1,AB]}
   +\overline{g}^{1C}\partial_{1}\overline{[C,AB]}
\}
\nonumber
\\
= &\ -\frac{1}{2}\overline{g}^{AB}\nu^{0}\partial_{1}\overline{\partial
_{0}g_{AB}}+\overline{g}^{AB}\nu^{0}\partial_{1}\partial_{A}\nu_{B}
-\overline{g}^{11}\overline{g}^{AB}\partial_{1}\chi_{AB}-\nu^{0}\nu^{C}\overline{g}
^{AB}\partial_{1}\overline{[C,AB]} ,
\label{gR1_12}
\\
IV_{1}:= &\
\overline{g}^{AB}\partial_{1}\overline{g}^{1\alpha}\overline{[\alpha,AB]}
\notag \\
\equiv &\
\overline{g}^{AB}\partial_{1}\nu^{0}
(\partial_{A}\nu_{B}-\frac{1}{2}\overline{\partial_{0}g_{AB}})
-\tau\partial_{1}\overline{g}^{11}
-\overline{g}^{AB}\partial_{1}(\nu^{0}\nu^{C})\overline{[C,AB]}.
\label{gR1_13}
\end{align}
Grouping terms gives
\begin{align}
\overline{g}^{AB}\overline{\partial_{1}\Gamma_{AB}^{1}}\equiv &\
-\frac{1}{2}\partial_{1}(\nu^{0}\overline{g}^{AB}\overline{\partial_{0}g_{AB}}%
)-\nu^{0}\chi^{AB}\overline{\partial_{0}g_{AB}}+\overline{g}^{AB}\partial_{1}(%
\nu^{0}\tilde{\nabla}_{A}\nu_{B})
\notag \\ &
-\overline{g}^{11}\overline{g}^{AB}\partial_{1}\chi_{AB}-\tau\partial_{1}g^{11} .
\label{gR1_14}
\end{align}

Finally we have
\begin{align}
\overline{g}^{AB}\partial_{C}\overline{\Gamma}_{AB}^{C}\equiv &\
\overline{g}^{AB}\nu ^{C}\nu
^{0}\partial_{C}\chi _{AB}+\tau \partial_{C}(\nu ^{C}\nu ^{0})
+\overline{g}^{AB}\partial_{C}\tilde{\Gamma}_{AB}^{C} ,
\label{gR1_15}
\\
-\overline{g}^{AB}\partial_{A}\overline{\Gamma}_{B\alpha} ^{\alpha}  \equiv &\
-\overline{g}^{AB}\partial_{AB}^{2}(\log \sqrt{\det g})\equiv
-\overline{g}^{AB}\partial_{AB}^{2}\{\log (\nu _{0}\sqrt{\det \tilde{g}})\}
\notag \\
\equiv &\
-\overline{g}^{AB}\{\partial_{A}\nu ^{0}\partial_{B}\nu _{0}+\nu
^{0}\partial_{AB}^{2}\nu _{0}+\partial_{A}\tilde{\Gamma}_{BC}^{C}\}.
\label{gR1_16}
\end{align}

\subsection{Computation of $\overline{g}^{AB}\overline{R}_{AB}^{(2)}$}
\label{gABRAB2}

We set
\begin{equation}
g^{AB}\overline{R}_{AB}^{(2)}:=\overline{g}^{AB}\{\overline{\Gamma}_{AB}^{\alpha} \overline{%
\Gamma}_{\alpha \beta} ^{\beta} -\overline{\Gamma}_{A\beta} ^{\alpha} \overline{\Gamma}%
_{B\alpha} ^{\beta} \}\equiv (I_{2}+II_{2}+III_{2}+IV_{2}+V_{2}+VI_{2}) ,
\label{gR2_1}
\end{equation}
with
\begin{equation}
I_{2}:=\overline{g}^{AB}
    \overline{\Gamma}_{AB}^{0}\overline{\Gamma}_{0\beta}^{\beta} ,
\qquad
II_{2}:=\overline{g}^{AB}
    \overline{\Gamma}_{AB}^{1}\overline{\Gamma}_{1\beta}^{\beta} .
\label{gR2_2}
\end{equation}
We find by straightforward computation
\begin{eqnarray}
I_{2} &\equiv&
-\nu^{0}\tau \{
   \nu^{0}\overline{\partial_{0}g_{01}}
  +\frac{1}{2}\overline{g}^{11}\overline{\partial_{0}g_{11}}
  -\nu^{0}\nu^{A}\overline{\partial_{0}g_{1A}}
  +\frac{1}{2}\overline{g}^{AB}\overline{\partial_{0}g_{AB}}
\} ,
\label{gR2_3} \\
II_{2} &\equiv&
\{ \nu^{0}\tilde{\nabla}_{A}\nu^{A}
  -\frac{1}{2}\nu ^{0}\overline{g}^{AB}\overline{\partial_{0}g_{AB}}
  -\overline{g}^{11}\tau
\} (\nu^{0}\partial_{1}\nu_{0}+\tau) .
\label{gR2_4}
\end{eqnarray}
Then we have, recalling that $\tilde{\Gamma}$ denotes Christoffel
symbols of the metric $\tilde{g}$,%
\begin{equation}
III_{2}:=\overline{g}^{AB}
    \overline{\Gamma}_{AB}^{C}\overline{\Gamma}_{C\beta}^{\beta}
\equiv
(\nu^{0}\nu^{C}\tau +\overline{g}^{AB}\tilde{\Gamma}_{AB}^{C})
(\tilde{\Gamma}_{CD}^{D}+\nu^{0}\partial_{C}\nu_{0})\;,
\label{gR2_5}
\end{equation}
Next,
\begin{equation}
IV_{2}:=-\overline{g}^{AB}\{\overline{\Gamma}_{A0}^{0}\overline{\Gamma}_{B0}^{0}+\overline{\Gamma%
}_{A1}^{1}\overline{\Gamma}_{B1}^{1}\}
 \;.
\label{gR2_6}
\end{equation}
Furthermore,
\begin{equation}
\overline{\Gamma}_{1A}^{1}\equiv -\overline{\Gamma}_{0A}^{0}
-\nu^{0}\nu_{B}\chi_{A}{}^{B}+\nu^{0}\partial_{A}\nu_{0} ,
\label{gR2_7}
\end{equation}
with ($\zeta_{A}$ is sometimes called the torsion form)
\begin{equation}
\overline{\Gamma}_{0A}^{0}\equiv \frac{1}{2}\nu^{0}\zeta_{A},
\qquad
\zeta_{A}:=\overline{\partial_{0}g_{1A}}+\partial_{A}\nu_{0}-\partial_{1}\nu_{A}.
\label{gR2_8}
\end{equation}
Hence, using (\ref{Harmonicity2}),
\begin{equation}
\zeta_{A}=
2\partial_{A}\nu_{0}-2\partial_{1}\nu_{A}
+2\nu_{B}\chi_{A}{}^{B}+\nu_{0}(f_{A}+\overline{H}_{A}+\nu_{A}\overline{H}^{0}) .
\label{gR2_9}
\end{equation}
In terms of this object we have
\begin{equation}
 IV_{2}\equiv -\overline{g}^{AB}(\nu ^{0})^{2}\left\{\frac{1}{2}\zeta _{A}\zeta
_{B}+\zeta _{A}(\nu _{C}\chi _{B}{}^{C}-\partial_{B}\nu _{0})+(\nu _{C}\chi _{A}{}^{C}
-\partial_{A}\nu _{0})(\nu_{D}\chi _{B}{}^{D}-\partial_{B}\nu
_{0})\right\}
 \;.
\label{gR2_10}
\end{equation}
We set
\begin{align}
V_{2}:= &\
-2\overline{g}^{AB}(\overline{\Gamma}_{AC}^{0}\overline{\Gamma}_{B0}^{C}+\overline{\Gamma}
_{AC}^{1}\overline{\Gamma}_{B1}^{C})
\notag \\
= &\
\chi^{AB}\{-(\nu^{0})^{2}\nu_{A}\zeta _{B}+2\nu ^{0}\overline{%
\partial_{0}g_{AB}}-2\nu^{0}\tilde{\nabla}_{A}\nu_{B}+2\overline{g}^{11}\chi
_{AB}\}.
\label{gr2_11}
\end{align}
Finally
\begin{align}
VI_{2}\equiv &\
-\overline{g}^{AB}\overline{\Gamma}_{AD}^{C}\overline{\Gamma}_{BC}^{D}\equiv
-\overline{g}^{AB}(\nu^{0}\nu^{C}\chi_{AD}+\tilde{\Gamma}_{AD}^{C})(\nu^{0}\nu
^{D}\chi_{BC}+\tilde{\Gamma}_{BC}^{D})
\notag \\
= &\
-(\nu^{0})^{2}\nu^{C}\nu^{D}\chi_{C}{}^{A}\chi_{AD}-2\tilde{
\Gamma}_{AD}^{C}\nu^{0}\nu^{D}\chi_{C}{}^{A}
-\overline{g}^{AB}\tilde{\Gamma}_{AD}^{C}\tilde{\Gamma}_{BC}^{D} .
\label{gr2_12}
\end{align}

\subsection{Final computation of $\overline{g}^{AB}\overline{R}_{AB}$}

Adding the results of sections~\ref{gABRAB1} and~\ref{gABRAB2},
we get the final result
\begin{align}
\overline{g}^{AB}\overline{R}_{AB} \equiv &\
 2(\partial_{1}+\overline{\Gamma}^{1}_{11})^{2}\overline{g}^{11}
+3\tau (\partial_{1}+\overline{\Gamma}^{1}_{11})\overline{g}^{11}+
(\partial_{1}\tau +\tau ^{2})\overline{g}^{11}
\notag \\ &
+2(\partial_{1}+\overline{\Gamma}^{1}_{11}+\tau)\overline{\Gamma}^{1}
+\tilde{R}-2\overline{g}^{AB}\overline{\Gamma}^{1}_{1A}\overline{\Gamma}^{1}_{1B}
-2\overline{g}^{AB}\tilde{\nabla}_{A}\overline{\Gamma}^{1}_{1B}
\label{gABRAB} \\
\equiv &\
2(\partial_{1}+\overline{\Gamma}^{1}_{11}+\tau)
\left[(\partial_{1}+\overline{\Gamma}^{1}_{11}+\frac{\tau}{2})
\overline{g}^{11} +\overline{\Gamma}^{1}\right] \notag \\
&\
+\tilde{R}-2\overline{g}^{AB}\overline{\Gamma}^{1}_{1A}\overline{\Gamma}^{1}_{1B}
-2\overline{g}^{AB}\tilde{\nabla}_{A}\overline{\Gamma}^{1}_{1B}
\label{gABRABsimpler}
\end{align}
(an explicit expression for $\overline{\Gamma}^{1}_{11}$ can
be found in (\ref{geodesic})), where
\begin{equation}
-2 \overline{\Gamma}^{1}_{1A}
= \nu^{0}\overline{\partial_{0}g_{1A}}
-\nu^{0}\partial_{1}\nu_{A} + 2\nu^{0}\nu_{B}\chi_{A}{}^{B}
-\nu^{0}\partial_{A}\nu_{0}
\; .
\end{equation}
In this way we have isolated the transversal derivatives in
$\overline{\Gamma}^{1}_{11}$, $\overline{\Gamma}^{1}$ and the
vector $\overline{\Gamma}^{1}_{1A}$. Decomposing $\overline{\Gamma}_{A}=
\overline{W}_{A}+\overline{H}_{A}$ we find the relation
\begin{equation}
-2 \overline{\Gamma}^{1}_{1A}
= \xi_{A} + \overline{H}_{A} + \nu^{0} \nu_{A} \overline{H}_{1},
\end{equation}
with $\xi_{A}$ defined in (\ref{xiA}).

We note the interesting fact that both the
second and third constraints naturally break into two
first-order equations, with the intermediate variable being a
Christoffel, respectively $\overline{\Gamma}^{1}_{1A}$ and
$\overline{g}^{AB}\overline{\Gamma}^1_{AB}$.

\begin{remark}{\rm
The expression in square brackets in (\ref{gABRABsimpler}) can
be rewritten as
\begin{equation}
(\partial_{1}+\overline{\Gamma}^{1}_{11}+\frac{\tau}{2})
\overline{g}^{11} +\overline{\Gamma}^{1}
=
-\frac{\tau}{2}\overline{g}^{11} +
\nu^{0}\overline{g}^{AB}(
  \tilde{\nabla}_B\nu_A -
  \frac{1}{2} \overline{\partial_{0}g_{AB}})
=
\overline{g}^{AB} \overline{\Gamma}^{1}_{AB} .
\end{equation}
This shows that $\overline{g}^{AB}\overline{R}_{AB}$ originally
contains only a first-order radial derivative of
$\overline{g}^{11}$, if we keep the radial derivative of
$\overline{g}^{AB}\overline{\partial_{0}g_{AB}}$. It is
precisely the elimination of the latter object using the
harmonicity condition (\ref{27V.2}) that introduces the
second-order radial derivative of $\overline{g}^{11}$. Hence,
our $\mathcal{C}_{0}$ constraint operator below will contain
such second-order derivative.
}\end{remark}

This leads to the following lemmata:

\begin{lemma}
All terms in  $\overline{g}^{AB}\overline{R}_{AB}$ involving
the derivatives $\overline{\partial_{0}g_{01}}$ and
$\chi^{AB}\overline {\partial_{0}g_{AB}}$ cancel out. The only
new remaining transversal derivative is
$\overline{g}^{AB}\overline{\partial_{0}g_{AB}} $, which can be
eliminated by using
$\overline{\Gamma}^{1}\equiv\overline{H}^{1}+\overline{W}^{1}$.
\end{lemma}

\begin{lemma}
It holds that
\begin{equation}
\overline{S}_{01}\equiv -\frac{1}{2}\nu _{0}\overline{g}^{AB}\overline{R}_{AB}+\overline{R}%
_{1A}\nu ^{A}-\frac{1}{2}\nu _{0}\overline{g}^{11}\overline{R}_{11}
\equiv
\mathcal{C}_{0}+\mathcal{L}_{0} ,
\label{Sbar01}
\end{equation}
where $\mathcal{C}_{0}$ depends only on the quadratic form
$\tilde{g}$ on the cone and the $\overline{W}^{\alpha}$
while $\mathcal{L}_{0}$ is obtained by replacing
$\overline{\Gamma}^{\alpha} $ by $\overline{H}^{\alpha}.$
\end{lemma}

The explicit formula for $\mathcal{L}_{0}$ reads
\begin{eqnarray}
2\nu^{0}\mathcal{L}_{0} &=&
-2\overline{g}^{1A}{\cal L}_{A} -\overline{g}^{11}{\cal L}_{1}
-2\partial_{1}\overline{H}^{1}
-\left(\tau+2\nu^{0}\partial_{1}\nu_{0}-\nu_{0}\overline{W}^{0}\right)\overline{H}^{1}
\notag \\ &&
-\tilde{\nabla}_{A}\overline{H}^{A}
+\xi_{A}\overline{H}^{A}
+\nu_{0}\overline{g}^{11}\partial_{1}\overline{H}^{0}
-\nu^{A}\partial_{A}\overline{H}^{0}
\notag \\ &&
+\Big(
 \overline{g}_{00}\overline{W}^{0}
+\nu_{0}\overline{W}^{1}
+\nu_{A}\overline{W}^{A}
+\frac{1}{2}\nu_{0}\tau\overline{g}^{11}
\notag \\ && \qquad
-\overline{g}^{11}\partial_{1}\nu_{0}
-2\nu^{0}\partial_{1}\overline{g}_{00}
-\overline{g}^{AB}\partial_{A}\nu_{B}
+2\nu^{0}\nu^{A}\partial_{1}\nu_{A}
\Big)\overline{H}^{0}
\notag \\ &&
+\frac{1}{2}\Big\{\overline{g}_{00}(\overline{H}^{0})^{2}
+2\nu_{0} \overline{H}^{0}\overline{H}^{1}
+2\nu_{A} \overline{H}^{0}\overline{H}^{A}
+\overline{g}_{AB}\overline{H}^{A}\overline{H}^{B}\Big\}
\; . \qquad
\label{L0final}
\end{eqnarray}
Note that the last line in the previous equation is quadratic
in the wave-gauge vector $\overline{H}$, and equals
\begin{equation}
\frac{1}{2}\overline{g}_{\mu\nu}\overline{H}^{\mu}\overline{H}^{\nu}
\; .
\end{equation}
All other terms are linear in $\overline{H}$.

\subsection{Constraint}

To write the  wave-map-gauge constraint
$\mathcal{C}_{0}-\overline{T}_{01}=0$ as an equation for
$\overline{g}^{11}$, we use the other constraints, which have
been satisfied since
$\mathcal{L}_{1}=\mathcal{L}_{A}=0=\overline{H}_{1}=\overline{H}_{A}$,
\begin{equation}
\overline{R}_{1A}=\overline{T}_{1A},\quad  \overline{R}_{11}=\overline{T}_{11}.
\label{C0_1}
\end{equation}
We find
\begin{align}
-2\nu^{0}(\mathcal{C}_{0}-\overline{T}_{01})\ \equiv &\
 2(\partial_{1}+\kappa)^{2}\overline{g}^{11}
+3\tau (\partial_{1}+\kappa)\overline{g}^{11}+
(\partial_{1}\tau +\tau^{2})\overline{g}^{11}
\notag \\ &
+2(\partial_{1}+\kappa+\tau)\overline{W}^{1}
+\tilde{R}-\frac{1}{2}\overline{g}^{AB}\xi _{A}\xi _{B}
+\overline{g}^{AB}\tilde{\nabla}_{A}\xi _{B}
\notag \\ &
+ \overline{g}^{11}\overline{T}_{11}+2\overline{g}^{1A}\overline{T}_{1A}+2\overline{g}^{01}\overline{T}
_{01}
\notag \\
=  & \ 0 ,
\label{C0_2}
\end{align}
where $\xi_{A}$ is the vector (\ref{xiA}) and recall that
\begin{equation*}
\kappa\equiv \nu^{0}\partial_{1}\nu_{0}-\frac{1}{2}(\overline{W}_{1}+\tau)
\; .
\end{equation*}
To avoid ambiguities, we emphasise that
the right-hand-side of (\ref{C0_2}) vanishes identically in
wave-map gauge when $T_{\mu\nu}$ there is replaced by the
Einstein tensor $S_{\mu\nu}$. This fact reflects
the identity, valid for any dimension with our choice of
coordinates,
\begin{equation}
  \overline{g}^{AB}\overline{R}_{AB}
+ \overline{g}^{11}\overline{S}_{11}
+ 2 \overline{g}^{1A}\overline{S}_{1A}
+ 2 \overline{g}^{01}\overline{S}_{01} = 0.
\end{equation}

A slightly simplified form of the differential part of the
constraint is, using (\ref{gABRABsimpler}),
\begin{align}
-2\nu^{0}(\mathcal{C}_{0}-\overline{T}_{01})\ \equiv &\
 2(\partial_{1}+\kappa+\tau)\left[
(\partial_{1}+\kappa+\frac{\tau}{2})\overline{g}^{11}
+ \overline{W}^{1} \right]
\notag \\ &
+\tilde{R}-\frac{1}{2}\overline{g}^{AB}\xi _{A}\xi _{B}
+\overline{g}^{AB}\tilde{\nabla}_{A}\xi _{B}
\notag \\ &
+ \overline{g}^{11}\overline{T}_{11}+2\overline{g}^{1A}\overline{T}_{1A}+2\overline{g}^{01}\overline{T}
_{01}
\notag \\
=  & \ 0 .
\label{C0_3}
\end{align}

Suppose that in addition to (\ref{6VI.21new})-(\ref{6VI.22new}) and
(\ref{8VI.21})-(\ref{8VI.22}) it holds that
\begin{eqnarray}
 \label{8VI.23}
 \overline{W}_{0} & = & \overline{W}_{0}(\gamma_{AB},\varphi,
    \nu_{0},\nu_{A},\overline{g}^{11},r,x^{A})
 \;,
\\
 \label{8VI.24}
 \overline{T}_{11} & = &
 \overline{T}_{11}(
 \ldots,\overline{\partial_{0} g_{1A}},\overline{g}^{11},
  \partial_{1}\overline{g}^{11}
 )
 \;, \phantom{xxxxx}
\end{eqnarray}
where $\ldots $ in (\ref{8VI.24}) denotes the collection of
fields already occurring in (\ref{8VI.22}). This is clearly
compatible with the wave-map gauge (\ref{Minkowski6}), and with
scalar fields or Maxwell fields as sources (compare
Section~\ref{sss5V.3}). Then (\ref{C0_2}) becomes a second
order ODE for $\overline{g}^{11}$, linear when the vacuum
Einstein equations and the wave-map gauge have been assumed.

\section{Solution of the $\mathcal{C}_{0}$ constraint}
 \label{solutionC0}

Throughout this section we assume that the target metric is
Minkowski, $\kappa=0$, and that the relevant components of the
tensor $\overline{T}$ are known (e.g., zero). Using the
$\mathcal{C}_{1}$ constraint,
\begin{equation}
\overline{T}_{11}=
 -\Big(\partial_{1}\tau+\frac{1}{n-1}\tau^{2}+|\sigma|^{2}\Big)
 \;,
\end{equation}
we find that the $\mathcal{C}_{0}$  wave-map-gauge constraint
operator can be written as
\begin{align}
\nu^{0}\mathcal{C}_{0} \equiv &\
-\partial_{11}^{2}\overline{g}^{11}
-\frac{3}{2}\tau\partial_{1}\overline{g}^{11}
-\frac{1}{2}\Big(\frac{n-2}{n-1}\tau^{2}-|\sigma|^{2}\Big)\overline{g}^{11}
\nonumber\\ &
-\partial_{1}\overline{W}^{1}
-\tau\overline{W}^{1}
-\frac{1}{2}\tilde{R}
+\frac{1}{4}\overline{g}^{AB}\xi_{A}\xi_{B}
-\frac{1}{2}\overline{g}^{AB}\tilde{\nabla}_{A}\xi_{B}
+\nu^{0}\overline{T}_{1A}\nu^{A}
 \;.
\end{align}
Hence, setting $\overline{g}^{11}\equiv1-\alpha$ and using
using previous notations, the equation for the
$\mathcal{C}_{0}$  wave-map-gauge constraint,
$\mathcal{C}_{0}-\overline{T}_{01}=0$, reads as the linear
second order ODE for $\alpha$
\begin{equation}
L(\alpha)\equiv
 \alpha^{\prime\prime}
+\frac{3}{2}\tau \alpha^{\prime}
+\frac{1}{2}\Big(\frac{n-2}{n-1}\tau^{2} -|\sigma|^{2}\Big)\alpha
=\Phi
 \;,
  \label{2IX.2}
\end{equation}
with $\Phi$ the function known from the previous sections
\begin{align}
\Phi:= &\
\frac{1}{2}\Big(\frac{n-2}{n-1}\tau^{2}-|\sigma|^{2}\Big)
+\partial_{1}\overline{W}^{1}+\tau\overline{W}^{1}
 \notag \\ &\
+\frac{1}{2}\tilde{R}-\frac{1}{4}\overline{g}^{AB}\xi_{A}\xi_{B}
+\frac{1}{2}\overline{g}^{AB}\tilde{\nabla}_{A}\xi_{B}
-\nu^{0} \overline{T}_{1A}\nu^{A}+\nu^{0}\overline{T}_{01}
 \;.
\end{align}
$L(\alpha)$ simplifies to
\begin{equation}
L(\alpha)\equiv
 \alpha^{\prime\prime}
+\frac{3}{2}\big(\frac{n-1}{r}-\psi\big)\alpha^{\prime}
+\frac{1}{2}\Big(\frac{(n-1)(n-2)}{r^{2}}
                 -2\frac{(n-2)}{r}\psi
                 +\frac{n-2}{n-1}\psi^{2}
                 -|\sigma|^{2}\Big)\alpha
 \label{2IX.1}
\end{equation}
This linear equation   has smooth coefficients for $r>0$,  it
has a global solution with initial data given for $r=a>0$.

We proceed to the study of solutions starting from $r=0$.

\subsection{NCT case}

In the NCT case it holds that $\nu^{0}=1$, $f_{A}=0$,
$\nu_{A}=0$, $\tau=\frac{n-1}{r}=-\overline{W}_{1}$,
$\partial_{1}\overline{W}^{1}+\tau\overline{W}^{1}=
-(n-1)(n-2)/r^{2}= -\tilde{R}$ and
$\overline{T}_{\alpha\beta}=0$. Hence $\Phi=0$.
The $\mathcal{C}_{0}$ wave-map-gauge constraint for $\alpha
=1-\overline{g}^{11}$ reduces to
\begin{equation}
2\alpha^{\prime\prime}+\frac{3(n-1)}{r}\alpha^{\prime}
+\frac{(n-1)(n-2)}{r^{2}}\alpha=0\;,\label{C0_NCT}%
\end{equation}
it is a Fuchsian equation with characteristic polynomial
\begin{equation*}
2p(p-1)+3(n-1)p+(n-1)(n-2)
 \;.
\end{equation*}
The zeroes of this polynomial are
\begin{equation*}
p_{+}=\frac{1-n}{2},\text{ \ \ }p_{-}=2-n;
\end{equation*}
both negative or zero for $n\ge 2$. The general
solution of (\ref{C0_NCT}) is
\begin{equation*}
\alpha:=a_{+}r^{(1-n)/2}+a_{-}r^{2-n} ,
\end{equation*}
with $a_{\pm}$ independent of $r$. The only member of this
general solution where $\alpha$ tends to
zero as $r$ tends to zero is $\alpha\equiv0$.

\subsection{General case}

We look for a solution starting from $r=0$ and such that
\begin{equation*}
\lim_{r\to 0}\alpha=\lim_{r\to 0}(r\partial_{1}\alpha)=0\;.
\end{equation*}
We set $\partial_{1}\alpha=\alpha^{\prime}$ and decompose $L$
as follows
\begin{equation}
L(\alpha)\equiv L_{0}(\alpha)+L_{1}(\alpha)
 \;,
\end{equation}
where $L_{0}$ is the Fuchsian operator appearing in the NCT
case,
\begin{equation}
L_{0}(\alpha)\equiv\alpha^{\prime\prime}+\frac{a_{0}}{r}\alpha^{\prime}%
+\frac{b_{0}}{r^{2}}\alpha:=
\alpha^{\prime\prime}+\frac{3(n-1)}{2r}\alpha^{\prime}
+\frac{(n-1)(n-2)}{2r^{2}}\alpha
\;,
\end{equation}
and
\begin{equation}
L_{1}(\alpha)\equiv a_{1}\alpha^{\prime}+b_{1}\alpha:=-\frac{3}{2}\psi
\alpha^{\prime}+\left\{-\frac{(n-2)}{r}\psi+\frac{1}{2}\left(\frac{n-2}{n-1}\psi
^{2}-|\sigma|^{2}\right)\right\}\alpha.
\end{equation}

In order to use the idea of the Fuchs theorem\footnote{{See
e.g.~\cite{YvonneBook}, Appendix V.}}, we write the second
order equation (\ref{2IX.2}) as a first order system for a pair
of unknowns $v:=(v_{1},v_{2})$ by setting $v_{1}:=\alpha$,
$v_{2}\equiv r\alpha^{\prime}$,  hence
$r\alpha^{\prime\prime}=-r^{-1}v_{2}+v_{2}^{\prime}$. The
system reads
\begin{equation*}
rv_{1}^{\prime}-v_{2}=0\;,
\end{equation*}
\begin{equation*}
rv_{2}^{\prime}+(a_{0}-1)v_{2}+b_{0}v_{1}+r(a_{1}v_{2}+b_{1}rv_{1}-r\Phi)=0\;.
\end{equation*}
It is of the form
\begin{equation}
 \label{2IX.3}
rv^{\prime}+Av=r\{F_{1}(r)v+F_{0}(r)\},
\end{equation}
with $A$ the constant linear operator
\begin{equation*}
A\equiv\left(
\begin{array}
[c]{cc}%
0 & -1\\
b_{0} & a_{0}-1
\end{array}
 \right)
\end{equation*}
whose eigenvalues $\mu_{\pm}$ are found by solving the equation
\begin{equation*}
\det\left(
\begin{array}
[c]{cc}%
0-\mu & -1\\
b_{0} & a_{0}-1-\mu
\end{array}
\right)\equiv\mu^{2}+\mu(1-a_{0})+b_{0}=0\;.
\end{equation*}
The solutions are the opposites, $-p_{\pm}$,  of the
characteristic indices computed in the NCT case, hence
nonnegative. Further
\begin{equation*}
F_{1}(r)v\equiv\left(
\begin{array}
[c]{c}%
0\\
-a_{1}v_{2}-rb_{1}v_{1}%
\end{array}
\right),\text{ \ \ }F_{0}(r)
 \equiv\left(
\begin{array}
[c]{c}%
0\\
r\Phi
\end{array}
\right)
 \;,
\end{equation*}
where
\begin{equation*}
a_{1}\equiv-\frac{3}{2}\psi, \quad
rb_{1}\equiv -(n-2)\psi
+\frac{r}{2}\left(\frac{n-2}{n-1}\psi^{2}-|\sigma|^{2}\right)
\end{equation*}
are bounded functions smooth away from $r=0$, as well as
$r\Phi$.  What has been said shows that $F_{0}$ and $F_{1}$ are
continuous at $r=0$ for admissible $\gamma_{AB}$.

\begin{lemma}
Let
\begin{equation}
 \label{2IX.3b}
rv'+Av=r\{F_{1}(r)v+F_{0}(r)\}
\; ,
\end{equation}
be a linear differential system with $A$ a constant linear
operator with non-negative
eigenvalues. Let $F_{1}$ be a continuous linear map and $F_{2}$
a continuous function, for $0\le r\le r_{0}$. The system admits
one and only one solution in $C^{1}([0,r_0])$ which vanishes at
$r=0$.
\end{lemma}

\begin{proof}
We set
$v=Mw$, with $M$ a $2\times2$ matrix satisfying the homogeneous
equation
\begin{equation*}
rM^{\prime}+AM=0
\; .
\end{equation*}
We choose for $M$ the matrix
\begin{equation*}
M=r^{-A}\equiv e^{-A\log r}
\; .
\end{equation*}
The equation (\ref{2IX.3}) reads then
\begin{equation}
w^{\prime}=r^{A}\{F_{1}(r)v+F_{0}(r)\}
\; .
\end{equation}
Hence the equation (\ref{2IX.3b}) together with the condition
$v|_{r=0}=0$ are equivalent to the integral equation
\begin{equation}
v(r)=\int_{0}^{r}(r^{-1}\rho)^{A}
 \{F_{1}(\rho)v(\rho)+F_{0}(\rho)\} d\rho
 \;.
 \label{2IX.4}
\end{equation}
We have
\begin{equation}
\sup_{0\leq\rho\leq r}|(r^{-1}\rho)^{A}|\le 1
\; .
 \label{supremum}
\end{equation}
We set, with $a$ an arbitrary positive number,
\begin{equation}
 \label{26I.1}
C_{1}:=\sup_{0\le r\le a}|F_{1}(r)|
 \;,
\qquad
C_{0}:=\sup_{0\le r\le a}|F_{0}(r)|
 \;.
\end{equation}
The integral equation (\ref{2IX.4}) can then be solved by
iteration, setting
\begin{equation*}
v_{0}(r)=\int_{0}^{r}(r^{-1}\rho)^{A}F_{0}(\rho)d\rho
\; .
\end{equation*}
Hence for $r\le a$
\begin{equation*}
|v_{0}(r)|\leq r C_{0}
\; ,
\end{equation*}
\begin{equation*}
v_{1}(r):=\int_{0}^{r}(r^{-1}\rho)^{A}F_{1}(\rho)v_{0}(\rho)d\rho+v_{0}(r)
\;,
\end{equation*}
and so
\begin{equation*}
|v_{1}(r)-v_{0}(r)|\leq\frac{1}{2}r^{2}\,C_{1}C_{0}
\; ,
\end{equation*}
\begin{equation*}
v_{n+1}(r):=\int_{0}^{r}(r^{-1}\rho)^{A}F_{1}(\rho)v_{n}(\rho)d\rho+v_{0}(r)
\; ,
\end{equation*}
\begin{equation*}
|v_{n+1}(r)-v_{n}(r)|\le
\int_{0}^{r} C_{1}|v_{n}(\rho)-v_{n-1}(\rho)|d\rho
\; .
\end{equation*}
Assume that, as satisfied for $n=1$,
\begin{equation*}
|v_{n}(\rho)-v_{n-1}(\rho)|\le
\frac{1}{(n+1)!}r^{n+1}C_{1}^{n}C_{0}
 \;,
\end{equation*}
Then the same inequality is satisfied when replacing $n$ by
$n+1$.  The sequence $v_{n}$ converges therefore in uniform
norm to a limit $v$, solution of the integral equation
(\ref{2IX.4});, hence of the differential equation
(\ref{2IX.3b}). This solution $v:=(\alpha,r\alpha^{\prime})$,
$\alpha(0)=(r\alpha^{\prime})(0)=0$, is defined, continuous and
bounded for any finite $r$.
\end{proof}

We deduce from this lemma the following theorem.

\begin{theorem}
In the interval of $r\ge 0$, possibly angle dependent, where
the $\mathcal{C}_{1}$ constraint has a global solution and
$\nu_{0}>0$, the $\mathcal{C}_{0}$  wave-map-gauge constraint
with coefficients deduced from the solutions of the
$\mathcal{C}_{1}$ and $\mathcal{C}_{A}$ constraints admits a
solution $\overline{g}^{11}\equiv 1-\alpha$ with
$\alpha(0)=(r\alpha')(0)=0$, $\alpha$ and $r\alpha'$ which are
$C^{1}$ in $r$. The solution is global when it is so of the
solution of the previous constraints, since the system is
linear.
\end{theorem}

\subsection{Vanishing of $\protect\overline{H}_{0}$}

In previous sections we have seen how to achieve
$\overline{H}^{0}= \overline{H}^{A}= 0$, and hence ${\cal
L}_{1}={\cal L}_{A}=0$. Specializing equation (\ref{L0final}) to
this case we get (with
$\overline{H}_{0}:=\nu_{0}\overline{H}^{1}$)
\begin{equation}
\mathcal{L}_{0}\equiv-\partial_{1}\overline{H}_{0}
+\frac{1}{2}(\overline{W}_{1}-\tau)\overline{H}_{0}
\; .
\label{H0preserve_1b}
\end{equation}
On the other hand, the identity
\begin{equation}
S_{01}\equiv S_{01}^{(H)}+
\frac{1}{2}(g_{0\alpha}\hat{D}_{1}H^{\alpha}
           +g_{1\alpha}\hat{D}_{0}H^{\alpha}
           -g_{01}\hat{D}_{\alpha}H^{\alpha})
\label{H0preserve_2b}
\end{equation}
reduces on $C_{O}$ to
\begin{equation}
\overline{S}_{01}\equiv \overline{S}_{01}^{(H)}+
\frac{1}{2}(\overline{g}_{00}\overline{\hat{D}_{1}H^{0}}+
            \nu_{A}\overline{\hat{D}_{1}H^{A}}-
            \nu_{0}\overline{\hat{D}_{A}H^{A}}).
\label{H0preserve_3b}
\end{equation}
Using again the conditions
$\overline{H}^{0}=\overline{H}^{A}=0$ we have, for an
arbitrary target space in adapted coordinates,
\begin{equation}
\overline{\hat{D}_{1}H^{0}} = 0 \;,\qquad
\overline{\hat{D}_{1}H^{A}} = 0 \;,\qquad
\overline{\hat{D}_{A}H^{A}} = \hat{\Gamma}^{A}_{A1}\overline{H}^{1}
\; ,
\label{H0preserve_4b}
\end{equation}
and hence (\ref{H0preserve_3b}) further reduces to
\begin{equation}
\overline{S}_{01}\equiv\overline{S}_{01}^{(H)}
-\frac{1}{2}\hat{\Gamma}^{A}_{A1}\overline{H}_{0}
\;.
\label{H0preserve_5b}
\end{equation}

For a solution of the Einstein equations in wave-map gauge it
holds then that
\begin{equation}
\overline{S}_{01}\equiv \mathcal{C}_{0}+\mathcal{L}_{0} =\overline{T}_{01}
-\frac{1}{2}\hat{\Gamma}^{A}_{A1} \overline{H}_{0}.
\label{H0preserve_6b}
\end{equation}
Therefore, when $\overline{H}^{0}=\overline{H}^{A}=0$ and the
initial data satisfy the  wave-map-gauge constraint
\begin{equation}
\mathcal{C}_{0}-\overline{T}_{01}=0 ,
\label{H0preserve_7b}
\end{equation}
then $\overline{H}_{0}$ satisfies the equation
\begin{equation}
\partial_{1}\overline{H}_{0} =
\frac{1}{2}(\overline{W}_{1}-\tau+\hat{\Gamma}^{A}_{A1})\overline{H}_{0}
\;  .
\label{H0preserve_8b}
\end{equation}
For a Minkowski target, and using the boundary conditions
(\ref{gvertexlimits}), we have
\begin{equation}
\hat{\Gamma}^{A}_{A1}\stackrel{\eta}{=}\frac{n-1}{r} \;, \quad
\lim_{r\to 0}(r\overline{W}_{1})\stackrel{\eta}{\equiv}
\lim_{r\to 0}(-r^{2}\overline{g}^{AB}s_{AB}) = -(n-1) , \quad
\lim_{r\to 0}(r\tau) = n-1
\;,
\label{H0preserve_9b}
\end{equation}
and hence equation (\ref{H0preserve_8b}) takes a Fuchsian form
\begin{equation}
r \partial_{1}\overline{H}_{0}+\frac{n-1}{2}\overline{H}_{0}
+rM\overline{H}_{0} \stackrel{\eta}{=} 0
\; .
\label{H0preserve_10b}
\end{equation}
with $M$ a continuous function up to $r=0$.

We want to prove that $H_{0}=0$ when the spacetime metric is a
$C^{2}$ solution of the Einstein equations in Minkowski-wave-ap
gauge; in this case the wave gauge vector $H$ is $C^{1}$, then
$\overline{H}_{0}$ tends to a finite limit at the vertex. The
equation (\ref{H0preserve_10b}) implies that this limit is zero
hence that the only solution is zero.

\begin{remark} {\rm
If we add constraint damping terms as in
(\ref{DampingRicci}) we obtain instead, using again
$\overline{H}^{0}=\overline{H}^{A}=0$,
\begin{equation}
\mathcal{L}_{0} = -\partial_{1}\overline{H}_{0}
+\frac{1}{2}(\overline{W}_{1}-\tau+\epsilon \rho n_{1}) \overline{H}_{0}
\; .
\label{damping0x}
\end{equation}
This new term does not change the leading orders in $r$ of
equation (\ref{H0preserve_10b}) and hence $\overline{H}_{0}=0$
is still the only regular solution. We conclude that the
addition of constraint damping terms is fully compatible with
the  wave-map-gauge constraint hierarchy. }
\end{remark}

\bigskip

\section{Wave-map gauge constraints: a summary}
\label{summary}

We have defined $C_{O}$ to be the cone represented in $\mathbf{R}^{n+1}$ by
the Minkowskian cone
\begin{equation}
y^{0}=r,\quad r^{2}:=\sum_{i=1}(y^{i})^{2},
\end{equation}
equivalently
\begin{equation}
x^{0}=0\;,\quad x^{1}=r,\quad\Theta^{i}(x^{A})=\frac{y^{i}}{r}.
\end{equation}
We have considered on $C_{O}$ a non degenerate quadratic form given in
$x^{\alpha}$ coordinates by:
\begin{equation*}
\overline{g}_{00}(dx^{0})^{2}+2\nu_{0}dx^{0}dx^{1}+2\nu_{A}dx^{0}dx^{A}+
\overline{g}_{AB}dx^{A}dx^{B}.
\end{equation*}
We have proved (recall that admissible means hypotheses on smoothness and
limits at the vertex spelled out in various sections)

\begin{theorem}
 \label{T21V10.1}
1. Let $\tilde{g}$ be a given admissible degenerate quadratic form on
$C_{O}^{T_{0}}$,
\begin{equation*}
\tilde{g}=\overline{g}_{AB}dx^{A}dx^{B}.
\end{equation*}
There exists on $C_{O}^{T}$, for some $T$ with $0<T\leq T_{0}$, coefficients
$\nu_{0}$, $\nu_{A}$, $\overline{g}_{00}$ satisfying the vacuum Einstein wave-map
gauge constraints, and unique modulo admissible vertex limits.

2. An admissible degenerate quadratic form $\tilde{g}$ together with a non
vanishing $\nu_{0}$ can be determined on $C_{O}^{T}$, for some $T$ with
$0<T\leq T_{0}$, from the first vacuum Einstein wave-map gauge constraints, an
admissible quadratic form $\gamma$ and a scalar function $\kappa$ being
arbitrarily given. Then $\tilde{g}$ is conformal to $\gamma$ and depends only
on its conformal class, $\nu_{0}$ is linked to $\kappa$ by the differential
equation (\ref{nu0equation}). They are unique under admissible vertex limits.

When $\tilde{g}$ is known, $\nu_{A}$ and $\overline{g}_{00}$ are determined as in
point 1.\ by the second and third wave-map gauge constraint and admissible
vertex limits.
\end{theorem}

\section{Local geometric uniqueness for the vacuum Einstein equations}
\label{slgu}

In this section only the vacuum Einstein equations will be
considered.

Recall that two spacetimes $(V_{a},g_{a})$ and $(V_{b},g_{b})$
are considered as (both geometrically and physically) the same
if there exists a diffeomorphism $\phi :V_{a}\rightarrow V_{b}$
such that on $V_{a}$ it holds that $g_{a}=\phi _{\ast }g_{b}$.
We have said before that given a $C^{1,1}$ metric $g_{a}$ on a
manifold $V_{a}$ and $O_{a}\in V_{a}$ there are in some
neighbourhood of $O_{a}$ {normal} coordinates $y_{a}^{\alpha }$ centered at $%
O_{a}$, where the characteristic cone $C_{O_{a}}$ is
represented, for $0\leq y_{a}^{0}\leq T_{a}$ by the equation of
a Minkowskian cone in $R^{n+1}$
\begin{equation*}
y_{a}^{0}=r_{a},\quad r_{a}^{2}:=\sum_{i=1}^{n}(y_{a}^{i})^{2}.
\end{equation*}
The null rays issued from $O_{a}$ are represented by the
generators of this cone. We have defined adapted null
coordinates by setting
\begin{equation}
x_{a}^{0}:=r_{a}-y_{a}^{0},\quad x_{a}^{1}=r_{a},\text{ \ with \ }%
r_{a}^{2}=\sum_{i=1}^{n}(y_{a}^{i})^{2}  \label{29III10.1}
\end{equation}
and $x_{a}^{A}$ local coordinates on the sphere $S^{n-1}$. In the
coordinates $x_{a}^{\alpha }$ the metric $g_{a}$ reads on the cone $%
C_{O_{a}} $
\begin{equation}
\overline{g}_{a,00}(dx_{a}^{0})^{2}+2\nu _{a,0}dx_{a}^{0}dx_{a}^{1}+2\nu
_{a,A}dx_{a,}^{0}dx_{a}^{A}+\overline{g}_{a,AB}dx_{a}^{A}dx_{a}^{B}.
\end{equation}
We have shown moreover (see Section~\ref{ssBcv}) that it is
possible to choose the coordinates $y_{a}^{\alpha }$ so that it
holds
\begin{equation}
\lim_{r\rightarrow 0}r_{a}^{-3}(\overline{g}_{a,AB}-r_{a}^{2}s_{AB})=0\;,
\quad
\lim_{r\rightarrow 0}r_{a}^{-2}\partial_{1}(\overline{g}_{a,AB}
-r_{a}^{2}s_{AB})=0\;,
 \label{24IV10.2}
\end{equation}
\begin{equation}
\lim_{r\rightarrow 0} r_{a}^{-1}(\nu_{a,0}-1) = 0
\quad
\lim_{r\rightarrow 0}r_{a}^{-2}\nu_{a,A}=
\lim_{r\rightarrow 0}r_{a}^{-1}\partial_{1}\nu _{a,A}=0\;,
\end{equation}
and even
\begin{equation}
\lim_{r\rightarrow 0}r_{a}^{-2}\overline{\partial_{0}(g_{a,AB}-r_{a}^{2}s_{AB})}=0\;,
\end{equation}
\begin{equation}
\lim_{r\rightarrow 0}r_{a}^{-1}\overline{\partial _{0}g_{a,1A}}%
=\lim_{r\rightarrow 0}r_{a}^{-1}\overline{\partial _{0}g_{a,0A}}=0\;,
\end{equation}
while
\begin{equation}
\lim_{r\rightarrow 0}r_{a}^{-1}(\overline{g}_{a,00}+1)=0\;,\quad
\lim_{r\rightarrow 0}\partial _{1}\overline{g}_{a,00}=0=\lim_{r\rightarrow 0}%
\overline{\partial _{0}g_{a,00}}\;.  \label{29III10.2}
\end{equation}

Having chosen such coordinates $y_{a}^{\mu}$, respectively
$y_{b}^{\mu}$, for the metrics $g_{a}$ and $g_{b}$, we obtain a
diffeomorphism $\phi_{N}$ by $y_{b}^{\alpha}(y_{a}^{\alpha
}):=y_{a}^{\alpha }$, defined in the subset $y_{a}^{0}\leq
T:=\min(T_{a},T_{b})$. Such a diffeomorphism will be called
\emph{canonical.}
We remark that canonical diffeomorphisms are \emph{not} unique,
and that above we have not required $r$ to be an affine
parameter.

The metrics $g_{b}$ and $\phi _{N,\ast }g_{b}$ are
geometrically equivalent,
and one has equality of components ($\phi_{N,\ast}g_{b})^{\lambda%
\mu}(y_{a})=g_{b}^{\lambda \mu }(y_{b})$ for $y_{b}^{\alpha}=y_{a}^{\alpha}$%
. The coordinates $y^{\alpha}$ are normal for both metrics and
they satisfy in the coordinates $x^{\alpha }$ the vertex limits
(\ref{24IV10.2}--\ref {29III10.2}) recalled above.

To study the geometric uniqueness of our characteristic Cauchy
problem we first consider two metrics $g_{a}$ and $g_{b}$ on
the same manifold which satisfy the characteristic Cauchy
problem on the same cone $C_{O}$. We will prove the following
theorem, using the notations given in previous sections for
$C_{O}$ and $Y_{O}$ (note that we are not assuming
an affine parameterisation of the cone generators here):

\begin{theorem}
\label{T29III10.1} Consider two smooth solutions $g_{a}$ and $g_{b}$ in $%
Y_{O}^{T}$ of the Cauchy problem for the vacuum Einstein equations $%
Ricci(g)=0$ with data on the cone $C_{O}^{T}$, characteristic
for both
metrics. There exists $T^{\prime }\leq T$ so that $g_{a}$ is equivalent to $%
g_{b}$ in $Y_{O}^{T^{\prime}}$ if and only if they induce on
$C_{O}^{T}$ the same degenerate quadratic form satisfying in
the coordinates $x^{\alpha}$ the vertex limits
(\ref{24IV10.2}--\ref{29III10.2}).
\end{theorem}

\begin{proof}
We put the metric $g_{a}$ in Minkowski wave-map gauge by
constructing a wave map $f_a $, that is a solution of the
semilinear, tensorial, partial differential equations which
read in abstract index notation
\begin{equation}
\square \,_{g_{a},\hat{g}}f ^{\alpha }_a\equiv g_{a}^{\lambda \mu
}(\partial _{\lambda \mu }^{2}f ^{\alpha} _a-\Gamma _{a},_{\lambda \mu
}^{\sigma }\partial _{\sigma }f ^{\alpha }_a+\partial _{\lambda
}f ^{\sigma }_a\partial _{\mu }f ^{\rho }_a\hat{\Gamma}_{\sigma \rho
}^{\alpha })=0\;,
\end{equation}
which on $C_{O}^{T}$ is the trace of the identity mapping $I$
of $\mathbf{R}^{n+1}$. To simplify the writing we suppress the
index $a$ in the following computations, valid for any metric
$g$ with normal coordinates $y^{\alpha }$ and adapted null
coordinates $x^{\alpha }$, we will reestablish $a$ and $b$ in
the conclusions.

The components $\underline{f^{\alpha }\text{ }}$ and $f^{\alpha
}$ of the image point are linked by the same relations as the
coordinates $y$ and $x$. They take in coordinates $x^{\alpha }$
the initial data
\begin{equation}
\overline{f}^{0}=0\;,\quad \overline{f}^{1}=x^{1},\quad \overline{f}^{A}=x^{A},\quad \text{%
for }x^{0}=0;
\label{partialf}
\end{equation}
and in the coordinates $y^{\alpha }$ the initial data
\begin{equation}
\underline{\overline{f}}^{i}=y^{i},\quad \underline{\overline{f}}^{0}=r,
\end{equation}
we see that in the $y$ coordinates the initial data are the
trace on $C_{O}$ of the set of $C^{\infty }$ functions on
$\mathbf{R}^{n+1}$
\begin{equation*}
\underline{I}^{i}=y^{i},\quad \underline{I}^{0}=y^{0}.
\end{equation*}
The existence of a $C^{2}$ wave map $f$ in some $Y_{O}^{T_{a}}$
taking these initial data can therefore be proved using the
Cagnac-Dossa theorem. In
fact, since the equations are linear in a coordinate system where the $\hat{%
\Gamma}$'s vanish, the usual linear theory~\cite{Friedlander}
suffices to
obtain the result. The resulting wave map extends to a $C^{2}$ mapping%
\footnote{%
though not in general to a $C^{2}$ wave map.}.

To prove that it is a diffeomorphism at least in a
neighbourhood of the vertex we first remark that our
definitions imply
\begin{equation}
\overline{\underline{\partial_{i}f}^{j}}=\delta_{i}^{j}.
\end{equation}
To study the derivatives \underline{$\partial _{0}$} we return
to the $x$ coordinates and consider the set of functions
\begin{equation}
f^{0}-x^{0},\quad f^{1}-x^{1},\quad f^{A}-x^{A}\;.
\end{equation}
They vanish on $C_{O}$, so do therefore their tangential derivatives on $%
C_{O}$, hence by application of the Lemma~\ref{lemma4.2}
\begin{equation}
\lim_{r\rightarrow 0}\overline{\partial_{0}f^{1}}=0\;,\quad
\lim_{r\rightarrow 0}\overline{\partial_{A}f^{0}}=0\;, \qquad
\lim_{r\rightarrow 0}\overline{\partial_{0}f^{0}}=1
 \;.
   \label{24IV10.5}
\end{equation}
By definition of the coordinates $x$ and $y$ we have
\begin{equation*}
\underline{f}^{0}\equiv
(\Sigma (f^{i})^{2})^{\frac{1}{2}}-f^{0}
\end{equation*}
hence
\begin{equation}
\underline{\partial_{0}f^{0}}:=
\frac{\partial}{\partial y^{0}}\underline{f^{0}}
\equiv
-\frac{\partial}{\partial y^{0}}f^{0} =
\frac{\partial}{\partial x^{0}}f^{0} :=
\partial_{0}f^{0}
\end{equation}
while $\underline{f^{i}}$ depends only on $f^{1}$ and
$f^{A}.$ Therefore
\begin{equation}
\lim_{r\rightarrow 0}\overline{\underline{\partial_{0}f^{i}}}=0\;,
\quad
\lim_{r\rightarrow 0}\overline{\underline{\partial_{0}f^{0}}}=1
 \;.
   \label{24IV10.5x}
\end{equation}
Since the Jacobian of the $C^{1}$ mapping $f $ tends to $1$ at
$O$, it is a diffeomorphism, between at least small
neighbourhoods of $O$.

The initial data, trace $\overline{g}^{(H)}$ of the metric
$g^{(H)}$ in wave gauge are linked with the original
$\overline{g}$ by the classical relation
\begin{equation}
 \overline{g}_{\alpha \beta }\equiv \overline{\partial _{\alpha }f^{\lambda }}%
 \overline{\partial _{\beta }f^{\mu }}\overline{g}_{\lambda \mu }^{(H)}
 \;.
 \label{24IV10.3}
\end{equation}
The values
of $\partial_{i}\overline{f}$ in the coordinates $x^{\alpha}$
show the equality of quadratic forms
$\tilde{g}^{(H)}\equiv \tilde{g}$, indeed in these coordinates:
\begin{equation}
\overline{g}_{11}^{(H)}=\overline{g}_{11}=0\;,\quad
\overline{g}_{1A}^{(H)}=\overline{g}_{1A}=0\;,\quad
\overline{g}_{AB}^{(H)}=\overline{g}_{AB} \;.
\end{equation}

Since $g^{(H)}$ is in wave gauge, and satisfies the
vacuum Einstein equations, its trace $\overline{g}^{(H)}$
satisfies the wave-map gauge constraint $\mathcal{C}_{1}=0$,
and $\nu _{0}^{(H)}$ satisfies the same equation than $\nu
_{0}$,
\begin{equation}
\partial_{1}\nu_{0}^{(H)}=\nu_{0}^{(H)}\{\frac{\partial_{1}\tau }{\tau }+
\frac{1}{2}(\nu_{0}^{(H)}\overline{W}^{0}+\tau)+\frac{|\sigma|^{2}}{\tau }
+\frac{\tau }{n-1}\},
 \label{C1finalxxx}
\end{equation}
since the coefficients depend only on $\tilde{g}$, to show that
$\nu_{0}^{(H)}$ tends to $1$ at the vertex like $\nu_{0}$  we
use the identity
\begin{equation}
\nu_{0}\equiv \overline{\partial_{0}f^{\lambda}}
\overline{\partial_{1}f^{\mu}}\overline{g}_{\lambda\mu}^{(H)}\;.
\end{equation}
and the limits \eqref{24IV10.5}-\eqref{24IV10.5x} give
\begin{equation*}
1=\lim_{r\rightarrow 0}\nu _{0}=\lim_{r\rightarrow 0}\nu _{0}^{(H)}.
\end{equation*}
Uniqueness of solutions of \eqref{C1finalxxx} (with non zero $\tau $) with
this limit at $O$ shows that the function $\nu _{0}^{(H)}$ depends only on
$\tilde{g}$, i.e. $\nu_{a,0}^{(H)}=\nu_{b,0}^{(H)}$ since $\tilde{g}_{a}=
\tilde{g}_{b}$, that is $\overline{g}_{a,AB}=\overline{g}_{b,AB}$.

As a consequence, the wave-map-gauge constraints $C_{a,A}=0$ written with
$\overline{g}_{AB}\equiv \overline{g}_{AB}^{(H)}$ and $\nu _{a,0}^{(H)}$
are the same equation for $\nu _{a,A}^{(H)}$, $a=1$ or $2$. The
vertex limit of $\nu _{A}^{(H)}$ will be deduced from the
definition:
\begin{equation}
\nu _{A}:=\overline{g}_{0A}\equiv \overline{\partial _{0}f^{\lambda }}\overline{%
\partial _{A}f^{\mu }}\overline{g}_{\lambda \mu }^{(H)}\equiv \overline{\partial
_{0}f^{\lambda }}\overline{g}_{A\lambda }^{(H)}  \label{24IV10.1}
\end{equation}
which implies using (\ref{partialf}) (compare \eqref{7VI.35})
\begin{equation}
0=\lim_{r\rightarrow 0}r^{-2}\nu _{A}=\lim_{r\rightarrow 0}r^{-2}\nu
_{A}^{(H)},\qquad \lim_{r\rightarrow 0}\nu ^{A}=\lim_{r\rightarrow 0}\overline{g}%
^{AB}\nu _{B}=\lim_{r\rightarrow 0}r^{-2}s^{AB}\nu _{B}=0\;.
\end{equation}
Differentiating (\ref{24IV10.1}) gives
\begin{equation*}
\partial _{1}\nu _{A}\equiv \overline{\partial _{1}\partial _{0}f^{\lambda }}%
\overline{g}_{A\lambda }^{(H)}+\overline{\partial _{0}f^{\lambda }}\partial _{1}%
\overline{g}_{A\lambda }^{(H)}.
\end{equation*}
We have
\begin{equation*}
\lim_{r\rightarrow 0}r^{-1}\partial _{1}\nu _{A}\equiv \lim_{r\rightarrow
0}\,(r^{-1}\overline{\partial _{1}\partial _{0}f^{\lambda }}\overline{g}%
_{A\lambda }^{(H)})+\lim_{r\rightarrow 0}\,(r^{-1}\overline{\partial
_{0}f^{\lambda }}\partial _{1}\overline{g}_{A\lambda }^{(H)})
\end{equation*}
with, by (\ref{partialf}) and \eqref{24IV10.5},
\begin{equation*}
\lim_{r\rightarrow 0}\,(\overline{\partial _{0}f^{\lambda }}r^{-1}\partial
_{1}\overline{g}_{A\lambda }^{(H)})=\lim_{r\rightarrow 0}r^{-1}\partial _{1}\overline{g%
}_{A0}^{(H)},
\end{equation*}
and
\begin{equation*}
\lim_{r\rightarrow 0}\,(r^{-1}\overline{\partial _{1}\partial _{0}f^{\lambda
}}\overline{g}_{A\lambda }^{(H)})=\lim_{r\rightarrow 0}\,(r\overline{\partial
_{1}\partial _{0}f^{B}})s_{AB}.
\end{equation*}
Taking the trace on the cone of the wave map equation, with
Minkowski target, gives
\begin{eqnarray}
2\overline{g}^{01}(
   \partial_{1}\overline{\partial_{0}f^{A}}
  -\overline{\Gamma}_{10}^{0}\overline{\partial_{0}f^{A}}
  -\overline{\Gamma}_{10}^{A}
  +\hat{\Gamma}_{1B}^{A}\overline{\partial_{0}f^{B}}
) &&  \notag \\
-2\overline{g}^{1B}(
   \overline{\Gamma}_{1B}^{A}
  -\hat{\Gamma}_{1B}^{A}
)-\overline{g}^{BC}(
   \overline{\Gamma}_{BC}^{0}\overline{\partial_{0}f^{A}}
  +\overline{\Gamma}_{BC}^{A}
  -\hat{\Gamma}_{BC}^{A}
) &=&0.
\end{eqnarray}
We have
\begin{equation*}
\lim_{r\rightarrow 0}\overline{\Gamma}_{10}^{0}=0\quad \text{and}\quad
\lim_{r\rightarrow 0}\overline{\partial _{0}f^{A}}=0\;,
\end{equation*}
\begin{equation*}
\lim_{r\rightarrow 0}r\overline{\Gamma}_{10}^{A}=\lim_{r\rightarrow 0}\frac{1}{2}%
\{-r\nu^{A}\nu^{0}\overline{\partial_{0}g_{11}}+r\overline{g}^{AB}(\overline{%
\partial_{0}g_{1B}}+\partial_{1}\nu_{B}-\partial_{B}\nu_{0})\}=0\;.
\end{equation*}
Finally, for a wave map $f$
\begin{equation*}
\lim_{r\rightarrow 0}\,(r\partial_{1}\overline{\partial_{0}f^{A}})=0\;,
\end{equation*}
hence, since $\lim_{r\rightarrow 0}r^{-1}\partial_{1}\nu_{A}=0$, we obtain
\begin{equation*}
\lim_{r\rightarrow 0}r^{-1}\partial _{1}\nu _{A}^{(H)}= 0\;.
\end{equation*}
Since $\nu_{a,A}$ and $\nu_{b,A}$ satisfy the same equation
and the same boundary conditions, they are equal.

It remains to analyse the boundary conditions for the functions
$\overline{g}_{a,00}$, which again satisfy the same equation for $a=1$ or
$2$. We have
\begin{equation*}
g_{00}=\partial _{0}f^{\lambda }\partial _{0}f^{\mu }g_{\lambda \mu }^{(H)}.
\end{equation*}%
It implies:
\begin{equation*}
\lim_{r\rightarrow 0}\overline{g}_{00}\equiv
\lim_{r\rightarrow 0}\left(
\overline{g}_{00}^{(H)}\overline{\partial_{0}f^{0}}\overline{\partial_{0}f^{0}}
+2\overline{\partial_{0}f^{0}}\overline{\partial_{0}f^{1}}\nu_{0}^{(H)}
+\overline{g}_{AB}^{(H)}\overline{\partial_{0}f^{A}}\overline{\partial_{0}f^{B}}
 \right)
  \;.
\end{equation*}
The previous limits imply then
\begin{equation*}
\lim_{r\rightarrow 0}\overline{g}_{00}=\lim_{r\rightarrow 0}\overline{g}_{00}^{(H)}=-1.
\end{equation*}
Also
\begin{equation*}
\partial_{1}g_{00}=
\partial_{0}f^{\lambda}\partial_{0}f^{\mu}\partial_{1}g_{\lambda\mu}^{(H)}
+2\partial_{1}\partial_{0}f^{\lambda}\partial_{0}f^{\mu}g_{\lambda\mu}^{(H)}\;,
\end{equation*}
hence using previous limits
\begin{equation*}
\lim_{r\rightarrow 0}r\partial _{1}\overline{g}_{00}=\lim_{r\rightarrow
0}\{r\partial _{1}\overline{g}_{00}^{(H)}+2r\,(\overline{\partial _{1}\partial
_{0}f^{0}}\overline{g}_{00}^{(H)}+\overline{\partial _{1}\partial _{0}f^{1}}\nu
_{0}^{(H)})\}.
\end{equation*}
We have, by definition of a wave map with Minkowskian target,
\begin{equation*}
2\nu^{0}(
   \partial_{1}\overline{\partial_{0}f^{0}}
  -\overline{\Gamma}_{10}^{0}\overline{\partial_{0}f^{0}}
)
+\nu^{0}\tau\overline{\partial_{0}f^{0}}
+\overline{g}^{AB}\hat{\Gamma}_{AB}^{0}=0\;.
\end{equation*}
Hence
\begin{equation*}
\lim_{r\rightarrow 0}r\{2(\partial _{1}\overline{\partial _{0}f^{0}}-\lim
\overline{\Gamma}_{10}^{0})+\frac{n-1}{r}-\psi -\frac{n-1}{r}\}=0\;,
\end{equation*}
which gives
\begin{equation*}
\lim_{r\rightarrow 0}r\partial _{1}\overline{\partial _{0}f^{0}}=0\;.
\end{equation*}
One finds also
\begin{equation*}
\lim_{r\rightarrow 0}r\partial _{1}\overline{\partial _{0}f^{1}}=0\;,
\end{equation*}
hence
\begin{equation*}
\lim_{r\rightarrow 0}r\partial _{1}\overline{g}_{00}=\lim_{r\rightarrow
0}r\partial _{1}\overline{g}_{00}^{(H)}=0\;.
\end{equation*}

We have proved that $\tilde{g}_{a}=\tilde{g}_{b}$ on $C_{O}^{T}$ implies $%
\overline{g}_{a}^{(H)}=\overline{g}_{b}^{(H)}$ on $C_{O}^{T}$ hence by
uniqueness for
the hyperbolic system of the Einstein equations in wave gauge $%
g_{a}^{(H)}=g_{b}^{(H)}$ in $Y_{O}^{T}$. The metrics $g_{a}$
and $g_{b}$ are geometrically equivalent.

The reverse implication is trivial.
\end{proof}

\bigskip

Our next result, one of the main results of this paper, is a
straightforward corollary of Theorem~\ref{T29III10.1}:

\begin{theorem}
\label{T29III10.2} Given points $O_{a}\in V_{a}$ and $O_{b}\in
V_{b}$ denote by $C_{O_{a}}$ and $C_{O_{b}}$ the characteristic
(null) cones of smooth
Lorentzian metrics $g_{a}$ on $V_{a}$ and $g_{b}$ on $V_{b}$. Denote by $%
J_{a}^{+}$ the future of the point $O_{a}$ in the metric
$g_{a}$. There are neighbourhoods $U_{a}$ of $O_{a}$ and
$U_{b}$ of $O_{b}$ such that the spacetimes $(U_{a}\cap
J_{a}^{+},g_{a})$ and $(U_{b}\cap J_{b}^{+},g_{b})$
are locally geometrically the same if and only if the pull back $%
\phi_{N}^{\ast}\tilde{g}_{b}$, where $\phi_{N}$ is a canonical
diffeomorphism of $U_{a}$ onto $U_{b}$, equals $\tilde{g}_{a}$.
\end{theorem}

\begin{proof}
The spacetimes $(U_{b}\cap J_{b}^{+},g_{b})$ and
$(\phi_{N}^{-1}(U_{b}\cap J_{b}^{+})\subset U_{a}\cap
J_{a}^{+}),\phi_{N}^{\ast }g_{b})$ are geometrically
equivalent. Theorem~\ref{T29III10.1} shows that the second one
is locally geometrically equivalent to $(U_{a}\cap
J_{a}^{+}),g_{A})$, the conclusion follows from the fact that
$\phi_{N}^{\ast}\tilde{g}_{b}=
\widetilde{\phi_{N}^{\ast}g_{b}}$ and satisfy the required
vertex limits.
\end{proof}

\bigskip

From the Uniqueness Theorem~\ref{T21V10.1} for the constraints
one deduces straightforwardly a formulation of geometric local
uniqueness starting from data $\gamma $ and $\kappa $.

\section{Conclusions, and open problems}

We have shown that the trace $\overline{g}$  on a characteristic
cone of a solution of Einstein equations which is also a
solution of the reduced Einstein equations in wave-map gauge
satisfies necessarily a set of $n+1$ equations which we have
called wave-map gauge constraints, written out explicitly and
solved. We have shown that, conversely a solution of the
reduced Einstein equations in wave-map gauge with trace
satisfying these wave-map gauge constraints satisfies the
original Einstein equations. Finally we have shown that every
solution of the vacuum Einstein equations is locally (i.e. in a
neighbourhood of the vertex) isometric to a solution in wave
map gauge, uniquely determined, modulo some regularity
conditions, by the degenerate quadratic form induced on the
characteristic cone by the spacetime metric.

There remain many interesting open problems:

\begin{itemize}
\item
Determine the minimum regularity, in particular at the vertex,
under which the initial data lead to a local solution (see
also~\cite{CCM3}).
\item
Extend our analysis to a  characteristic cone with vertex at
$i^-$ (cf.~\cite{F1}).
\item
Study the asymptotic behaviour of the solutions of the wave-map
gauge constraint equations  at future null infinity.
\item
Prove  global existence for small initial data of solutions of
the Einstein equations in higher dimensions by a conformal
method, as was done for the spacelike Cauchy problem with data
identically Schwarzschild outside of a bounded
region~\cite{CCL}.
\item
Prove global existence using the approach of
Lindblad-Rodnianski~\cite{LindbladRodnianski2,Loizelet:AFT}
(compare~\cite{Friedrich:86,CaciottaNicoloI}).
\end{itemize}

\appendix
\section{Collected formulae}
 \label{A7VI.1}
The metric on $C_{O}=\{x^{0}=0\}$ is written as
\begin{equation}
 g = \overline{g}_{00} (dx^{0})^{2} + 2 \overline{g}_{0A} dx^{0} dx^{A} + 2 \overline{g}_{01} dx^{0} dx^{1} +
 \overline{g}_{AB} dx^{A} dx^{B}
 \;,
\end{equation}
and recall that we \emph{do not} assume that this form of the
metric is preserved under differentiation in the
$x^{0}$--direction. {Here and elsewhere we put overbars on
the relevant quantities whenever the formulae hold only on
$C_{O}$.} The inverse is
\begin{equation}
\label{29VIII.1}
g^{\sharp} = \overline{g}^{11}\partial^{2}_{1} + 2 \overline{g}^{1A} \partial_{1} \partial_{A}
+ 2 \overline{g}^{01} \partial_{0} \partial_{1} + \overline{g}^{AB}\partial_{A} \partial_{B}
  \;,
\end{equation}
with
\begin{equation}
\overline{g}^{01} = \frac{1}{\overline{g}_{01}}
\;,
\qquad
\overline{g}^{1A} = -\overline{g}^{01} \overline{g}^{AB} \overline{g}_{0B}
\;,
\qquad
\overline{g}^{11} = (\overline{g}^{01})^{2} (-\overline{g}_{00} + \overline{g}^{AB}\overline{g}_{0A} \overline{g}_{0B})
\;.
\end{equation}

We introduce the special notations
\begin{equation}
\nu_{0} := \overline{g}_{01} \;, \qquad
\nu_{A} := \overline{g}_{0A} \;, \qquad
\tilde{g} := \overline{g}_{AB} dx^{A} dx^{B} \;,
\end{equation}
\begin{equation}
\nu^{0} := \overline{g}^{01} = \frac{1}{\nu_{0}} \;, \qquad
\nu^{A} := -\overline{g}_{01}\overline{g}^{1A} = \overline{g}^{AB} \nu_{B} \;.
\end{equation}
Then
\begin{equation}
\overline{g}^{1A} = -\nu^{0} \nu^{A} \;, \qquad
\overline{g}^{11} = (\nu^{0})^{2}(-\overline{g}_{00}+\nu^{A}\nu_{A}) \;.
\end{equation}

The determinant reads
\begin{equation}
 \label{11VI.A7}
\sqrt{|\det \overline{g}|} = \nu_{0} \sqrt{\det\tilde g_{\Sigma}} \; .
\end{equation}

The Levi-Civita connection of the metric $\overline{g}_{AB}$
will be denoted as $\tilde{\nabla}_{A}$, with corresponding Christoffel
symbols $\tilde{\Gamma}^C_{AB}$ with respect to the derivative
$\partial_{A}$.

We have the following Christoffel symbols on the null
hypersurface,
%
\begin{eqnarray}
\overline{\Gamma}^{0}_{00}
&=&
\frac{1}{2} \nu^{0} ( -\partial_{1}  \overline{g}_{00} + 2 \overline{\partial_{0} g_{01}} )
\; , \\
\overline{\Gamma}^{0}_{01}
&=&
\frac{1}{2} \nu^{0} \overline{\partial_{0} g_{11}}
\; , \\
\overline{\Gamma}^{0}_{11}
&=&
0
\; , \\
\overline{\Gamma}^{1}_{00}
&=&
  \frac{1}{2}\nu^{0}\nu^{A} ( \partial_{A}\overline{g}_{00} - 2\overline{\partial_{0} g_{0A}} )
+ \frac{1}{2} \overline{g}^{11} ( -\partial_{1}  \overline{g}_{00} + 2\overline{\partial_{0} g_{01}} )
 \nonumber
\\
&&
+ \frac{1}{2}\nu^{0}\overline{\partial_{0} g_{00}}
\; , \\
\overline{\Gamma}^{1}_{01}
&=&
\frac{1}{2}\left(
 \nu^{0}\partial_{1} \overline{g}_{00}
+ \nu^{0}\nu^{A} ( \partial_{A}\nu_{0} - \partial_{1} \nu_{A} - \overline{\partial_{0} g_{1A}} )
+ \overline{g}^{11} \overline{\partial_{0} g_{11}}
\right)
 \; ,
 \phantom{xxx}
\\
\overline{\Gamma}^{1}_{11}
&=&
\nu^{0}\partial_{1} \nu_{0}
- \frac{1}{2}\nu^{0}\overline{\partial_{0} g_{11}}
\; , \\
\overline{\Gamma}^{0}_{A0}
&=&
\frac{1}{2}\nu^{0} ( \partial_{A}\nu_{0} + \overline{\partial_{0} g_{1A}} - \partial_{1} \nu_{A})
\; ,  \\
\overline{\Gamma}^{0}_{A1}
&=&
0
\; , \\
\overline{\Gamma}^{1}_{A0}
&=&
\frac{1}{2} \nu^{0} (\partial_{A}\overline{g}_{00}
- \nu^{B} (\tilde{\nabla}_{A}\nu_{B}
- \tilde{\nabla}_{B}\nu_{A} + \overline{\partial_{0} g_{AB}}))
 \nonumber
\\
 &&
+ \frac{1}{2} \overline{g}^{11} (\partial_{A}\nu_{0} + \overline{\partial_{0} g_{1A}} - \partial_{1} \nu_{A})
\; , \\
\overline{\Gamma}^{1}_{A1}
&=&
\frac{1}{2} \nu^{0} ( \partial_{A}\nu_{0} - \overline{\partial_{0} g_{1A}} + \partial_{1} \nu_{A}
- \nu^{B} \partial_{1}  \overline{g}_{AB} )
\; , \\
\overline{\Gamma}^{0}_{AB}
&=&
- \frac{1}{2} \nu^{0} \partial_{1} \overline{g}_{AB}
\; , \\
\overline{\Gamma}^{1}_{AB}
&=&
\frac{1}{2}\nu^{0} (\tilde{\nabla}_{A}\nu_{B}+\tilde{\nabla}_{B}\nu_{A}-\overline{\partial_{0} g_{AB}})
- \frac{1}{2}\overline{g}^{11} \partial_{1} \overline{g}_{AB}
\; , \\
\overline{\Gamma}^{C}_{00}
&=&
- \frac{1}{2} \overline{g}^{CA}\partial_{A}\overline{g}_{00}
+ \frac{1}{2} \nu^{0}\nu^{C} \partial_{1} \overline{g}_{00}
+ \overline{g}^{CA} \overline{\partial_{0} g_{0A}}
- \nu^{0}\nu^C \overline{\partial_{0} g_{01}}
\; , \\
\overline{\Gamma}^{C}_{01}
&=&
  \frac{1}{2} \overline{g}^{CA} ( \overline{\partial_{0} g_{1A}} + \partial_{1} \nu_{A} - \partial_{A}\nu_{0} )
- \frac{1}{2} \nu^{0}\nu^{C} \overline{\partial_{0} g_{11}}
\; , \\
\overline{\Gamma}^{C}_{11}
&=&
0
\; , \\
\overline{\Gamma}^{C}_{A0}
&=&
- \frac{1}{2} \nu^{0} \nu^{C} (\overline{\partial_{0} g_{1A}} + \partial_{A}\nu_{0} - \partial_{1} \nu_{A} )
\nonumber
\\
&&
+ \frac{1}{2} \overline{g}^{BC}(\tilde{\nabla}_{A}\nu_{B} - \tilde{\nabla}_{B}\nu_{A} + \overline{\partial_{0} g_{AB}} )
\; , \\
\overline{\Gamma}^{C}_{A1}
&=&
\frac{1}{2} \overline{g}^{BC} \partial_{1} \overline{g}_{AB}
\; , \\
\overline{\Gamma}^{C}_{AB}
&=&
\tilde{\Gamma}^C_{AB}
+ \frac{1}{2} \nu^{0}\nu^{C} \partial_{1}  \overline{g}_{AB}
\; .
\end{eqnarray}
The remaining ones are obtainable by symmetry.
Note that in spite of having
$\overline{g}_{AB}=\tilde{g}_{AB}$, the Christoffel symbols
$\overline{\Gamma}^{C}_{AB}$ (a part of
$\overline{\Gamma}^{\lambda}_{\mu\nu}$) and
$\tilde{\Gamma}^{C}_{AB}$ (the Christoffel symbols of
$\tilde{g}_{AB}$) do not coincide in general.

We note the following traces of the Christoffel symbols:
\begin{eqnarray}
\overline{\Gamma}^{\mu}_{0\mu}
&=&
\nu^{0} \overline{\partial_{0} g_{01}}
+\frac{1}{2}\overline{g}^{11} \overline{\partial_{0} g_{11}}
-\nu^{0}\nu^{A} \overline{\partial_{0} g_{1A}}
+\frac{1}{2} \overline{g}^{AB}\overline{\partial_{0} g_{AB}}
\; , \\
\overline{\Gamma}^{\mu}_{1\mu}
&=&
\nu^{0}\partial_{1} \nu_{0}
+\frac{1}{2} \overline{g}^{AB}\partial_{1} \overline{g}_{AB}
\; , \\
\overline{\Gamma}^{\mu}_{A\mu}
&=&
\nu^{0}\partial_{A}\nu_{0}
+\frac{1}{2}\overline{g}^{BC}\partial_{A}\overline{g}_{BC}
\; .
\end{eqnarray}

The harmonicity vector on the null surface reads:
\begin{eqnarray}
\overline{\Gamma}^{0}
&=&
 (\nu^{0})^{2} \overline{\partial_{0} g_{11}}
-\frac{1}{2} \nu^{0} \overline{g}^{AB} \partial_{1}  \overline{g}_{AB}
\; , \\
\overline{\Gamma}^{1}
&=&
  \nu^{0}  \overline{g}^{AB}\tilde{\nabla}_B \nu_{A}
+\overline{g}^{11} \nu^{0}\partial_{1} \nu_{0}
- \frac{1}{2}\overline{g}^{11} \overline{g}^{AB}
 \partial_{1} \overline{g}_{AB}
+ (\nu^{0})^{2}\nu^{A}\nu^{B}
 \partial_{1} \overline{g}_{AB}
 \nonumber
\\ &&
+ (\nu^{0})^{2} \partial_{1}  \overline{g}_{00}
- 2 (\nu^{0})^{2} \nu^{A} \partial_{1}  \nu_{A}
- \frac{1}{2} \nu^{0} \overline{g}^{AB} \overline{\partial_{0} g_{AB}}
+ \frac{1}{2} \nu^{0} \overline{g}^{11} \overline{\partial_{0} g_{11}}
\\
&=&
  \nu^{0}  \overline{g}^{AB}\tilde{\nabla}_ B \nu_{A}
- \frac{\partial_{1} (\nu_{0}\overline{g}^{11}\sqrt{\det\tilde g_\Sigma})}{\nu_{0}\sqrt{\det\tilde g_\Sigma}}
- \frac{1}{2} \nu^{0} \overline{g}^{AB} \overline{\partial_{0} g_{AB}}
+ \frac{1}{2} \nu^{0} \overline{g}^{11} \overline{\partial_{0} g_{11}}
\; , \phantom{xxxxxx} \\
\overline{\Gamma}^{A}
&=&
- \overline{g}^{AB} \nu^{0}\partial_{B}\nu_{0}
+ \overline{g}^{CD}\tilde{\Gamma}^{A}_{CD}
+ \frac{1}{2} \overline{g}^{BC} \nu^{0} \nu^{A} \partial_{1}  \overline{g}_{BC}
- \overline{g}^{AC} \nu^{0} \nu^{B} \partial_{1}  \overline{g}_{BC}
\nonumber \\ &&
+ \nu^{0}
 ( \overline{g}^{AB} \partial_{1} \nu_{B} + \overline{g}^{AB} \overline{\partial_{0} g_{1B}} - \nu^{0} \nu^{A} \overline{\partial_{0} g_{11}} )
\; , \\
\overline{g}_{0\mu}\overline{\Gamma}^{\mu}
&=&
- \nu^{A} \nu^{0} \partial_{A}\nu_{0}
+ \overline{g}^{AB}\partial_{B}\nu_{A}
+ \nu^{0} \partial_{1} \overline{g}_{00}
+ \overline{g}^{11} \partial_{1} \nu_{0}
- \nu^{0} \nu^{A}\partial_{1} \nu_{A}
\nonumber \\ &&
-\frac{1}{2} \overline{g}^{AB}\overline{\partial_{0} g_{AB}}
+ \nu^{0} \nu^{A} \overline{\partial_{0} g_{1A}}
-\frac{1}{2} \overline{g}^{11} \overline{\partial_{0} g_{11}}
\; , \\
\overline{g}_{1\mu}\overline{\Gamma}^{\mu}
&=&
- \frac{1}{2} \overline{g}^{AB} \partial_{1} \overline{g}_{AB}
+ \nu^{0} \overline{\partial_{0} g_{11}}
\; , \\
\overline{g}_{A\mu}\overline{\Gamma}^{\mu}
&=&
- \nu^{0}( \partial_{A}\nu_{0} - \partial_{1} \nu_{A}
- \overline{\partial_{0} g_{1A}}
+ \nu^{B} \partial_{1} \overline{g}_{AB} )
+ \overline{g}^{BC}\overline{g}_{AD}\tilde{\Gamma}^D_{BC}
\; .
\end{eqnarray}
(In the main body of the paper we also use
$\overline{\Gamma}_A := \overline{g}_{AB} \overline{\Gamma}^B$,
see \eqref{GammaA}.)

The following formulae are often used in our calculations:
\begin{equation}
\overline{\partial_{0} g^{00}} \equiv
 -(\nu^{0})^{2}\overline{\partial_{0}g_{11}}
\;,\quad
\overline{\partial_{0}g^{0B}}\equiv
 -\nu^{0}(-\nu^{0}\nu^{B}\overline{\partial_{0} g_{11}}
 +\overline{g}^{BC}\overline{\partial_{0}g_{1C}})
\;,
\label{null9}
\end{equation}
\begin{equation}
\overline{\partial_{0} g^{10}}\equiv
 -\{ (\nu^{0})^{2}\overline{\partial_{0} g_{01}}
     +\nu^{0}\overline{g}^{11}\overline{\partial_{0} g_{11}}
     -(\nu^{0})^{2}\nu^{C}\overline{\partial_{0} g_{1C}})
  \}
\;.
\label{null10}
\end{equation}

The scalar wave operator acting on a function $f$ reads
\begin{eqnarray}
\overline{\Box_g f} &=&
\frac{1}{\sqrt{|\det \overline{g}|}}
    \overline{\partial_{\mu}(\sqrt{|\det g|} g^{\mu\nu}\partial_{\nu}f)}
\nonumber
\\
&=&
-\overline{\Gamma}^{\nu}\overline{\partial_{\nu}f}
+ \overline{g}^{11}\partial_{1}^{2} \overline{f}
- 2\nu^{0}\nu^{A}\partial_{1}\partial_{A} \overline{f}
+ 2\nu^{0}\partial_{1}\overline{\partial_{0} f}
\nonumber
\\
 & &
+ \overline{g}^{AB}\partial_{A}\partial_{B} \overline{f}
\; .
\end{eqnarray}

The tensor computations in this article have
been checked with the computer algebra framework {\em xAct}
\cite{xAct}.

\bigskip
\noindent{\textsc{Acknowledgements:}  {PTC and YCB are grateful
to the Mittag-Leffler Institute, Djursholm, Sweden, for
hospitality and financial support during part of work on this
paper. They acknowledge useful discussions with Vincent
Moncrief, as well as comments from Roger Tagn\'{e} Wafo. YCB
wishes to thank Thibault Damour for making available his
detailed manuscript calculations   in the case $n=3$ leading to
equations (22) of~\cite{DamourSchmidt}. JMM thanks OxPDE for
hospitality. He was supported by the French ANR grant
BLAN07-1\_201699 entitled ``LISA Science'', and also in part by
the Spanish MICINN project FIS2009-11893. PTC was supported in
part  by the EC project KRAGEOMP-MTKD-CT-2006-042360, by the
Polish Ministry of Science and Higher Education grant Nr N N201
372736, and by the EPSRC Science and Innovation award to the
Oxford Centre for Nonlinear PDE (EP/E035027/1).}}

\bibliographystyle{amsplain}
\bibliography{./char}

\def\polhk#1{\setbox0=\hbox{#1}{\ooalign{\hidewidth
  \lower1.5ex\hbox{`}\hidewidth\crcr\unhbox0}}}
  \def\polhk#1{\setbox0=\hbox{#1}{\ooalign{\hidewidth
  \lower1.5ex\hbox{`}\hidewidth\crcr\unhbox0}}} \def\cprime{$'$}
  \def\cprime{$'$} \def\cprime{$'$} \def\cprime{$'$}
\providecommand{\bysame}{\leavevmode\hbox to3em{\hrulefill}\thinspace}
\providecommand{\MR}{\relax\ifhmode\unskip\space\fi MR }
\providecommand{\MRhref}[2]{%
  \href{http://www.ams.org/mathscinet-getitem?mr=#1}{#2}
}
\providecommand{\href}[2]{#2}
\begin{thebibliography}{10}

\bibitem{CaciottaNicoloI}
G.~Caciotta and F.~Nicol{\`o}, \emph{Global characteristic problem for
  {E}instein vacuum equations with small initial data. {I}. {T}he initial data
  constraints}, J. Hyperbolic Differ. Equ. \textbf{2} (2005), no.~1, 201--277,
  arXiv:gr-qc/0409028. \MR{MR2134959 (2006i:58042)}

\bibitem{CaciottaNicoloII}
\bysame, \emph{{Global characteristic problem for the Einstein vacuum equations
  with small initial data, (II): The existence proof}},  (2006), arXiv:
  gr-qc/0608038.

\bibitem{CagnacEinsteinCRAS1}
F.~Cagnac, \emph{Probl\`eme de {C}auchy sur les hypersurfaces
  caract\'eristiques des \'equations d'{E}instein du vide}, C. R. Acad. Sci.
  Paris S\'er. A-B \textbf{262} (1966), A1488--A1491. \MR{MR0198931 (33
  \#7081)}

\bibitem{CagnacEinsteinCRAS2}
\bysame, \emph{Probl\`eme de {C}auchy sur les hypersurfaces caract\'eristiques
  des \'equations d'{E}instein du vide}, C. R. Acad. Sci. Paris S\'er. A-B
  \textbf{262} (1966), A1356--A1359. \MR{MR0198930 (33 \#7080)}

\bibitem{Cagnac73join}
\bysame, \emph{Applications du probl\`eme de {C}auchy caract\'eristique}, C. R.
  Acad. Sci. Paris S\'er. A-B \textbf{276} (1973), A195--A198, and
  \emph{Probl\`eme de {C}auchy caract\'eristique pour certains syst\`emes}, C.
  R. Acad. Sci. Paris S\'er. A-B \textbf{276} (1973), A133--A136. \MR{MR0320541
  (47 \#9078)}. \MR{MR0320542 (47 \#9079)}

\bibitem{Cagnac1980}
\bysame, \emph{Probl\`eme de {C}auchy sur un cono\"\i de caract\'eristique},
  Ann. Fac. Sci. Toulouse Math. (5) \textbf{2} (1980), no.~1, 11--19.
  \MR{MR583901 (81m:35083)}

\bibitem{Cagnac1981}
\bysame, \emph{Probl\`eme de {C}auchy sur un cono\"\i de caract\'eristique pour
  des \'equations quasi-lin\'eaires}, Ann. Mat. Pura Appl. (4) \textbf{129}
  (1981), 13--41. \MR{MR648323 (84a:35185)}

\bibitem{YvonneCIVP}
Y.~Choquet-Bruhat, \emph{Probl\`eme des conditions initiales sur un cono\"ide
  caract\'eristique}, C. R. Acad. Sci. Paris \textbf{256} (1963), 3971--3973.

\bibitem{YvonneBook}
\bysame, \emph{General relativity and the {E}instein equations}, Oxford
  Mathematical Monographs, Oxford University Press, Oxford, 2009.
  \MR{MR2473363}

\bibitem{CCL}
Y.~Choquet-Bruhat, P.T. Chru\'{s}ciel, and J.~Loizelet, \emph{{Global solutions
  of the Einstein--Maxwell equations in higher dimension}}, Class.\ Quantum
  Grav. (2006), 7383--7394, arXiv:gr-qc/0608108.

\bibitem{CCG}
Y.~Choquet-Bruhat, P.T. Chru\'{s}ciel, and J.M. Mart\'in-Garc\'ia, \emph{{The
  light-cone theorem}}, Class.\ Quantum Grav. \textbf{26} (2009), 135011 (22
  pp), arXiv:0905.2133 [gr-qc].

\bibitem{CCM3}
\bysame, \emph{{An existence theorem for the Cauchy problem on a characteristic
  cone for the Einstein equations}}, Cont.\ Math. (2010), in press, Proceedings
  of ``Complex Analysis \& Dynamical Systems IV", Nahariya, May 2009.
  arXiv:1006.xxx [gr-qc].

\bibitem{ChBdWMII}
Y.~Choquet-Bruhat and C.~DeWitt-Morette, \emph{Analysis, manifolds and physics.
  {P}art {II}}, North-Holland Publishing Co., Amsterdam, 1989, 92 applications.
  \MR{MR1016603 (91e:58001)}

\bibitem{ChrBHF}
D.~Christodoulou, \emph{The formation of black holes in general relativity},
  EMS Monographs in Mathematics, European Math.\ Soc., 2008.

\bibitem{ChristodoulouMzH}
D.~Christodoulou and H.~M{\"u}ller zum Hagen, \emph{Probl\`eme de valeur
  initiale caract\'eristique pour des syst\`emes quasi lin\'eaires du second
  ordre}, C. R. Acad. Sci. Paris S\'er. I Math. \textbf{293} (1981), 39--42.
  \MR{MR633558 (82i:35118)}

\bibitem{CJK}
P.T. Chru\'{s}ciel, J.~Jezierski, and J.~Kijowski, \emph{{H}amiltonian field
  theory in the radiating regime}, Lect. Notes in Physics, vol. m70, Springer,
  Berlin, Heidelberg, New York, 2001, URL
  \url{http://www.phys.univ-tours.fr/~piotr/papers/hamiltonian_structure}.
  \MR{MR1903925 (2003f:83040)}

\bibitem{DamourSchmidt}
T.~Damour and B.~Schmidt, \emph{Reliability of perturbation theory in general
  relativity}, Jour.\ Math.\ Phys. \textbf{31} (1990), 2441--2453.
  \MR{MR1072957 (91m:83007)}

\bibitem{Dautcourt}
G.~Dautcourt, \emph{Zum charakteristischen {A}nfangswertproblem der
  {E}insteinschen {F}eldgleichungen}, Ann. Physik (7) \textbf{12} (1963),
  302--324. \MR{MR0165949 (29 \#3229)}

\bibitem{Dossa97}
M.~Dossa, \emph{Espaces de {S}obolev non isotropes, \`a poids et probl\`emes de
  {C}auchy quasi-lin\'eaires sur un cono\"\i de caract\'eristique}, Ann. Inst.
  H. Poincar\'e Phys. Th\'eor. (1997), no.~1, 37--107. \MR{MR1434115
  (98b:35117)}

\bibitem{DossaAHP}
\bysame, \emph{Probl\`emes de {C}auchy sur un cono\"\i de caract\'eristique
  pour les \'equations d'{E}instein (conformes) du vide et pour les \'equations
  de {Y}ang-{M}ills-{H}iggs}, Ann.\ Henri Poincar\'e \textbf{4} (2003),
  385--411. \MR{MR1985778 (2004h:58041)}

\bibitem{ChBActa}
Y.~Four{\`e}s-Bruhat, \emph{Th\'eor\`eme d'existence pour certains syst\`emes
  d'\'equations aux d\'eriv\'ees partielles non lin\'eaires}, Acta Math.
  \textbf{88} (1952), 141--225.

\bibitem{Friedlander}
F.G. Friedlander, \emph{The wave equation on a curved space-time}, Cambridge
  University Press, Cambridge, 1975, Cambridge Monographs on Mathematical
  Physics, No. 2. \MR{MR0460898 (57 \#889)}

\bibitem{F2}
H.~Friedrich, \emph{The asymptotic characteristic initial value problem for
  {E}instein's vacuum field equations as an initial value problem for a
  first-order quasilinear symmetric hyperbolic system}, Proc.\ Roy.\ Soc.\
  London Ser.\ A \textbf{378} (1981), 401--421. \MR{MR637872 (83a:83007)}

\bibitem{F1}
\bysame, \emph{On the regular and the asymptotic characteristic initial value
  problem for {E}instein's vacuum field equations}, Proc.\ Roy.\ Soc.\ London
  Ser.\ A \textbf{375} (1981), 169--184. \MR{MR618984 (82k:83002)}

\bibitem{FriedrichCMP}
\bysame, \emph{On the hyperbolicity of {E}instein's and other gauge field
  equations}, Commun.\ Math.\ Phys. \textbf{100} (1985), 525--543. \MR{MR806251
  (86m:83009)}

\bibitem{Friedrich:86}
\bysame, \emph{Existence and structure of past asymptotically simple solutions
 of Einstein's field equations with positive cosmological constant},
 Jour.\ Geom. Phys. \textbf{3} (1986), 101--117.

\bibitem{FriedrichCMP86}
\bysame, \emph{On purely radiative space-times}, Comm. Math. Phys.
  \textbf{103} (1986), 35--65. \MR{MR826857 (87e:83029)}

\bibitem{galloway-nullsplitting}
G.J. Galloway, \emph{Maximum principles for null hypersurfaces and null
  splitting theorems}, Ann.\ H.\ Poincar\'e \textbf{1} (2000), 543--567.
  \MR{MR1777311 (2002b:53052)}

\bibitem{GourgoulhonJaramillo}
E.~Gourgoulhon and J.L. Jaramillo, \emph{A {$3+1$} perspective on null
  hypersurfaces and isolated horizons}, Phys.\ Rep. \textbf{423} (2006),
  159--294. \MR{MR2195374 (2007f:83055)}

\bibitem{ConstraintDamping}
C.~Gundlach, G.~Calabrese, I.~Hinder, and J.M. Mart{\'{\i}}n-Garc{\'{\i}}a,
  \emph{Constraint damping in the {Z}4 formulation and harmonic gauge},
  Classical Quantum Gravity \textbf{22} (2005), 3767--3773. \MR{MR2168553
  (2006d:83012)}

\bibitem{BondiEtal62}
M.~G. Jour.\ van der~Berg H.~Bondi and A.~W.~K. Metzner, \emph{Gravitational
  waves in general relativity {VII}}, Proc. Roy. Soc. Lond. \textbf{A269}
  (1962), 21--51.

\bibitem{HaywardNullSurfaceEquations}
S.A. Hayward, \emph{The general solution to the {E}instein equations on a null
  surface}, Classical Quantum Gravity \textbf{10} (1993), 773--778.
  \MR{MR1214441 (94e:83004)}

\bibitem{HormanderCIVP}
Lars H{\"o}rmander, \emph{A remark on the characteristic {C}auchy problem},
  Jour.\ Funct.\ Anal. \textbf{93} (1990), no.~2, 270--277. \MR{MR1073287
  (91m:58154)}

\bibitem{JKCPRD}
J.~Jezierski, J.~Kijowski, and E.~Czuchry, \emph{Dynamics of a self-gravitating
  lightlike matter shell: a gauge-invariant {L}agrangian and {H}amiltonian
  description}, Phys.\ Rev.\ D (3) \textbf{65} (2002), 064036, 20.
  \MR{MR1918464 (2003f:83063)}

\bibitem{RodnianskiKlainerman:scarred}
S.~Klainerman and I.~Rodnianski, \emph{{On emerging scarred surfaces for the
  Einstein vacuum equations}},  (2010), arXiv:1002.2656 [gr-qc].

\bibitem{Lee:Rm}
J.M. Lee, \emph{Riemannian manifolds}, Graduate Texts in Mathematics, vol. 176,
  Springer-Verlag, New York, 1997. \MR{MR1468735 (98d:53001)}

\bibitem{Leray}
J.~Leray, \emph{Hyperbolic differential equations}, mimeographed notes, 1953,
  Princeton.

\bibitem{LindbladRodnianski2}
H.~Lindblad and I.~Rodnianski, \emph{The global stability of the {Minkowski}
  space-time in harmonic gauge},  (2004), arXiv:math.ap/0411109.

\bibitem{Loizelet:AFT}
J.~Loizelet, \emph{Solutions globales d'\'equations {Einstein Maxwell}}, Ann.\
  Fac.\ Sci.\ Toulouse (2008), in press.

\bibitem{xAct}
J.M. Mart\'{\i}n-Garc\'{\i}a, \emph{{xAct: Efficient Tensor Computer Algebra}},
  \url{http://metric.iem.csic.es/Martin-Garcia/xAct}.

\bibitem{VinceJimcompactCauchyCMP}
V.~Moncrief and J.~Isenberg, \emph{Symmetries of cosmological {C}auchy
  horizons}, Commun.\ Math.\ Phys. \textbf{89} (1983), 387--413.

\bibitem{NicolasCIVPCRAS}
J.-P. Nicolas, \emph{On {L}ars {H}\"ormander's remark on the characteristic
  {C}auchy problem}, C. R. Math. Acad. Sci. Paris \textbf{344} (2007), no.~10,
  621--626. \MR{MR2334072 (2008c:35165)}

\bibitem{PenroseCIVP}
R.~Penrose, \emph{Null hypersurface initial data for classical fields of
  arbitrary spin and for general relativity}, Gen.\ Rel.\ Grav. \textbf{12}
  (1980), 225--264. \MR{MR574333 (81d:83044)}

\bibitem{Pretorius}
F.~Pretorius, \emph{Evolution of binary black hole space-times}, Phys.\ Rev.\
  Lett. \textbf{95} (2005), 121101, arXiv:gr-qc/0507014.

\bibitem{ReitererTrubowitz}
M.~Reiterer and E.~Trubowitz, \emph{Strongly focused gravitational waves},
  (2009), arXiv:0906.3812 [gr-qc].

\bibitem{RendallCIVP}
A.D. Rendall, \emph{Reduction of the characteristic initial value problem to
  the {C}auchy problem and its applications to the {E}instein equations},
  Proc.\ Roy.\ Soc.\ London A \textbf{427} (1990), 221--239. \MR{MR1032984
  (91a:83004)}

\bibitem{RendallCIVP2}
\bysame, \emph{The characteristic initial value problem for the {E}instein
  equations}, Nonlinear hyperbolic equations and field theory ({L}ake {C}omo,
  1990), Pitman Res. Notes Math. Ser., vol. 253, Longman Sci. Tech., Harlow,
  1992, pp.~154--163. \MR{MR1175208 (93j:83010)}

\bibitem{SachsCIVP}
R.K. Sachs, \emph{On the characteristic initial value problem in gravitational
  theory}, Jour.\ Math.\ Phys. \textbf{3} (1962), 908--914.

\bibitem{Thomas}
T.Y. Thomas, \emph{The differential invariants of generalized spaces},
  Cambridge University Press, 1934.

\bibitem{MzHSeifertCIVP}
H.~M{\"u}ller zum Hagen and H.-J. Seifert, \emph{On characteristic
  initial-value and mixed problems}, Gen.\ Rel.\ Grav. \textbf{8} (1977),
  259--301. \MR{MR0606056 (58 \#29307)}

\end{thebibliography}
\end{document}